\documentclass[unsortedaddress,superscriptaddress,pra,notitlepage]{revtex4}

\usepackage{amsmath,amssymb}
\usepackage{braket}
\usepackage{theorem}
\usepackage{color}
\usepackage{amsfonts}
\usepackage[mathscr]{eucal}
\usepackage[hidelinks]{hyperref}
\usepackage{graphicx}
\usepackage{xcolor}
\usepackage{tikz}
\usepackage{mathtools}

\usetikzlibrary{calc,decorations.pathreplacing}
\usetikzlibrary{arrows,shapes}

\newtheorem{definition}{Definition}[section]
\newtheorem{theorem}[definition]{Theorem}
\newtheorem{prop}[definition]{Proposition}
\newtheorem{lemma}[definition]{Lemma}

\newtheorem{rem}[definition]{Remark}

\newtheorem{cor}[definition]{Corollary}

\newtheorem{example}[definition]{Example}

\newenvironment{proof}[1][Proof]{\begin{trivlist}
\item[\hskip \labelsep {\bfseries #1}]}{\hfill$\Box$\end{trivlist}}

\newcommand{\nn}{\nonumber \\}

\def\B{{\mathcal B}}
\def\C{{\mathcal C}}
\def\D{{\mathcal D}}

\def\M{\mathcal{M}}
\def\N{\mathcal{N}}

\def\S{{\mathcal S}}

\def\X{{\mathcal X}}
\def\Y{{\mathcal Y}}

\def\bN{\mathbb{N}}
\def\bC{\mathbb{C}}
\def\bR{\mathbb{R}}

\def\theta{\vartheta}
\def\ep{\varepsilon}
\def\rho{\varrho}
\def\hil{{\mathcal H}}
\def\kil{{\mathcal K}}

\def\half{\frac{1}{2}}
\def\iff{\Longleftrightarrow}
\def\imp{\Longrightarrow}

\def\bz{\left(}
\def\jz{\right)}

\def\inv{^{-1}}

\def\egy{\mathbf 1}
\def\map{\Phi}

\def\what{\widehat}
\def\oll{\overline}

\def\povm{\mathrm{POVM}}

\def\nn{\nonumber}
\def\CP{\mathrm{CP}}

\def\cptp{\mathrm{CPTP}}

\def\nw{^{*}}

\def\meas{\mathrm{meas}}
\def\sc{\mathrm{sc}}

\def\p{_{\ge 0}}
\def\pne{_{\gneq 0}}
\def\pp{_{>0}}
\def\sa{\mathrm{sa}}

\def\cple{\le_{\mathrm{CP}}}
\def\valt{\cdot}
\def\cl{\mathrm{cl}}
\def\d{\mathrm{d}}
\def\reg{\mathrm{reg}}
\def\cpso{\text{CPSO} }
\def\cpsos{\text{CPSOs} }

\def\cH{\mathcal{H}}
\def\<{\langle}
\def\>{\rangle}

\def\eps{\varepsilon}

\def\Um{\mathrm{Um}}
\def\<{\langle}
\def\>{\rangle}
\def\cl{\mathrm{cl}}
\def\test{\mathrm{test}}
\def\pwr{\kappa}

\newcommand{\ki}[1]{\textit{\textit{#1}}}

\newcommand{\s}{\mbox{ }}
\newcommand{\ds}{\mbox{ }\mbox{ }}
\newcommand{\norm}[1]{\left\| #1\right\|}
\newcommand{\snorm}[1]{\| #1\|}
\newcommand{\vnorm}[1]{\left\| #1\right\|_{\infty}}
\newcommand{\inner}[2]{\left\langle #1 , #2\right\rangle}

\newcommand{\abs}[1]{\left| #1 \right|}

\newcommand{\diad}[2]{\left|#1\right\rangle\!\left\langle #2\right|}

\newcommand{\pr}[1]{\diad{#1}{#1}}

\newcommand{\vertcont}{\rotatebox{90}{$\,\supseteq$}}

\makeatletter
\renewcommand{\p@enumii}{}
\makeatother

\DeclareMathOperator{\id}{id}
\DeclareMathOperator{\Tr}{Tr}

\DeclareMathOperator{\ran}{ran}

\DeclareMathOperator{\logn}{\widehat\log}

\DeclareMathOperator{\divv}{\Delta}

\DeclareMathOperator{\DU}{\mathit{D}^{Um}}

\DeclareMathOperator{\conv}{conv}

\begin{document}

\pagestyle{plain}

\title{Some continuity properties of quantum R\'enyi divergences}

\author{Mil\'an Mosonyi}
\email{milan.mosonyi@gmail.com}

\affiliation{
MTA-BME Lend\"ulet Quantum Information Theory Research Group, 
Budapest University of Technology and Economics,
M\H uegyetem rkp. 3., H-1111 Budapest, Hungary
}

\affiliation{
Department of Analysis and Operations Research,
Institute of Mathematics\\
Budapest University of technology and Economics
M\H uegyetem rkp.~3., H-1111 Budapest, Hungary.
}

\author{Fumio Hiai}
\email{hiai.fumio@gmail.com}

\affiliation{
Graduate School of Information Sciences, Tohoku University, \\
Aoba-ku, Sendai 980-8579, Japan
}

\begin{abstract}
\centerline{\textbf{Abstract}}
\vspace{.3cm}

In the problem of binary quantum channel discrimination with product inputs, the supremum of all type II error exponents
for which the optimal type I errors go to zero is equal to the Umegaki channel relative entropy, while the
infimum of all type II error exponents
for which the optimal type I errors go to one is equal to the infimum of the sandwiched channel 
R\'enyi $\alpha$-divergences over all $\alpha>1$. 
We prove the equality of these two threshold values 
(and therefore the strong converse property for this problem)
using a minimax argument
based on a newly established continuity property of the sandwiched R\'enyi divergences.
Motivated by this, we give a detailed analysis of 
the continuity properties of various other quantum (channel) R\'enyi divergences, which may be of independent interest. 
\end{abstract}

\maketitle

\section{Introduction}

In the problem of quantum binary channel discrimination, an experimenter is presented with a quantum device (a black box) with the promise that it either implements some quantum channel $\N_1$ or another channel $\N_2$, which map
the set of density operators $\S(\hil)$ on some finite-dimensional Hilbert space $\hil$ into 
the set of density operators $\S(\kil)$ on a finite-dimensional Hilbert space $\kil$.
The experimenter can use the device several times, prepare arbitrary inputs to the channels, and make arbitrary 
measurements at the outputs, to make a guess about the identity of the channel.
There are two ways the experimenter may make an erroneous guess: by identifying the channel 
as $\N_2$ when it is $\N_1$ (type I error), or the other way around (type II error). 
In the case where the experimenter can only use product inputs, which is the scenario that we consider here, the problem essentially reduces to a 
binary state discrimination task. 

In a Hoeffding-type scenario, where the experimenter's task is to optimize
the type I error exponent under the constraint that the type II error probabilities decrease exponentially 
in the number of channel uses, with an exponent at least $r$, 
the optimal type I error exponent $\d_r(\N_1\|\N_2)$ 
was determined in 
\cite{CMW} (based on  \cite{ANSzV,Hayashicq,Nagaoka}) as
\begin{align*}
\d_r(\N_1\|\N_2)=\sup_{\alpha\in(0,1)}\frac{\alpha-1}{\alpha}\left[r-D_{\alpha,1}(\N_1\|\N_2)\right],
\end{align*}
where $D_{\alpha,1}(\N_1\|\N_2):=\sup_{\rho\in\S(\hil\otimes\hil)}D_{\alpha,1}((\id\otimes\N_1)\rho\|(\id\otimes\N_2)\rho)$ is the Petz-type channel R\'enyi $\alpha$-divergence \cite{P86,CMW} (see Sections \ref{sec:Renyidiv} and \ref{sec:chdiv cont} for the precise definitions). It is possible to make both types of error probabilities vanish with an exponential speed if and only if 
$\d_r(\N_1\|\N_2)>0$ for some $r>0$, which is equivalent to 
\begin{align*}
r
<
\sup_{\alpha\in(0,1)}D_{\alpha,1}(\N_1\|\N_2)
&=
\sup_{\alpha\in(0,1)}\sup_{\rho\in\S(\hil\otimes\hil)}D_{\alpha,1}((\id\otimes\N_1)\rho\|(\id\otimes\N_2)\rho)\\
&=
\sup_{\rho\in\S(\hil\otimes\hil)}\sup_{\alpha\in(0,1)}D_{\alpha,1}((\id\otimes\N_1)\rho\|(\id\otimes\N_2)\rho)\\
&=
\sup_{\rho\in\S(\hil\otimes\hil)}\DU((\id\otimes\N_1)\rho\|(\id\otimes\N_2)\rho)\\
&=
\DU(\N_1\|\N_2),
\end{align*}
where $\DU$ is Umegaki's relative entropy \cite{Umegaki}, and
$\DU(\N_1\|\N_2)$ is the Umegaki channel relative entropy.

Thus, if the type II error exponent $r\ge \DU(\N_1\|\N_2)$ then the type I errors do not go to zero with an exponential speed. In fact, one may expect that for large enough $r$, they go to $1$ with an exponential speed, and the optimal 
type I success probability exponent is known to be of the form \cite{CMW,MO}
\begin{align*}
\sc_r(\N_1\|\N_2)=\sup_{\alpha>1}\frac{\alpha-1}{\alpha}\left[r-D_{\alpha,\alpha}(\N_1\|\N_2)\right],
\end{align*}
where $D_{\alpha,\alpha}(\N_1\|\N_2):=\sup_{\rho\in\S(\hil\otimes\hil)}D_{\alpha,\alpha}((\id\otimes\N_1)\rho\|(\id\otimes\N_2)\rho)$ is the sandwiched channel R\'enyi $\alpha$-divergence \cite{Renyi_new,WWY,CMW}.
This is strictly positive if and only if 
\begin{align}
r
>
\inf_{\alpha>1}D_{\alpha,\alpha}(\N_1\|\N_2)
&=
\inf_{\alpha>1}\sup_{\rho\in\S(\hil\otimes\hil)}D_{\alpha,\alpha}((\id\otimes\N_1)\rho\|(\id\otimes\N_2)\rho)\nn\\
&\ge
\sup_{\rho\in\S(\hil\otimes\hil)}\inf_{\alpha>1}D_{\alpha,\alpha}((\id\otimes\N_1)\rho\|(\id\otimes\N_2)\rho)
\label{minimax}\\
&=
\sup_{\rho\in\S(\hil\otimes\hil)}\DU((\id\otimes\N_1)\rho\|(\id\otimes\N_2)\rho)\nn\\
&=
\DU(\N_1\|\N_2).\nn
\end{align}
Thus, if the inequality in \eqref{minimax} was an equality then we would obtain 
the appealing picture that for all type II error rates $r<\DU(\N_1\|\N_2)$,
the type I errors go to $0$ exponentially fast, while for all 
type II error rates $r>\DU(\N_1\|\N_2)$,
the type I errors go to $1$ exponentially fast. 
In particular, this would show that the channel Umegaki relative entropy $\DU(\N_1\|\N_2)$
is both the optimal direct and the optimal strong converse rate.

In \cite[Lemma 10]{CMW}, a minimax argument was presented claiming to prove the equality in \eqref{minimax}; however, it was
noticed later \cite{BDS2021} that the proof referred to a continuity property of the sandwiched R\'enyi divergences
(lower semi-continuity in their arguments) that is not suitable for the application of the minimax argument,
while upper semi-continuity,
which would do the job, does not hold on the set of all possible pairs of input states. 
The authors of \cite{CMW} came up with 
two different approaches to remedy this problem. One is based on a quantitative 
bound on the difference between the (Petz-type or sandwiched) R\'enyi divergences and the Umegaki relative entropy \cite{TCR2009,TomamichelPhD}, which avoids the minimax argument, and appeared in \cite[Appendix B]{Dingetal2020}.
The other solution, presented in this paper, fixes the minimax argument by establishing 
a suitable continuity property of the sandwiched R\'enyi divergence on pairs of states. 
Namely, we show that for the minimax argument to work, it is sufficient to have continuity of the sandwiched R\'enyi divergences on sets of the form 
\begin{align}\label{S lambda}
\{(\rho,\sigma)\in\S(\hil)\times\S(\hil):\,\rho\le\lambda\sigma\}
\end{align}
for any finite-dimensional Hilbert space $\hil$ and any $\lambda>1$, and we prove that 
the sandwiched R\'enyi divergences indeed have this continuity property.
Motivated by this, we give a detailed analysis of the continuity properties
of more general R\'enyi $(\alpha,z)$-divergences, as well as the minimal (measured) and maximal R\'enyi divergences, and the corresponding channel R\'enyi divergences. 
It seems worth emphasizing that with the above approach, we get continuity
of the R\'enyi divergences in the parameter $\alpha$
(for channels) from continuity of the R\'enyi divergences in their arguments (for states).
Thus, as a common generalization of continuity of 
the R\'enyi divergences 
in their arguments (pairs of states or channels) and in the R\'enyi parameter $\alpha$
(as well as $z$ in the case of the $(\alpha,z)$-divergences), 
we also study joint continuity in these variables.
Moreover, in the case when the arguments are states, we 
study continuity in the arguments in a more general setting where they 
can vary over sets of the form 
\begin{align}\label{S lambda2}
\{(\rho,\sigma)\in\S(\hil)\times\S(\hil):\,\rho\le\lambda\sigma^{\pwr}\}
\end{align}
with some $\lambda,\pwr\in(0,+\infty)$, or even more generally, over
\begin{align}\label{S lambda3}
\{(\rho,\sigma)\in\S(\hil)\times\S(\hil):\,\rho\le f(\sigma)\}
\end{align}
with some function $f:\,[0,+\infty)\to[0,+\infty)$ satisfying
$\limsup_{x\searrow 0}f(x)x^{-\kappa}<+\infty$ with some 
$\kappa\in(0,+\infty)$. 

The structure of the paper is as follows. 
In Section \ref{sec:miscmath}, we collect the necessary preliminaries, and in Section 
\ref{sec:Renyidiv} we give a brief overview of the various (classical and quantum) R\'enyi divergences considered in the paper.  
In Sections \ref{sec:basic} and \ref{sec:classical cont}
we discuss the continuity properties of classical and measured R\'enyi divergences,
and show that the latter are jointly continuous in their arguments on sets of the form 
\eqref{S lambda3}, 
and in the R\'enyi parameter $\alpha$ provided that it satisfies 
$(1-\alpha)\pwr<1$.
In Section \ref{sec:discont} we show examples demonstrating that this type of continuity is not true in general; in particular, the R\'enyi $(\alpha,z)$-divergences are not continuous on 
sets of the form \eqref{S lambda2} when $\alpha>1$ and $z\in(0,(\alpha-1)\pwr]$, which includes the
Petz-type R\'enyi divergences with $\alpha\in[\pwr+1,+\infty)$.

Section \ref{sec:cont} contains our main result on continuity: we show in 
Theorem \ref{thm:main} and Corollary \ref{cor:az-continuity} that the R\'enyi $(\alpha,z)$-divergences are jointly continuous in their arguments when restricted to sets of the form 
\eqref{S lambda3}, and in the $(\alpha,z)$ parameters when the latter can vary over 
$\alpha\in(1,+\infty)$, $(1-\pwr)\alpha<1$, and 
$z\in((\alpha-1)/\pwr,+\infty)$. 
This yields, in particular, 
continuity of the sandwiched R\'enyi $\alpha$-divergences 
on \eqref{S lambda3}
for all $\alpha\in[1,+\infty)$
with $(1-\pwr)\alpha<1$, and
continuity of the Petz-type R\'enyi $\alpha$-divergences 
on \eqref{S lambda3}
for $\alpha\in[1,\pwr+1)$.
In fact, we 
prove continuity on the operator level, i.e., the continuity of 
$(\alpha,z,\rho,\sigma)\mapsto 
(\sigma^{\frac{1-\alpha}{2z}}
\rho^{\frac{\alpha}{z}}\sigma^{\frac{1-\alpha}{2z}})^z$ on sets described above.
In Section \ref{sec:path cont} we prove the continuity of 
$\alpha\mapsto D_{\alpha,z(\alpha)}(\rho\|\sigma)$ at $1$ for fixed arguments $\rho,\sigma$
under the only assumption that $\liminf_{\alpha\to 1}z(\alpha)>0$, thereby extending 
previous such results given in \cite{LinTomamichel15}.
In Section \ref{sec:Dmax bound} we characterize the $(\alpha,z)$ values for which 
$D_{\alpha,z}$ is upper bounded by $D_{\max}$; such a bound puts limitations on the 
possible violation of the continuity of $D_{\alpha,z}$ on \eqref{S lambda}.

In Section \ref{sec:chdiv cont} we extend many of our continuity 
results to channel R\'enyi divergences. In particular, we prove that 
for fixed channels, 
the measured, the 
regularized measured (equivalently, the sandwiched), and the Petz-type channel R\'enyi 
$\alpha$-divergences are all continuous at $\alpha=1$, which in the first case means 
convergence to the measured channel relative entropy as $\alpha\to 1$, and 
in the latter two cases convergence to the channel Umegaki relative entropy. 
In particular, this settles the original problem described above
around \eqref{minimax}.

In Appendix \ref{sec:maxfdiv lsc} we give a simple proof of the joint lower semi-continuity 
of Matsumoto's maximal $f$-divergences.

In Appendix \ref{sec:variational} we present a different approach to prove continuity 
results for R\'enyi $(\alpha,z)$-divergences, based on a variational formula. 
While this approach gives weaker results than Theorem \ref{thm:main}, it is nevertheless sufficient to prove continuity of the sandwiched R\'enyi $\alpha$-divergences on 
sets of the form \eqref{S lambda}, which in turn is sufficient to prove the continuity of the sandwiched channel divergences at $\alpha=1$. We include the details as the proof method
might be of independent interest.

In Appendix \ref{sec:zero Renyi limit}, we give a detailed analysis 
of the limits of $D_{\alpha,0}(\rho\|\sigma)$ as $\alpha\nearrow 1$ and 
$\alpha\searrow 1$, under some mild conditions on $\rho$ and $\sigma$.

In Appendices \ref{sec:cont from conv} and \ref{sec:discont ex} we give auxiliary results to the main body of the paper.

\section{Preliminaries}

\subsection{Miscellaneous mathematics}
\label{sec:miscmath}

In what follows, $\hil$ and $\kil$ will always denote finite-dimensional Hilbert spaces. 
For a finite-dimensional Hilbert space $\hil$, 
$\B(\hil)\p$ will denote the set of positive semi-definite (PSD) operators, 
$\B(\hil)\pne$ the set of non-zero PSD operators, 
$\B(\hil)\pp$ the set of positive definite operators, 
and 
$\S(\hil)$ the set of \ki{density operators (states)} on $\hil$.
The PSD order on $\B(\hil)$ is defined by $A\ge B$ if $A-B\in\B(\hil)\p$.
For any $c\in[0,+\infty)$, 
\begin{align*}
[0,cI]
&:=\{A\in\B(\hil)\p:\,A\le c I\}\\
&=
\B(\hil)\p\cap\{A\in\B(\hil):\,\norm{A}_{\infty}\le c\}
\end{align*}
is a compact subset of $\B(\hil)$, where 
$\norm{\valt}_{\infty}$ denotes the usual operator norm.

For a PSD operator $\sigma\in\B(\hil)\p$, we take its real powers as
$\sigma^x:=\sum_{s>0}s^x P^{\sigma}(s)$, $x\in\bR$, where $P^{\sigma}(s)$ is the spectral projection of $\sigma$ corresponding to $\{s\}\subseteq\bR$. In particular, for any $x\in\bR$,
$\sigma^{-x}\sigma^x=\sigma^0$, where the latter is the projection onto the support of $\sigma$. We will need the following simple fact; we give a proof for the readers' convenience.

\begin{lemma}\label{lemma:uniform conv}
Let $f_n:\,[0,+\infty)\to\bR$, $n\in\bN$, be a sequence of continuous functions converging to a function $f$ uniformly on every compact set, and let $A_n\in\B(\hil)\p$, $n\in\bN$, be a sequence of PSD operators converging to some $A\in\B(\hil)\p$. Then $f_n(A_n)$ converges to $f(A)$. 

In particular, the map $\B(\hil)\p\times(0,+\infty)\ni(A,r)\mapsto A^r$ is continuous. 
\end{lemma}
\begin{proof}
Let $K:=\sup_{n\in\bN}\norm{A_n}_{\infty}$. For any $\ep>0$, let $p_{\ep}$ be a polynomial such that $\norm{f-p_{\ep}}_{[0,K]}:=\max_{t\in[0,K]}|f(t)-p_{\ep}(t)|<\ep$. Then  
\begin{align*}
\norm{f_n(A_n)-f(A)}_{\infty}&\le
\underbrace{\norm{f_n(A_n)-f(A_n)}_{\infty}}_{\le\norm{f_n-f}_{[0,K]}\xrightarrow[n\to+\infty]{}0}+
\underbrace{\norm{f(A_n)-p_{\ep}(A_n)}_{\infty}}_{\le \norm{f-p_{\ep}}_{[0,K]}\le\ep}\\
&
+\underbrace{\norm{p_{\ep}(A_n)-p_{\ep}(A)}_{\infty}}_{\xrightarrow[n\to+\infty]{}0}
+\underbrace{\norm{p_{\ep}(A)-f(A)}_{\infty}}_{\le \norm{f-p_{\ep}}_{[0,K]}\le\ep},
\end{align*}
whence $\limsup_{n\to+\infty}\norm{f_n(A_n)-f(A)}_{\infty}\le 2\ep$
for any $\ep>0$,
completing the proof.
\end{proof}

For any finite set $\X$, let 
$\ell^{\infty}(\X):=\bC^{\X}$
be the
commutative $C^*$-algebra of complex-valued functions on $\X$, equipped with the maximum norm 
$\norm{f}:=\max_{x\in\X}|f(x)|$. For $n\in\bN\setminus\{0\}$ and $\X=[n]$, we will use the simpler notation
\begin{align*}
\ell^{\infty}_n:=\ell^{\infty}([n]), \ds\ds\text{where}\ds\ds
[n]:=\{1,2,\ldots,n\}.
\end{align*}
Similarly to the above, we will use the notation
\begin{align*}
\ell^{\infty}(\X)\pne:=\{f\in\ell^{\infty}(\X):\,f(x)\in[0,+\infty),\,x\in\X,\,f\not\equiv 0\}
\end{align*}
for the set of non-zero non-negative functions, and
\begin{align*}
\S(\X):=\left\{f\in\ell^{\infty}(\X)\pne:\,\sum\nolimits_{x\in\X} f(x)=1\right\}
\end{align*}
for the set of probability density functions
on $\X$.

For any two finite-dimensional Hilbert spaces $\hil,\kil$, let 
$\CP^+(\hil,\kil)$ denote the set of completely positive (CP) maps 
from 
$\B(\hil)$ to $\B(\kil)$ that map non-zero PSD operators into non-zero 
PSD operators. Similarly, let 
$\cptp(\hil,\kil)$ denote the set of  
completely positive and trace-preserving (CPTP) maps
(i.e., quantum channels)
from $\B(\hil)$ to $\B(\kil)$. 
The \ki{completely positive order} between super-operators 
$\N_1,\N_2:\,\B(\hil)\to\B(\kil)$ is defined as $\N_1\cple\N_2$ if $\N_2-\N_1$ is completely positive.
We will also consider completely positive maps of the form 
$\N:\,\ell^{\infty}(\X)\to\B(\hil)$
and $\N:\,\B(\hil)\to\ell^{\infty}(\X)$
(note that in this case positivity actually implies complete positivity).

For a finite-dimensional Hilbert space $\hil$, and a natural number $k\in\bN$, 
\begin{align*}
\povm(\hil,k):=\left\{
(M_i)_{i=1}^k\in\bz\B(\hil)\p\jz^k:\,\sum\nolimits_{i=1}^k M_i=I
\right\}
\end{align*}
denotes the set of $k$-outcome \ki{positive operator-valued measures (POVMs)} on $\hil$.
Any norm $\norm{\valt}$ on $\B(\hil)$ makes $\B(\hil)^k$ a finite-dimensional normed space
with $\norm{(X_i)_{i=1}^k}:=\max_{1\le i\le k}\norm{X_i}$, in which 
$\povm(\hil,k)$ is a compact set. For any $M\in\povm(\hil,k)$,
\begin{align*}
\M:\,X\mapsto (\Tr M_iX)_{i=1}^k,\ds\ds\ds X\in\B(\hil)
\end{align*}
is a CPTP map from $\B(\hil)$ to $\ell^{\infty}([k])$,
(i.e., a \ki{quantum--to--classical channel}),
where $[k]:=\{1,\ldots,k\}$.

Throughout the paper, $\log$ will denote a logarithm with some fixed base
that is strictly larger than $1$ (the exact value of which is irrelevant).
We will also use the extensions
\begin{align*}
\log x
:=
\begin{cases} 
-\infty,& x=0,\\
\log x,&x\in(0,+\infty),\\
+\infty,&x=+\infty,
\end{cases}
\ds\ds\ds
\logn x:=
\begin{cases} 
0,& x=0,\\
\log x,&x\in(0,+\infty),\\
+\infty,&x=+\infty.
\end{cases}
\end{align*}
In particular, $\logn \sigma$ is well-defined for any $\sigma\in\B(\hil)\p$.
 
\medskip

The following minimax theorem is from \cite[Corollary A.2]{MH}.
\begin{lemma}\label{lemma:minimax2}
Let $X$ be a compact topological space, $Y$ be an ordered set, and let $f:\,X\times Y\to \bR\cup\{\pm\infty\}$ be a function. Assume that
\smallskip

\s(i) $f(\valt,\,y)$ is upper semicontinuous for every $y\in Y$ and
\smallskip

(ii) \begin{minipage}[t]{15cm}
$f(x,\valt)$ is monotonic increasing for every $x\in X$, or
$f(x,\valt)$ is monotonic decreasing for every $x\in X$.
\end{minipage}
\smallskip

\noindent Then 
\begin{align}\label{minimax statement}
\sup_{x\in X}\inf_{y\in Y}f(x,y)=
\inf_{y\in Y}\sup_{x\in X}f(x,y),
\end{align}
and the suprema in \eqref{minimax statement} can be replaced by maxima.
\end{lemma}

\begin{lemma}\label{lemma:usc}
Let $X$ be a topological space, $Y$ be an arbitrary set, and 
$f:\,X\times Y\to\bR\cup\{\pm\infty\}$ be a function.
\begin{enumerate}
\item\label{usc1}
If $f(\valt,y)$ is upper semi-continuous for every $y\in Y$ then $\inf_{y\in Y}f(\valt,y)$ is upper semi-continuous.

\item\label{usc2}
If $Y$ is a compact topological space, and $f$ is upper semi-continuous on $X\times Y$ with respect to the product topology, then 
$\sup_{y\in Y}f(\valt,y)$ is upper semi-continuous.

\item\label{usc3}
If $Y$ is a compact topological space, and $f$ is continuous on $X\times Y$ with respect to the product topology, then 
$\sup_{y\in Y}f(\valt,y)$ and
$\inf_{y\in Y}f(\valt,y)$ 
are continuous.
\end{enumerate}
\end{lemma}
\begin{proof}
\ref{usc1} is obvious by definition. For a proof of 
\ref{usc2}, see \cite{BaryRenyi}.
In \ref{usc3}, upper semi-continuity 
of $\sup_{y\in Y}f(\valt,y)$
is obvious from \ref{usc2}, and lower semi-continuity 
follows by \ref{usc1} applied to $-f$.
This proves the continuity of $\sup_{y\in Y}f(\valt,y)$, and the continuity of 
$\inf_{y\in Y}f(\valt,y)$ follows from this by replacing $f$ with $-f$.

\end{proof}
\medskip

The following is well known:
\begin{lemma}\label{lemma:trace functions}
Let $f:\,[0,+\infty)\to\bR$ be a function.
\begin{enumerate}
\item
If $f$ is monotone increasing then 
so is $\Tr f(\valt)$ on $\B(\hil)\p$, i.e., 
\begin{align}\label{trace monotonicity}
A,B\in\B(\hil)\p,\ds A\le B\ds\imp\ds \Tr f(A)\le\Tr f(B).
\end{align}
If, moreover, $f$ is strictly increasing then 
\begin{align}\label{trace monotonicity2}
A,B\in\B(\hil)\p,\ds A\lneq B\ds\imp\ds \Tr f(A)<\Tr f(B).
\end{align}

\item
If $f$ is convex on $[0,+\infty)$ then $\Tr f(\valt)$ is convex on $\B(\hil)\p$.
\end{enumerate}
\end{lemma}

The following is again well known, but we include a proof for the readers' convenience.

\begin{lemma}\label{lemma:power inequality}
Let $A\in\B(\hil)\p$ and $P\in\B(\hil)$ be a projection. Then 
\begin{align}\label{power inequality}
\Tr (PAP)^z\le\Tr A^z,\ds\ds\ds z\in(0,+\infty),
\end{align}
and the inequality is strict if $A^0\nleq P$. 
\end{lemma}
\begin{proof}
As is well known, for any $X\in\B(\hil)$, the eigenvalues of $XX^*$ and $X^*X$ are the same, including multiplicities.
Applying this to $X:=A^{1/2}P$ yields $\Tr (PAP)^z=\Tr (A^{1/2}PA^{1/2})^z$.
The inequality $A^{1/2}PA^{1/2}\le A^{1/2}IA^{1/2}=A$ is obvious, and thus, by Lemma \ref{lemma:trace functions}, 
$\Tr (A^{1/2}PA^{1/2})^z
\le
\Tr A^z$. 
If $A^0\nleq P$ then 
$A^{1/2}PA^{1/2}\lneq A$ above, and strict inequality in 
\eqref{power inequality} follows from Lemma \ref{lemma:trace functions}.
\end{proof}

\subsection{Quantum R\'enyi divergences}
\label{sec:Renyidiv}

By a quantum divergence $\divv$ we mean a map 
\begin{align*}
\divv:\,\cup_{d\in\bN}\bz\B(\bC^d)\pne\times\B(\bC^d)\pne\jz\to\bR\cup\{+\infty\}
\end{align*}
that is invariant under isometries, i.e., 
for any $\rho,\sigma\in\B(\bC^d)\pne$ and any isometry $V:\,\bC^d\to\bC^{d'}$, 
\begin{align*}
\divv\bz V\rho V^*\|V\sigma V^*\jz=\divv(\rho\|\sigma).
\end{align*}
It is clear that any quantum divergence $\divv$ can be uniquely extended to pairs of 
non-zero PSD operators $\rho,\sigma$ on an arbitrary finite-dimensional Hilbert space $\hil$ by 
mapping $\hil$ into some $\bC^d$ with an isometry $V$, and defining 
$\divv(\rho\|\sigma):=\divv(V\rho V^*\|V\sigma V^*)$. Clearly, this extension is 
well-defined and it is also invariant under isometries. 

For any finite set $\X$ and $\rho,\sigma\in\ell^{\infty}(\X)\pne$, 
the 
\ki{classical R\'enyi $\alpha$-divergence} of $\rho$ and $\sigma$ is defined for 
$\alpha\in(0,1)\cup(1,+\infty)$ as \cite{Renyi}
\begin{align*}
D_{\alpha}^{\cl}(\rho\|\sigma)
&:=
\frac{\psi_{\alpha}^{\cl}(\rho\|\sigma)-\psi_{1}^{\cl}(\rho\|\sigma)}{\alpha-1},
\end{align*}
where
\begin{align*}
\psi_{\alpha}^{\cl}(\rho\|\sigma)&:=\log Q_{\alpha}^{\cl}(\rho\|\sigma),\\
Q_{\alpha}^{\cl}(\rho\|\sigma)
&:=\begin{cases}
\sum_{x\in\X}\rho(x)^{\alpha}\sigma(x)^{1-\alpha},&\rho\ll\sigma\text{ or }\alpha\in(0,1),\\
\sum_{x\in\X}\rho(x),&\alpha=1,\\
+\infty,&\text{otherwise},
\end{cases}
\end{align*}
and $\rho\ll\sigma$ means $\sigma(x)=0\imp\rho(x)=0$.
It is straightforward to verify that for any fixed $\rho,\sigma$, $
\alpha\mapsto D_{\alpha}^{\cl}(\rho\|\sigma)$ is monotone increasing on $(0,1)\cup(1,+\infty)$, and 
\begin{align*}
D_{1}^{\cl}(\rho\|\sigma)
&:=\lim_{\alpha\to1}
D_{\alpha}^{\cl}(\rho\|\sigma)
=
\begin{cases}
\frac{1}{\sum_{x\in\X}\rho(x)}\sum_{x\in\X}\rho(x)(\logn\rho(x)-\logn\sigma(x)),& \rho\ll\sigma,\\
+\infty,&\text{otherwise},
\end{cases}
\end{align*}
is (a normalized version of) the \ki{Kullback-Leibler divergence}, or \ki{relative entropy}, 
while
\begin{align*}
D_{\infty}^{\cl}(\rho\|\sigma)
&:=\lim_{\alpha\to+\infty}
D_{\alpha}^{\cl}(\rho\|\sigma)
=
\log\inf\{\lambda>0:\,\rho\le\lambda\sigma\}.
\end{align*}
Note that with the above notations,
\begin{align*}
\psi_{\alpha}^{\cl}(\rho\|\sigma)
=(\alpha-1)D_{\alpha}^{\cl}(\rho\|\sigma)+\log\sum\nolimits_{x\in\X}\rho(x),\ds\ds\ds
\alpha\in(0,+\infty),
\end{align*}
with the convention
\begin{align*}
0\cdot(\pm\infty):=0.
\end{align*}

We say that for a given $\alpha\in(0,+\infty]$, a quantum divergence $D_{\alpha}^q$ is a \ki{quantum R\'enyi $\alpha$-divergence}
if for any  orthonormal basis $(\ket{i})_{i=1}^d$ in some Hilbert space $\hil$, and any 
$\rho,\sigma\in[0,+\infty)^d\setminus\{0\}$,
\begin{align}\label{qRenyi def}
D_{\alpha}^q\bz\sum\nolimits_{i=1}^d\rho(i)\pr{i}\Big\|\sum\nolimits_{i=1}^d\sigma(i)\pr{i}\jz
&=
D_{\alpha}^{\cl}\bz(\rho(i))_{i=1}^d\Big\|(\sigma(i))_{i=1}^d\jz.
\end{align}
For any quantum  R\'enyi $\alpha$-divergence $D_{\alpha}^q$, and any 
$\rho,\sigma\in\B(\hil)\pne$, we define
\begin{align*}
\psi_{\alpha}^q(\rho\|\sigma)&:=(\alpha-1)D_{\alpha}^q(\rho\|\sigma)+\log\Tr\rho,\\
Q_{\alpha}^q(\rho\|\sigma)&:=\exp(\psi_{\alpha}^q(\rho\|\sigma)),
\end{align*}
for every $\alpha\in(0,+\infty)$,
so that 
\begin{align*}
D_{\alpha}^q(\rho\|\sigma)=\frac{1}{\alpha-1}\log Q_{\alpha}^q(\rho\|\sigma)
-\frac{1}{\alpha-1}\log\Tr\rho.
\end{align*}
for every $\alpha\in(0,1)\cup(1,+\infty)$.

One natural way to obtain quantum R\'enyi $\alpha$-divergences is 
by some direct extension of the classical R\'enyi $\alpha$-divergences.
Notable 
examples include the \ki{measured R\'enyi $\alpha$-divergence} $D_{\alpha}^{\meas}$, and 
the \ki{regularized measured R\'enyi $\alpha$-divergence} $\oll{D}_{\alpha}^{\meas}$,
defined for $\rho,\sigma\in\B(\hil)\pne$ as
\begin{align}
D_{\alpha}^{\meas}(\rho\|\sigma)&:=
\sup\left\{D_{\alpha}^{\cl}\bz(\Tr M_i\rho)_{i=1}^k\|(\Tr M_i\sigma)_{i=1}^k\jz:\,
(M_i)_{i=1}^k\in\povm(\hil,k),\,k\in\bN
\right\},\label{measured Renyi}\\
\oll{D}_{\alpha}^{\meas}(\rho\|\sigma)&:=\sup_{n\in\bN}\frac{1}{n}
D_{\alpha}^{\meas}(\rho^{\otimes n}\|\sigma^{\otimes n}),
\label{regularized measured Renyi}
\end{align}
respectively.
Matsumoto's \ki{maximal R\'enyi $\alpha$-divergences} \cite{Matsumoto_newfdiv} are defined
as
\begin{align*}
D_{\alpha}^{\max}(\rho\|\sigma):=\inf\left\{D_{\alpha}^{\cl}(p\|q):\,\Gamma(p)=\rho,\,\Gamma(q)=\sigma\right\},
\end{align*}
where the infimum is taken over triples $(\Gamma,p,q)$ (so-called \ki{reverse tests}), where $p,q$ 
are non-negative functions on some finite set $\X$, and $\Gamma:\,\ell^{\infty}(\X)\to\B(\hil)$ is a 
(completely) positive trace-preserving linear map.
It was shown in \cite{Matsumoto_newfdiv}
that if $\alpha\in(0,1)$, or $\alpha\in(1,2]$ and
$\rho^0\le\sigma^0$, then
\begin{align}\label{maxRenyi persp}
D_{\alpha}^{\max}(\rho\|\sigma)=\frac{1}{\alpha-1}\log
\Tr\sigma^{1/2}\bz\sigma^{-1/2}\rho\sigma^{-1/2}\jz^{\alpha}\sigma^{1/2}
-\frac{1}{\alpha-1}\log\Tr\rho
=:\what D_{\alpha}(\rho\|\sigma).
\end{align}
(See also \cite{HiaiMosonyi2017} for a detailed analysis of $\what D_{\alpha}$.)

\begin{rem}
It is easy to see from the concrete form of the optimal reverse test 
given in \cite[Section 4.2]{Matsumoto_newfdiv} that 
$D_{\alpha}^{\max}(\rho\|\sigma)\le \what D_{\alpha}(\rho\|\sigma)$
for any $\alpha\in(1,+\infty)$ when 
$\rho^0\le\sigma^0$ (see also \cite[Section 4.2.3]{TomamichelBook}). 
However, while $D_{\alpha}^{\max}$ is monotone under CPTP maps by definition, 
the same is not true for $\what D_{\alpha}$ when $\alpha>2$. 
This follows by a standard argument from the fact that 
$\what Q_{\alpha}$ is not convex in its arguments when $\alpha>2$; see 
\cite[Proposition A.1]{HiaiMosonyi2017}.
Hence, $D_{\alpha}^{\max}\ne\what D_{\alpha}$ for every $\alpha>2$.
\end{rem}

For any $\alpha\in(0,1)\cup(1,+\infty)$ and $z\in(0,+\infty)$, and any pair of non-zero
PSD operators $\rho,\sigma\in\B(\hil)\pne$, their 
\ki{R\'enyi $(\alpha,z)$-divergence} is defined as \cite{AD,JOPP}
\begin{align*}
D_{\alpha,z}(\rho\|\sigma):=\frac{1}{\alpha-1}\log Q_{\alpha,z}(\rho\|\sigma) 
-\frac{1}{\alpha-1}\log\Tr\rho,
\end{align*}
where 
\begin{align}\label{Q alpha z def}
Q_{\alpha,z}(\rho\|\sigma)
&:=
\begin{cases}
\Tr\bz\rho^{\frac{\alpha}{2z}}\sigma^{\frac{1-\alpha}{z}}\rho^{\frac{\alpha}{2z}}\jz^z,&
\rho^0\le\sigma^0\ds\text{or}\ds\alpha\in(0,1),\\
+\infty,&\text{otherwise}.
\end{cases}
\end{align}
Note that since for any $X\in\B(\hil)$, the eigenvalues of $XX^*$ and $X^*X$ are 
the same, counted with multiplicities, we have 
\begin{align*}
\Tr\bz\rho^{\frac{\alpha}{2z}}\sigma^{\frac{1-\alpha}{z}}\rho^{\frac{\alpha}{2z}}\jz^z
=
\Tr\bz\sigma^{\frac{1-\alpha}{2z}}\rho^{\frac{\alpha}{z}}\sigma^{\frac{1-\alpha}{2z}}\jz^z.
\end{align*}
For $z=+\infty$, we have 
\begin{align*}
D_{\alpha,+\infty}(\rho\|\sigma)
&:=
\lim_{z\to+\infty}D_{\alpha,z}(\rho\|\sigma)
=
\frac{1}{\alpha-1}\log Q_{\alpha,+\infty}(\rho\|\sigma)-\frac{1}{\alpha-1}\log\Tr\rho,
\end{align*}
where
\begin{align*}
Q_{\alpha,+\infty}(\rho\|\sigma):=
\begin{cases}
\Tr Pe^{\alpha P(\logn\rho)P+(1-\alpha)P(\logn\sigma)P},&\rho^0\le\sigma^0\text{ or }\alpha\in(0,1),\\
+\infty,&\text{otherwise},
\end{cases}
\end{align*}
with $P:=\rho^0\wedge\sigma^0$; see \cite{AD,HP_GT,MO-cqconv}.
It is straightforward to verify that $D_{\alpha,z}$ is a quantum 
R\'enyi $\alpha$-divergence for any $z\in(0,+\infty]$. 
For $\alpha\in(0,1)\cup(1,+\infty)$ the special case
$D_{\alpha,1}$ is called the \ki{Petz-type}, or \ki{standard} R\'enyi
$\alpha$-divergence \cite{P86}, and 
$D_{\alpha}\nw:=D_{\alpha,\alpha}$ the \ki{sandwiched R\'enyi $\alpha$-divergence}
\cite{Renyi_new,WWY}.

For $\alpha=1$ and every $z\in(0,+\infty]$ we define
\begin{align}
D_{1,z}(\rho\|\sigma):=\DU(\rho\|\sigma)
:=\begin{cases}
\frac{1}{\Tr\rho}\Tr\rho(\logn\rho-\logn\sigma),&\rho^0\le\sigma^0,\\
+\infty,&\text{otherwise},
\end{cases}\label{relent}
\end{align}
to be (a normalized version of) \ki{Umegaki's relative entropy} of $\rho$ and $\sigma$ \cite{Umegaki}.
Then
\begin{align}
D_1\nw(\rho\|\sigma)
&:=
\lim_{\alpha\to 1}D_{\alpha}\nw(\rho\|\sigma)\nn\\
&=\DU(\rho\|\sigma)\nn\\
&=
\lim_{\alpha\to 1}D_{\alpha,z}(\rho\|\sigma),\ds\ds z\in(0,+\infty],
\label{alpha=1,z}
\end{align}
where the first equality is due to \cite[Theorem 5]{Renyi_new}, the second equality 
was shown in \cite[Proposition 3]{LinTomamichel15} in the case 
$\rho^0\le\sigma^0$ and $z\in(0,+\infty)$, and in 
\cite[Lemma 3.5]{MO-cqconv} for $z=+\infty$, and we prove it in the general case in 
Proposition \ref{prop:a-z cont in a} below.
We note that \eqref{alpha=1,z} does not hold in general when $z=0$; 
see \eqref{z=0 def} for the definition of $D_{\alpha,0}$, 
Proposition \ref{prop:zero Renyi pure} for a more precise statement, and Appendix 
\ref{sec:zero Renyi limit} for a detailed analysis of the problem.
Also by \cite[Theorem 5]{Renyi_new},
\begin{align*}
D_{\infty}\nw(\rho\|\sigma):=\sup_{\alpha>0}D_{\alpha}\nw(\rho\|\sigma)=
D_{\max}(\rho\|\sigma)
&:=
\log\inf\{\lambda>0:\,\rho\le\lambda\sigma\}\nn\\
&=\begin{cases}
\log\norm{\sigma^{-1/2}\rho\sigma^{-1/2}}_{\infty},&\rho^0\le\sigma^0,\\
+\infty,&\text{otherwise},
\end{cases}
\end{align*}
is the \ki{max-relative entropy} of $\rho$ and $\sigma$ \cite{Datta,RennerPhD}.
It is known that for any $\rho,\sigma\in\B(\hil)\pne$, 
\begin{align}\label{reg measured}
\oll{D}_{\alpha}^{\meas}(\rho\|\sigma)
=
\begin{cases}
D_{\alpha}\nw(\rho\|\sigma)=D_{\alpha,\alpha}(\rho\|\sigma),&\alpha\in[1/2,1)\cup(1,+\infty),\\
\frac{\alpha}{1-\alpha}D_{1-\alpha}\nw(\sigma\|\rho)
+\frac{1}{\alpha-1}\log\frac{\Tr\rho}{\Tr\sigma}
=D_{\alpha,1-\alpha}(\rho\|\sigma),&\alpha\in(0,1/2);
\end{cases}
\end{align}
see \cite{HT14,MO,MH-testdiv,Vrana_private}.

It is straightforward to verify that all the above quantum R\'enyi divergences satisfy the scaling laws
\begin{align}\label{scaling}
D_{\alpha}^q(\lambda\rho\|\eta\sigma)
=
D_{\alpha}^q(\rho\|\sigma)+\log\lambda-\log\eta,
\end{align}
for any $\rho,\sigma\in\B(\hil)\pne$ and $\lambda,\eta>0$.

For two quantum divergences $\divv_1$ and $\divv_2$ we write
\begin{align*}
\divv_1\le \divv_2\ds\ds\text{if}\ds\ds
\divv_1(\rho\|\sigma)\le \divv_2(\rho\|\sigma)
\end{align*}
for any 
finite-dimensional Hilbert space $\hil$ and 
$\rho,\sigma\in\B(\hil)\pne$.
For a given Hilbert space $\hil$, we say that $\divv_1\le\divv_2$ on $\S(\hil)$ if 
$\divv_1(\rho\|\sigma)\le\divv_2(\rho\|\sigma)$ for every $\rho,\sigma\in\S(\hil)$.

According to the Araki-Lieb-Thirring inequality \cite{Araki,LT} and its converse given in 
\cite{Audenaert-ALT}, 
\begin{align*}
\Tr (A^rB^rA^r)^q\le\Tr (ABA)^{rq}\le \bz\Tr (A^rB^rA^r)^q\jz^r\norm{A}_{\infty}^{2rq(1-r)}
(\Tr B^{rq})^{1-r}
\end{align*}
for any $A,B\in\B(\hil)\p$, $q\in[0,+\infty)$ and $r\in[0,1]$.
Applying this to
$0<z_1\le z_2<+\infty$, $r:=z_1/z_2$, $q:=z_2$, and 
$A:=\rho^{\frac{\alpha}{2z_1}}$, $B:=\sigma^{\frac{1-\alpha}{z_1}}$, yields
\begin{align}
Q_{\alpha,z_2}(\rho\|\sigma)
=
\Tr\bz\rho^{\frac{\alpha}{2z_2}}\sigma^{\frac{1-\alpha}{z_2}}\rho^{\frac{\alpha}{2z_2}}\jz^{z_2}
&\le
\Tr\bz\rho^{\frac{\alpha}{2z_1}}\sigma^{\frac{1-\alpha}{z_1}}\rho^{\frac{\alpha}{2z_1}}\jz^{z_1}=
Q_{\alpha,z_1}(\rho\|\sigma)
\label{ALT1}\\
&\le
Q_{\alpha,z_2}(\rho\|\sigma)^{\frac{z_1}{z_2}}\norm{\rho}_{\infty}^{\alpha\bz 1-\frac{z_1}{z_2}\jz}
\bz\Tr\sigma^{1-\alpha}\jz^{\bz 1-\frac{z_1}{z_2}\jz}.\label{ALT2}
\end{align}
In particular, we have 
\begin{align}\label{ALT3}
0<z_1\le z_2\ds\imp\ds
D_{\alpha,z_1}\le D_{\alpha,z_2},
\ds\alpha\in(0,1),\ds\ds
D_{\alpha,z_1}\ge D_{\alpha,z_2},\ds\alpha>1,
\end{align}
as was already observed, e.g., in \cite[Proposition 1]{LinTomamichel15}.
As a consequence, for any $\alpha\in(0,1)\cup(1,+\infty)$ and 
$\rho,\sigma\in\B(\hil)\pne$, the limit
\begin{align}\label{z=0 def}
D_{\alpha,0}(\rho\|\sigma):=\lim_{z\searrow 0}D_{\alpha,z}(\rho\|\sigma)
\end{align}
exists; see \cite{AH2019} for a detailed analysis of this quantity.
For any $\alpha_1<1<\alpha_2$ and $z_1,z_2\in[0,+\infty]$, 
\begin{align}\label{Umegaki bound on az}
D_{\alpha_1,z_1}\le
D_{\alpha_1,+\infty}\le
\DU\le 
D_{\alpha_2,+\infty}\le
D_{\alpha_2,z_2},
\end{align} 
where the first and the last inequalities follow from \eqref{ALT3},
and the second and the third inequalities from 
\cite[Theorem 3.6]{MO-cqconv}. 
Note that the combination of \eqref{Umegaki bound on az} and the equality in 
\eqref{alpha=1,z} for some fixed $z_0\in(0,+\infty)$ yields 
the equality in \eqref{alpha=1,z} for any $z\in[z_0,+\infty]$.

We say that a quantum divergence $\divv$ is \ki{monotone under CPTP maps} if for 
any finite-dimensional Hilbert spaces $\hil,\kil$, any 
$\N\in\cptp(\hil,\kil)$, and any
$\rho,\sigma\in\B(\hil)\pne$,
\begin{align*}
\divv(\N(\rho)\|\N(\sigma))\le\divv(\rho\|\sigma).
\end{align*}
It is clear from their definitions that $D_{\alpha}^{\meas}$, $\oll{D}_{\alpha}^{\meas}$,
and $D_{\alpha}^{\max}$
are monotone under CPTP maps for any $\alpha\in(0,+\infty)$. 
For $\alpha>1$, $D_{\alpha,z}$ is monotone under CPTP maps if and only if 
$\max\{\alpha/2,\alpha-1\}\le z\le \alpha$ \cite{Hiai_concavity2013,Zhang2018}.

\begin{rem}\label{rem:ort}
It is easy to verify from the definitions that if 
$\omega_1,\omega_2,\omega_3\in\B(\hil)\pne$ are such that 
$\omega_3\omega_1=0=\omega_3\omega_2$ then 
\begin{align}
&D_{\alpha,z}(\omega_1\|\omega_2+\omega_3)=
D_{\alpha,z}(\omega_1\|\omega_2),\ds\ds \alpha,z\in(0,+\infty),\label{ort1}\\
& D_{\alpha}^{\max}(\omega_1\|\omega_2+\omega_3)=
D_{\alpha}^{\max}(\omega_1\|\omega_2),\ds\ds \alpha\in(0,+\infty),\label{ort2}\\
& D_{\max}(\omega_1\|\omega_2+\omega_3)=
D_{\max}(\omega_1\|\omega_2).\label{ort3}
\end{align}
We will use these in the proof of Proposition \ref{prop:a>1 discont}.
\end{rem}

\section{Continuity of quantum R\'enyi $\alpha$-divergences for PSD operators}
\label{sec:PSD}

\subsection{Basic observations}
\label{sec:basic}

The following is obvious from Lemma \ref{lemma:uniform conv} and the definition
of the R\'enyi $(\alpha,z)$-divergences:
\begin{lemma}\label{lemma:0-1 cont}
For any finite-dimensional Hilbert space $\hil$,
\begin{align*}
(0,1)\times(0,+\infty)\times\B(\hil)\pne\times\B(\hil)\pne
\ni(\alpha,z,\rho,\sigma)\mapsto 
\bz\sigma^{\frac{1-\alpha}{2z}}\rho^{\frac{\alpha}{z}}
\sigma^{\frac{1-\alpha}{2z}}\jz^z
\end{align*}
is continuous, and hence so is 
$(\alpha,z,\rho,\sigma)\mapsto D_{\alpha,z}(\rho\|\sigma)$ on the same set.
\end{lemma}

On the other hand, for $\alpha>1$, none of the (classical or quantum) R\'enyi 
$\alpha$-divergences are continuous on the whole of $\B(\hil)\pne\times\B(\hil)\pne$, for 
any Hilbert space with dimension at least $2$. 
In fact, continuity does not hold even on the smaller set 
\begin{align*}
\{(\rho,\sigma)\in\B(\hil)\pne\times\B(\hil)\pne:\,\rho^0\le\sigma^0\},
\end{align*}
and even in the simplest $2$-dimensional classical case, as the following simple example shows:

\begin{example}\label{rem:discont}
Let
$\rho_n:=(1-p_n,p_n)$ and $\sigma_n=(1-q_n,q_n)$ with $p_n:=cn^{-\beta}$,
$q_n:=dn^{-\gamma}$, where $c,d,\beta,\gamma>0$.
Then $\rho_n,\sigma_n\to (1,0)=:\rho_{\infty}=:\sigma_{\infty}$, so that 
$D_{\alpha}^{\cl}( \rho_{\infty}\|\sigma_{\infty})=0$, while 
\begin{align*}
\lim_{n\to+\infty}D_{\alpha}^{\cl}(\rho_n\|\sigma_n)=
\begin{cases}
+\infty,&\beta/\gamma<1-1/\alpha,\\
\frac{1}{\alpha-1}\log(1+c^{\alpha}d^{1-\alpha}),&\beta/\gamma=1-1/\alpha,\\
0,&\beta/\gamma>1-1/\alpha.
\end{cases}
\end{align*} 
Moreover, with different choices of $p_n,q_n$, it is also easy to construct examples where 
$D_{\alpha}^{\cl}(\rho_n\|\sigma_n)$ does not have a limit.
\end{example}

Thus, our aim is to find non-trivial subsets of $\B(\hil)\pne\times\B(\hil)\pne$ for
various quantum R\'enyi $\alpha$-divergences on which the given R\'enyi $\alpha$-divergence is continuous. 
Note that for any $\alpha,z\in(0,+\infty)$, $D_{\alpha,z}$ is continuous on $\B(\hil)\pne\times\B(\hil)\pp$ for any 
finite-dimensional Hilbert space $\hil$;
however, this is not the kind of continuity property that we need, e.g., in 
Section \ref{sec:chdiv cont}.
Instead, we are interested in continuity on sets of the form 
\begin{align}
\bz\B(\hil)\times\B(\hil)\jz_{f}&:=\{(\rho,\sigma)\in\B(\hil)\pne
\times\B(\hil)\pne:\,\rho\le f(\sigma)\},\label{B f}\\
\bz[0,cI]\times[0,cI]\jz_{f}&:=\{(\rho,\sigma)\in[0,cI]\times[0,cI]:\,
\rho,\sigma\ne 0,\,\rho\le
f(\sigma)\},\label{B f2}\\
\bz\S(\hil)\times\S(\hil)\jz_{f}&:=\{(\rho,\sigma)\in\S(\hil)
\times\S(\hil):\,\rho\le f(\sigma)\},
\label{B f3}
\end{align}
with some function $f:\,[0,+\infty)\to[0,+\infty)$ and $c\in(0,+\infty)$. 
We will use the notations
\begin{align}
\bz\B(\hil)\times\B(\hil)\jz_{\lambda,\pwr},\ds\ds
\bz[0,cI]\times[0,cI]\jz_{\lambda,\pwr},\ds\ds
\bz\S(\hil)\times\S(\hil)\jz_{\lambda,\pwr},\ds\ds
\end{align}
in the special case $f=\lambda\id_{[0,+\infty)}^{\pwr}$, and
\begin{align*}
\bz\B(\hil)\times\B(\hil)\jz_{\lambda},\ds\ds
\bz[0,cI]\times[0,cI]\jz_{\lambda},\ds\ds
\bz\S(\hil)\times\S(\hil)\jz_{\lambda}
\end{align*}
when we further have $\pwr=1$. In the commutative case, we use the notations
\begin{align*}
(\ell^{\infty}(\X)\times\ell^{\infty}(\X))_{f},\ds\ds
(\S(\X)\times\S(\X))_{f},\ds\ds\text{etc.}
\end{align*}
when $\B(\hil)$ (resp.~$\S(\hil)$) in the above definitions is replaced by 
$\ell^{\infty}(\X)$ (resp.~$\S(\X)$). 

In fact, our main concern is continuity on sets of the form 
$\bz\S(\hil)\times\S(\hil)\jz_{\lambda}$, $\lambda>1$; however, since our techniques 
are applicable to the study of continuity on more general sets of the above types,
we consider them, too.

\begin{definition}\label{def:pwr-function}
Let $\pwr>0$. 
By a \ki{$\pwr$-function} we mean a continuous function 
$f:\,[0,+\infty)\to[0,+\infty)$ such that 
\begin{align*}
\limsup_{x\searrow 0}f(x)x^{-\pwr}<+\infty.
\end{align*}
We say that $f$ is a \ki{strict $\pwr$-function} if, moreover, 
\begin{align*}
\liminf_{x\searrow 0}f(x)x^{-\pwr}>0.
\end{align*}
\end{definition}

\begin{rem}\label{rem:zero limit}
In particular, 
if $f$ is a $\pwr$-function then 
\begin{align}\label{zero limit}
\lim_{x\searrow 0}f(x)=f(0)=0=\lim_{x\searrow 0}x^{\delta-\pwr}f(x),\ds\ds\ds\delta>0.
\end{align}
\end{rem}

\begin{rem}
Functions of the following form are strict $\pwr$-functions:
\begin{align}\label{power function}
f(x)=
\begin{cases}
\lambda x^{\pwr},&x\in[0,c),\\
g(x),&x\in[c,+\infty),
\end{cases}
\end{align}
where $\lambda,c,>0$, and $g:\,[c,+\infty)\to\bR$ is a monotone increasing continuous function such that $g(c)=\lambda c^{\kappa}$, so that $f$ is continuous, and $g(x)\ge\lambda x^{\pwr}$, $x\in[c,+\infty)$.
For instance, $g$ can be of the form $g(x)=\lambda'x^{\gamma}$, with some $\gamma>\pwr$
(and $\lambda'= \lambda c^{\pwr-\gamma}$), or it can be even faster increasing on $[c,+\infty)$, e.g., as $g(x)=\alpha\beta^x$ with some $\beta>1$ (and $\alpha=\lambda c^{\pwr}/\beta^c$). 

Note that for density operators $\rho,\sigma$, 
$\rho\le\lambda\sigma^{\pwr_1}$ $\imp$ $\rho\le\lambda\sigma^{\pwr_2}$
if $\kappa_1\ge\kappa_2$, hence it is best to prove continuity on sets of the form
$\{(\rho,\sigma):\,\rho\le\lambda\sigma^{\pwr}\}$ with as small $\pwr$ as possible. 
The same does not seem a priori obvious for more general pairs of operators, but it turns out to be nevertheless true; see Lemma \ref{lemma:B S cont}.
In the end it turns out that when considering continuity on sets of the form 
given in \eqref{B f}--\eqref{B f3}, all that really matters is the local behaviour of 
$f$ at $0$; see Lemma \ref{lemma:B S cont} and 
the proof of Theorem \ref{thm:main}.
\end{rem}

\begin{rem}
A given quantum R\'enyi $\alpha$-divergence is a function on 
$\B(\hil)\pne\times\B(\hil)\pne$, and it makes sense to try to characterize the set of points 
at which it is continuous. This, however, does not seem to be the relevant question for 
applications. For instance, as Example \ref{rem:discont} shows, $D_{\alpha}^{\cl}$ is 
discontinuous at $((1,0),(1,0))$ for any $\alpha>1$.
However, as we show below, $D_{\alpha}^{\cl}$ is continuous on 
$(\ell^{\infty}(\X)\times\ell^{\infty}(\X))_{\lambda}$ for any $\lambda>0$, and obviously,
$((1,0),(1,0))\in(\ell^{\infty}(\X)\times\ell^{\infty}(\X))_{\lambda}$ for $\lambda\ge 1$.
Thus, our aim is to study continuity of the R\'enyi divergences when their domain is 
restricted to some proper subset of their maximal domain $\B(\hil)\pne\times\B(\hil)\pne$.

In this context, 
it might be tempting to look for a ``maximal set of continuty'' for a given quantum 
R\'enyi $\alpha$-divergence. However, such a set cannot be defined in a meaningful way. 
Indeed, we will show below that certain quantum 
R\'enyi $\alpha$-divergences with $\alpha>1$ are continuous on 
$\bz\B(\hil)\times\B(\hil)\jz_{\lambda}$ for every $\lambda>0$; however, by 
Example \ref{rem:discont}, they are not continuous
on $\cup_{\lambda>0}\bz\B(\hil)\times\B(\hil)\jz_{\lambda}=\{(\rho,\sigma)\in\B(\hil)\pne\times\B(\hil)\pne):\,\rho^0\le\sigma^0\}$.
\end{rem}

Clearly, continuity on a larger set implies continuity on a smaller one, and the following relations are easy to see:
\begin{align}\label{containmnets}
\begin{array}{ccc}
\bz\S(\hil)\times\S(\hil)\jz_{\lambda,\pwr}&\subseteq &
\bz\S(\hil)\times\S(\hil)\jz_{\lambda',\pwr'}\\
\vertcont & & \vertcont \\
\bz[0,I]\times[0,I]\jz_{\lambda,\pwr}&\subseteq &\bz[0,I]\times[0,I]\jz_{\lambda',\pwr'}\\
\vertcont & &  \\
\bz\B(\hil)\times\B(\hil)\jz_{\lambda,\pwr}&\subseteq &
\bz\B(\hil)\times\B(\hil)\jz_{\lambda',\pwr}
\end{array}
\end{align}
when $\lambda\le\lambda'$ and $\kappa\ge\kappa'$.

\begin{rem}
For states $\rho,\sigma$, $\rho\le\lambda\sigma^{\pwr}$ 
implies $\rho\le\lambda\sigma$ when $\pwr\ge 1$, whence
$\lambda\ge 1$ has to hold, and
$\lambda=1$ is only possible if $\rho=\sigma$. 
Hence, for $\pwr\ge 1$, the study of continuity on 
$\bz\S(\hil)\times\S(\hil)\jz_{\lambda,\pwr}$ is only meaningful for $\lambda>1$.
\end{rem}

Continuity on the above types of sets are closely related as follows:
\begin{lemma}\label{lemma:B S cont}
Let $D_{\alpha}^q$ be a quantum R\'enyi $\alpha$-divergence for some $\alpha\in(0,+\infty)$,
let $\pwr>0$,
and let $\hil$ be a finite-dimensional Hilbert space with $\dim\hil\ge 2$.
If $D_{\alpha}^q$ satisfies the scaling law \eqref{scaling} then
the following are equivalent:
\begin{enumerate}
\item\label{az cont dom0}
$D_{\alpha}^q$ is continuous on $(\B(\hil)\times\B(\hil))_{\lambda,\pwr'}$ for every $\lambda>0$ and $\pwr'\in[\pwr,+\infty)$. 
\item\label{az cont dom1}
$D_{\alpha}^q$ is continuous on $(\B(\hil)\times\B(\hil))_{\lambda,\pwr}$ for every $\lambda>0$.
\item\label{az cont dom2}
$D_{\alpha}^q$ is continuous on $(\B(\hil)\times\B(\hil))_{\lambda,\pwr}$ for some $\lambda>0$.
\item\label{az cont dom3}
$D_{\alpha}^q$ is continuous on $([0,cI]\times[0,cI])_{\lambda,\pwr}$ for every $\lambda>0$ and $c>0$.
\item\label{az cont dom4}
$D_{\alpha}^q$ is continuous on $([0,cI]\times[0,cI])_{\lambda,\pwr}$ for some $\lambda>0$
and $c>0$.
\item\label{az cont dom5}
$D_{\alpha}^q$ is continuous on $(\S(\hil)\times\S(\hil))_{\lambda,\pwr}$ for every $\lambda>0$.
\end{enumerate}
\end{lemma}
\begin{proof}
The implications 
\ref{az cont dom0}$\imp$\ref{az cont dom1}$\imp$\ref{az cont dom2},
\ref{az cont dom1}$\imp$\ref{az cont dom3}$\imp$ \ref{az cont dom4}, 
and \ref{az cont dom3}$\imp$\ref{az cont dom5}
are obvious.

We prove \ref{az cont dom2}$\imp$\ref{az cont dom1}, 
\ref{az cont dom4}$\imp$\ref{az cont dom1},
and
\ref{az cont dom5}$\imp$\ref{az cont dom1}
by contraposition. Assume that 
$D_{\alpha}^q$ is not continuous on $(\B(\hil)\times\B(\hil))_{\lambda,\pwr}$ for some $\lambda>0$, i.e., there exists a sequence 
$(\rho_n,\sigma_n)\in (\B(\hil)\times\B(\hil))_{\lambda,\pwr}$ converging to 
some $(\rho,\sigma)\in (\B(\hil)\times\B(\hil))_{\lambda,\pwr}$ such that 
\begin{align*}
D_{\alpha}^q(\rho_n\|\sigma_n)-D_{\alpha}^q(\rho\|\sigma)
\not\to 0
\end{align*}
as $n\to+\infty$. Let $\lambda'>0$.
Then with $\tilde\rho_n:=\rho_n$, $\tilde\sigma_n:=(\lambda/\lambda')^{1/\pwr}\sigma_n$, $n\in\bN$, we have
$(\tilde\rho_n,\tilde\sigma_n)\in(\B(\hil)\times\B(\hil))_{\lambda',\pwr}$, $n\in\bN$, 
$\tilde\rho:=\lim_{n\to+\infty}\tilde\rho_n=\rho$,
$\tilde\sigma:=\lim_{n\to+\infty}\tilde\sigma_n=(\lambda/\lambda')^{1/\pwr}\sigma$,
whence $(\tilde\rho,\tilde\sigma)\in(\B(\hil)\times\B(\hil))_{\lambda',\pwr}$,
and 
\begin{align*}
D_{\alpha}^q(\tilde\rho_n\|\tilde\sigma_n)-
D_{\alpha}^q(\tilde\rho\|\tilde\sigma)\
=
D_{\alpha}^q(\rho_n\|\sigma_n)-D_{\alpha}^q(\rho\|\sigma)
\not\to 0,
\end{align*}
where the equality follows from the scaling law \eqref{scaling}.
This proves \ref{az cont dom2}$\imp$\ref{az cont dom1}.

The proof of \ref{az cont dom4}$\imp$\ref{az cont dom1} goes very similarly:
we define
$\tilde\rho_n:=\rho_n/K$, $\tilde\sigma_n:=(\lambda/\lambda')^{1/\pwr}\sigma_n/K'$ with 
$K' :=(\lambda/\lambda')^{1/\pwr}\norm{\sigma}_{\infty}c+1$, 
$K:=\max\{(K')^{\pwr},\norm{\rho}_{\infty}c\}+1$.

To prove \ref{az cont dom5}$\imp$\ref{az cont dom1},
define 
$\hat\rho_n:=\rho_n/\Tr\rho_n$ and 
$\hat\sigma_n:=\sigma_n/\Tr\sigma_n$, $n\in\bN$, so that 
$\hat\rho:=\lim_{n\to+\infty}\hat\rho_n=\rho/\Tr\rho$, 
$\hat\sigma:=\lim_{n\to+\infty}\hat\sigma_n=\sigma/\Tr\sigma$.
By assumption,
\begin{align*}
\hat\rho_n=\frac{\rho_n}{\Tr\rho_n}
\le
\frac{\lambda((\Tr\sigma_n)\hat\sigma_n)^{\pwr}}{\Tr\rho_n}
\le
\bz\frac{\lambda(\Tr\sigma)^{\pwr}}{\Tr\rho}+1\jz\hat\sigma_n^{\pwr}
\end{align*}
for every large enough $n$, i.e., 
$(\hat\rho_n,\hat\sigma_n)\in(\S(\hil)\times\S(\hil))_{\lambda',\pwr}$
with
$\lambda' :=\lambda(\Tr\sigma)^{\pwr}/\Tr\rho+1$,
and hence also $(\hat\rho,\hat\sigma)\in
(\S(\hil)\times\S(\hil))_{\lambda',\pwr}$. Then,
\begin{align*}
&D_{\alpha}^q(\hat\rho_n\|\hat\sigma_n)-
D_{\alpha}^q(\hat\rho\|\hat\sigma)\\
&\ds=
\underbrace{D_{\alpha}^q(\rho_n\|\sigma_n)-D_{\alpha}^q(\rho\|\sigma)}_{
\not\to 0}
+\underbrace{\log\Tr\rho-\log\Tr\rho_n}_{\to 0}
+\underbrace{\log\Tr\sigma_n-\log\Tr\sigma}_{\to 0}
\not\to 0,
\end{align*}
where the equality is again due the scaling law \eqref{scaling}.

Finally, \ref{az cont dom3} implies that $D_{\alpha}^q$ is continuous on 
$([0,I]\times[0,I])_{\lambda,\pwr'}$ for every $\lambda>0$
and every $\pwr'\in[\pwr,+\infty)$, according to \eqref{containmnets}.
Thus, for any fixed $\lambda>0$, \ref{az cont dom4} holds with $c=1$ and $\pwr'$ in place of $\pwr$, and the application of \ref{az cont dom4}$\imp$\ref{az cont dom1} above with $\pwr'$ in place of $\pwr$ yields continuity on $(\B(\hil)\times\B(\hil))_{\lambda,\pwr'}$. Since this holds for every $\lambda>0$ and $\pwr'\in[\pwr,+\infty)$, \ref{az cont dom0} holds.
\end{proof}

\begin{rem}
Note that the implication  \ref{az cont dom1}$\imp$\ref{az cont dom0}
above is non-trivial in the sense that 
for $\pwr\ne\pwr'$ and any $\lambda,\lambda'$, 
\begin{align*}
(\B(\hil)\times\B(\hil))_{\lambda,\pwr}
\left\{\begin{array}{c}\not\subseteq\\ \not\supseteq\end{array}\right\}
(\B(\hil)\times\B(\hil))_{\lambda',\pwr'}.
%
\end{align*}
\end{rem}

We prove the continuity of various quantum R\'enyi divergences on sets of the form
$(\B(\hil)\times\B(\hil))_{\lambda,\pwr}$ in Sections \ref{sec:classical cont} and \ref{sec:cont} below.
\medskip

Finally, we remark that while continuity of quantum R\'enyi divergences is a non-trivial problem, lower semi-continuity 
holds under very general conditions. The following argument has been applied to 
prove the lower semi-continuity of various quantum R\'enyi divergences; see, e.g., 
\cite[Lemma 3.26, Corollary 3.27]{MO-cqconv} and 
\cite[Lemma IV.8]{MO-cqconv-cc}. We state it here for completeness.

\begin{lemma}\label{lemma:Renyi lsc}
Let $\alpha\in(0,+\infty]$, let $\hil$ be a finite-dimensional Hilbert space, and let $D_{\alpha}^q$ be a quantum R\'enyi 
$\alpha$-divergence satisfying the following:
\begin{enumerate}
\item
For every $\ep\in(0,+\infty)$, $\B(\hil)\pne\times\B(\hil)\pne\ni(\rho,\sigma)\mapsto
D_{\alpha}^q(\rho\|\sigma+\ep I)$ is continuous. 
\item
For every $\rho,\sigma\in\B(\hil)\pne$, $(0,+\infty)\ni\ep\mapsto D_{\alpha}^q(\rho\|\sigma+\ep I)$ is monotone decreasing, and 
\begin{align*}
D_{\alpha}^q(\rho\|\sigma)
=
\lim_{\ep\searrow 0}D_{\alpha}^q(\rho\|\sigma+\ep I).
\end{align*}
\end{enumerate}
Then $D_{\alpha}^q$ is lower semi-continuous continuous on $\B(\hil)\pne\times\B(\hil)\pne$.
\end{lemma}
\begin{proof}
Immediate from the fact that the supremum of continuous functions is lower semi-continuous.
\end{proof}

\begin{prop}\label{prop:lsc quantum Renyi}
The following quantum R\'enyi divergences are lower semi-continuous on 
$\B(\hil)\pne\times\B(\hil)\pne$ for any finite-dimensional Hilbert space $\hil$:
\begin{enumerate}
\item\label{Renyi lsc1}
$D_{\alpha,z}$, $\alpha\in(0,+\infty)$, $z\in(0,+\infty]$; 
\item\label{Renyi lsc2}
$D_{\max}$;
\item\label{Renyi lsc3}
$D_{\alpha}^{\meas}$, $\alpha\in(0,+\infty)$.
\end{enumerate}
\end{prop}
\begin{proof}
It is straightforward to verify that for every $\alpha\in(0,+\infty)$, $z\in(0,+\infty)$,
$D_{\alpha,z}$ satisfies the conditions in Lemma \ref{lemma:Renyi lsc}, 
while the case $z=+\infty$ is covered in \cite[Lemma 3.26]{MO-cqconv}.
From this,
\ref{Renyi lsc1} follows. Since $D_{\max}=\sup_{\alpha>1}D_{\alpha,\alpha}$ is the supremum of lower semi-continuous functions, it is itself lower semi-continuous.
This proves \ref{Renyi lsc2}.
By the above, for every $\alpha\in(0,+\infty)$, and 
for any $\M\in\povm(\hil,k)$ and $\ep>0$, 
\begin{align*}
\B(\hil)\pne\times\B(\hil)\pne\ni(\rho,\sigma)\mapsto
D_{\alpha}^{\cl}(\M(\rho)\|\M(\sigma)+\ep I)
\end{align*}
is continuous, and 
\begin{align*}
D_{\alpha}^{\meas}(\rho\|\sigma)=\sup_{k\in\bN}\sup_{M\in\povm(\hil,k)}\sup_{\ep>0}
D_{\alpha}^{\cl}(\M(\rho)\|\M(\sigma)+\ep I),
\end{align*}
whence by the same argument as above, $D_{\alpha}^{\meas}$ is lower semi-continuous on 
$\B(\hil)\pne\times\B(\hil)\pne$. This proves \ref{Renyi lsc3}.
\end{proof}

For $\alpha\in(0,2]$, the lower semi-continuity of the maximal R\'enyi $\alpha$-divergences 
follows from \cite[Theorem 5.5]{Hiai_maxfdiv}, which was proved 
more generally for maximal $f$-divergences corresponding to operator convex functions, and 
in 
the general von Neumann algebra setting. We give a simple proof of 
the lower semi-continuity of the maximal $f$-divergences in the finite-dimensional case 
in Appendix \ref{sec:maxfdiv lsc}, which also works for more general convex (and not necessarily operator convex) functions. 
In particular, Theorem \ref{thm:maxfdiv lsc} implies the following:

\begin{prop}
For any finite-dimensional Hilbert space $\hil$, and for any $\alpha\in(0,+\infty)$, 
$D_{\alpha}^{\max}$ is lower semi-continuous on $\B(\hil)\pne\times\B(\hil)\pne$.
\end{prop}
\begin{proof}
Immediate from 
Theorem \ref{thm:maxfdiv lsc} due to \eqref{max Renyi from maxfdiv}
and \eqref{maxrelentr}.
\end{proof}

\subsection{Classical and measured R\' enyi divergences}
\label{sec:classical cont}

We start with the following properties of the classical R\'enyi divergences, which are 
partly well known and partly easy to verify:
\begin{lemma}\label{lemma:classical cont}
Let $\X$ be a finite set, and let $f$ be a $\pwr$-function.
\begin{enumerate}
\item\label{cl cont1}
The functions 
\begin{align}
&(0,1)\times\ell^{\infty}(\X)\pne\times\ell^{\infty}(\X)\pne\ni(\alpha,\rho,\sigma)\mapsto
D_{\alpha}^{\cl}(\rho\|\sigma),\label{classical jointcont1}\\
&\{\alpha\in(1,+\infty):\,(1-\pwr)\alpha<1\}\times\bz\ell^{\infty}(\X)\times\ell^{\infty}(\X)\jz_{f}\ni(\alpha,\rho,\sigma)\mapsto
D_{\alpha}^{\cl}(\rho\|\sigma)\label{classical jointcont2}
\end{align}
are continuous.
\item\label{cl cont2}
For every $\alpha\in(0,1)$,
\begin{align}\label{classical cont1}
D_{\alpha}^{\cl} \ds\text{ is continuous on  }\ds\ell^{\infty}(\X)\pne\times\ell^{\infty}(\X)\pne,
\end{align}
and for every $\alpha\in(0,+\infty)$ such that $(1-\pwr)\alpha<1$,
\begin{align}\label{classical cont}
D_{\alpha}^{\cl} \ds\text{ is continuous on  }\ds
\bz\ell^{\infty}(\X)\times\ell^{\infty}(\X)\jz_{f}.
\end{align}
\item\label{cl cont0}
Assume that $(1-\pwr)\alpha\ge 1$, $|\X|\ge 2$, and 
$f$ is a strict $\pwr$-function such that $x\le f(x)$, $x\in(1-\ep,1)$, with some 
$\ep\in(0,1)$. Then 
$D_{\alpha}^{\cl}$ is not continuous on $\bz\S(\X)\times\S(\X)\jz_{f}$.
\item\label{cl cont3}
For any $\rho,\sigma\in \ell^{\infty}(\X)\pne$,
$\alpha\mapsto\psi_{\alpha}^{\cl}(\rho\|\sigma)$ is convex, and 
$\alpha\mapsto D_{\alpha}^{\cl}(\rho\|\sigma)$ is monotone increasing and continuous
 on $(0,+\infty)$.
\end{enumerate}
\end{lemma}
\begin{proof}
The continuity of $\alpha\mapsto D_{\alpha}^{\cl}(\rho\|\sigma)$ on $(0,1)\cup(1,+\infty)$ is obvious by definition, and its continuity at $\alpha=1$ follows by a 
straightforward computation.
Convexity of $\alpha\mapsto\psi_{\alpha}^{\cl}(\rho\|\sigma)$ is straightforward to verify by simply computing its second derivative, and from this, the monotonicity of $\alpha\mapsto D_{\alpha}^{\cl}(\rho\|\sigma)$ follows immediately.
This proves \ref{cl cont3}.

The continuity of \eqref{classical jointcont1} is obvious by definition, and the continuity in \eqref{classical cont1}
follows from it trivially. 
The continuity in \eqref{classical cont} follows trivially from 
\eqref{classical cont1} for $\alpha\in(0,1)$, and from 
the continuity in \eqref{classical jointcont2} when $\alpha\in(1,+\infty)$. 

To prove the continuity of \eqref{classical jointcont2},
note that $Q_{\alpha}^{\cl}(\rho\|\sigma)=\sum_{x\in\X}P_{\alpha}(\rho(x),\sigma(x))$,
where $P_{\alpha}$ is the perspective function of $\id_{[0,+\infty)}^{\alpha}$, defined for $x,y\in[0,+\infty)$ as
\begin{align*}
P_{\alpha}(x,y):=\lim_{\ep\searrow 0}(y+\ep)\bz\frac{x+\ep}{y+\ep}\jz^{\alpha}
=\begin{cases}
y(x/y)^{\alpha},& y>0,\\
+\infty,&y=0,\, x>0,\, \alpha>1,\\
0,&\text{otherwise}.
\end{cases}
\end{align*}
Thus, it is sufficient to prove that if $x_n,y_n\in[0,+\infty)$ and $\alpha_n\in(1,+\infty)$ are such that $x_n\le f(y_n)$ and $(1-\pwr)\alpha_n<1$, $n\in\bN$, and
$x_n\to x$, $y_n\to y$, $\alpha_n\to\alpha\in(1,+\infty)$
such that $(1-\pwr)\alpha<1$,
as $n\to+\infty$, then 
\begin{align}\label{persp cont}
P_{\alpha_n}(x_n,y_n)\xrightarrow[n\to+\infty]{}P_{\alpha}(x,y).
\end{align}
If $y>0$ then $y_n>0$ for all large enough $n$, and \eqref{persp cont} holds
trivially. Assume therefore that $y=0$, which implies that $x=0$, whence $P_{\alpha}(x,y)=0$. 
By the assumption that $x_n\le f(y_n)$, we have 
\begin{align*}
0\le P_{\alpha_n}(x_n,y_n)
=\begin{cases}
0,&y_n=0,\\
x_n^{\alpha_n}y_n^{1-\alpha_n}\le
f(y_n)^{\alpha_n}y_n^{1-\alpha_n},&y_n>0.
\end{cases}
\end{align*}
By Definition \ref{def:pwr-function}, there exist some $M>0$ and $N_1\in\bN$ such that $f(y_n)\le M y_n^{\pwr}$, $n\ge N_1$.
By the assumption $(1-\pwr)\alpha<1$, there also exists an $N_2\in\bN$ such that 
$\pwr\alpha_n+1-\alpha_n>0$ for every $n\ge N_2$. Hence,
\begin{align*}
0\le P_{\alpha_n}(x_n,y_n)\le M^{\alpha_n}y_n^{\pwr\alpha_n+1-\alpha_n},
\ds\ds\ds n\ge\max\{N_1,N_2\},
\end{align*}
whence
\begin{align*}
\lim_{n\to 0}P_{\alpha_n}(x_n,y_n)=0.
\end{align*}
This completes the proof of \eqref{classical jointcont2}.

The proof of \eqref{classical cont} for $\alpha=1$ follows similarly as above, by noting that $D_1^{\cl}(\rho\|\sigma)=(\sum_x\rho(x))\inv\sum_x P_{\eta}(\rho(x),\sigma(x))$, 
where $P_{\eta}$ is the perspective function of 
$\eta(t):=t\log t$, $t\in[0,+\infty)$; we leave the simple proof to the reader.

Finally, let us prove \ref{cl cont0}.
By the assumption that $f$ is a strict $\pwr$-function, there exist 
$\delta,\eta>0$ such that $\delta x^{\pwr}<f(x)$, $x\in(0,\eta)$. 
We may assume without loss of generality that $\X=\{1,\ldots,m\}$ for some $m\in\bN$, 
$m\ge 2$, and define
$\rho_n:=(1-p_n,p_n,0,\ldots,0)$,
$\sigma_n:=(1-q_n,q_n,0,\ldots,0)$,
with $p_n$ and $q_n$ as in Example \ref{rem:discont}, 
with $\pwr\le\beta/\gamma\le 1-1/\alpha$ and $c/d^{\pwr}\le\delta$.
Then $p_n\le \delta q_n^{\pwr}\le f(q_n)$ for every large enough $n$.
Note that $(1-\pwr)\alpha\ge 1$ implies that $\alpha>1$. 
Thus, $\beta/\gamma<1$, and $p_n>q_n$, whence $1-p_n\le 1-q_n\le f(1-q_n)$, for every large enough $n$. Therefore,  
$\rho_n\le f(\sigma_n)$ for every large enough $n$, 
and $\lim_{n\to+\infty}\rho_n=\lim_{n\to+\infty}\sigma_n=\egy_{\{1\}}$
(the indicator function of the singleton $\{1\}$), 
but 
$\lim_{n\to+\infty}D_{\alpha}^{\cl}(\rho_n\|\sigma_n)\ne 0=
D_{\alpha}(\egy_{\{1\}}\|\egy_{\{1\}})$, proving the asserted discontinuity.
\end{proof}

\begin{rem}\label{rem:alpha restriction}
Note that the condition $(1-\pwr)\alpha<1$ in \eqref{classical jointcont2} is trivially satisfied for all $\alpha\in(0,+\infty)$ when $\pwr\ge 1$, and otherwise it gives the restriction $\alpha<1/(1-\pwr)$.
According to \ref{cl cont0}, the condition $(1-\pwr)\alpha<1$ is the best possible in
the setting of Lemma \ref{lemma:classical cont}.
\end{rem}

Most of the above properties of classical R\'enyi divergences are inherited by the measured R\'enyi divergences:

\begin{prop}\label{prop:meas cont}
Let $\hil$ be a finite-dimensional Hilbert space, and let 
$f$ be a $\pwr$-function.
\begin{enumerate}
\item\label{measured cont1}
The functions 
\begin{align}
&(0,1)\times\B(\hil)\pne\times\B(\hil)\pne\ni(\alpha,\rho,\sigma)\mapsto
D_{\alpha}^{\meas}(\rho\|\sigma),\label{meas jointcont}\\
&\{\alpha\in(1,+\infty):\,(1-\pwr)\alpha<1\}\times\bz\B(\hil)\times\B(\hil)\jz_{f}\ni(\alpha,\rho,\sigma)\mapsto
D_{\alpha}^{\meas}(\rho\|\sigma)\label{meas jointcont2}
\end{align}
are continuous.
\item\label{measured cont2}
For every $\alpha\in(0,1)$,
\begin{align}\label{meas cont1}
D_{\alpha}^{\meas} \ds\text{ is continuous on  }\ds\B(\hil)\pne\times\B(\hil)\pne,
\end{align}
and for every $\alpha\in(1,+\infty)$ such that $(1-\pwr)\alpha<1$,
\begin{align}\label{meas cont}
D_{\alpha}^{\meas} \ds\text{ is continuous on  }\ds
\bz\B(\hil)\times\B(\hil)\jz_{f}.
\end{align}
\item\label{measured cont3}
For any $\rho,\sigma\in \B(\hil)\pne$,
$\alpha\mapsto\psi_{\alpha}^{\meas}(\rho\|\sigma)$ is convex on $(1,+\infty)$, and 
$\alpha\mapsto D_{\alpha}^{\meas}(\rho\|\sigma)$ is monotone increasing and continuous
 on $(0,+\infty)$.
\end{enumerate}
\end{prop}
\begin{proof}
It is known that in the definition \eqref{measured Renyi} of $D_{\alpha}^{\meas}$, 
one may restrict to POVMs with $d^2:=(\dim\hil)^2$ outcomes, i.e., 
\begin{align*}
D_{\alpha}^{\meas}(\rho\|\sigma)&=
\max\left\{D_{\alpha}^{\cl}\bz(\Tr M_i\rho)_{i=1}^{d^2}\big\|(\Tr M_i\sigma)_{i=1}^{d^2}\jz:\,
(M_i)_{i=1}^{d^2}\in\povm(\hil,d^2)\right\};
\end{align*}
see, e.g., \cite[Lemma 4.14]{HiaiMosonyi2017}. 
By \ref{cl cont1} of Lemma \ref{lemma:classical cont}, the functions
\begin{align*}
&(0,1)\times\B(\hil)\pne\times\B(\hil)\pne\times\povm(\hil,d^2)\\
&\ds\ds\ni\bz\alpha,\rho,\sigma,(M_i)_{i=1}^{d^2}\jz\mapsto 
D_{\alpha}^{\cl}\bz(\Tr M_i\rho)_{i=1}^{d^2}\big\|(\Tr M_i\sigma)_{i=1}^{d^2}\jz,\\
&\{\alpha\in(1,+\infty):\,(1-\pwr)\alpha<1\}\times\bz\B(\hil)\times\B(\hil)\jz_{f}\times\povm(\hil,d^2)\\
&\ds\ds\ni\bz\alpha,\rho,\sigma,(M_i)_{i=1}^{d^2}\jz\mapsto 
D_{\alpha}^{\cl}\bz(\Tr M_i\rho)_{i=1}^{d^2}\big\|(\Tr M_i\sigma)_{i=1}^{d^2}\jz
\end{align*}
are continuous. Since $\povm(\hil,d^2)$ is a compact set, \ref{measured cont1} follows
by \ref{usc3} of Lemma \ref{lemma:usc}. The statements in \ref{measured cont2}
follow trivially from \ref{measured cont1}, except for \eqref{meas cont} with $\alpha=1$.
That we can obtain by noting that 
\begin{align*}
\bz\B(\hil)\times\B(\hil)\jz_{f}\times\povm(\hil,d^2)
\ni\bz\rho,\sigma,(M_i)_{i=1}^{d^2}\jz\mapsto 
D_{1}^{\cl}\bz(\Tr M_i\rho)_{i=1}^{d^2}\big\|(\Tr M_i\sigma)_{i=1}^{d^2}\jz
\end{align*}
is continuous by \eqref{classical cont}, and applying again \ref{usc3} of Lemma \ref{lemma:usc}.

Since the supremum of convex functions is convex, and the supremum of monotone increasing functions is monotone increasing,  \ref{cl cont3} of Lemma \ref{lemma:classical cont} implies that 
$\alpha\mapsto\psi_{\alpha}^{\meas}(\rho\|\sigma)$ 
is convex on $(1,+\infty)$, and 
$\alpha\mapsto D_{\alpha}^{\meas}(\rho\|\sigma)$ 
is monotone increasing on $(0,+\infty)$.
Continuity of $\alpha\mapsto D_{\alpha}^{\meas}(\rho\|\sigma)$ on $(0,1)$ is trivial from 
the continuity of \eqref{meas jointcont}.
If $\alpha>1$ and $\rho^0\nleq\sigma^0$ then $D_{\alpha}^{\meas}(\rho\|\sigma)\equiv +\infty$, and hence 
$\alpha\mapsto D_{\alpha}^{\meas}(\rho\|\sigma)$ is continuous, on $(1,+\infty)$.
If $\rho^0\le\sigma^0$ then, by the above, 
$\alpha\mapsto \psi_{\alpha}^{\meas}(\rho\|\sigma)$ is a finite-valued convex function
on $(1,+\infty)$,  
from which the continuity of $\alpha\mapsto D_{\alpha}^{\meas}(\rho\|\sigma)$
on $(1,+\infty)$ follows immediately. 
Finally, continuity of  $\alpha\mapsto D_{\alpha}^{\meas}(\rho\|\sigma)$ at $\alpha=1$
was given in \cite[Proposition III.32]{MH-testdiv}.
\end{proof}

\begin{rem}
Since for $\alpha\in(0,1)$, 
\begin{align*}
\psi_{\alpha}^{\meas}(\rho\|\sigma)
=
\min_{M\in\povm(\hil,d^2)}\psi_{\alpha}^{\cl}
\bz(\Tr M_i\rho)_{i=1}^{d^2}\big\|(\Tr M_i\sigma)_{i=1}^{d^2}\jz,
\end{align*}
the above argument cannot be used to establish convexity of 
$\alpha\mapsto \psi_{\alpha}^{\meas}(\rho\|\sigma)$ on $(0,1)$; in fact, it does not even seem to be known whether convexity holds in this case.
\end{rem}

\subsection{Examples of discontinuity}
\label{sec:discont}

Unlike the measured R\'enyi divergences, general quantum R\'enyi divergences need not inherit the continuity properties of the classical R\'enyi divergences given in Lemma \ref{lemma:classical cont}.
Below we give such examples for the 
maximal R\'enyi divergences, and the
R\'enyi $(\alpha,z)$-divergences with $\alpha>1$. 

\begin{prop}\label{prop:a>2 discont}
For every $\gamma>0$, there exist two sequences of qubit states $(\rho_n)_{n\in\bN}$, 
$(\sigma_n)_{n\in\bN}$, such that 
\begin{align}\label{P-R counter5}
\lim_{n\to+\infty}D_{\max}(\rho_n\|\sigma_n)=0,
\end{align}
while 
\begin{align*}
\lim_{n\to+\infty}D_{\alpha,z}(\rho_n\|\sigma_n)=+\infty
\end{align*}
if
\begin{align}\label{P-R counter3}
\alpha>2+\gamma\ds\text{and}\ds z\in(0,1],
\end{align}
or
\begin{align}\label{P-R counter4}
\alpha>1+z(1+\gamma)\ds\text{and}\ds z>1.
\end{align}
\end{prop}
\begin{proof}
Let $\gamma>0$ be fixed. Let $X=\begin{bmatrix}0 & 1\\ 1 & 0\end{bmatrix}$ and 
$Z=\begin{bmatrix}1 & 0\\ 0 & -1\end{bmatrix}$ be the Pauli $X$ and $Z$ operators, respectively. Let 
\begin{align*}
1/2<c_1<c_2<\ldots\to 1,\ds\ds
\delta_n:=(1-c_n)^{1+\gamma},\ds\ds
b_n:=c_n-\delta_n,\ds\ds
a_n:=\sqrt{c_n^2-b_n^2},\ds\ds n\in\bN,
\end{align*}
and
\begin{align*}
\rho_n&:=\half\bz I+c_n Z\jz=\underbrace{\frac{1+c_n}{2}}_{=:\kappa_{n,0}}\pr{e_0}
+\underbrace{\frac{1-c_n}{2}}_{=:\kappa_{n,1}}\pr{e_1},\\
\sigma_n&:=\half(I+a_nX+b_nZ)=\underbrace{\frac{1+c_n}{2}}_{=:\eta_{n,0}}\pr{f_{n,0}}
+\underbrace{\frac{1-c_n}{2}}_{=:\eta_{n,1}}\pr{f_{n,1}},
\ds\ds\ds
n\in\bN,
\end{align*}
be two sequences of qubit density operators, where 
$e_0,e_1$ are the canonical basis vectors of $\bC^2$, and 
$f_{0,n},f_{1,n}$ are normalized eigenvectors of $\sigma_n$.
For any $\lambda>1$, 
\begin{align*}
\lambda\sigma_n-\rho_n
=
\half\left[(\lambda-1)I+\lambda a_n X+(\lambda b_n-c_n)Z\right],
\end{align*}
which is PSD if and only if
\begin{align}
&(\lambda-1)^2\ge (\lambda a_n)^2+(\lambda b_n-c_n)^2\nn\\
&\ds\iff\ds
\lambda^2\underbrace{(1-a_n^2-b_n^2)}_{=1-c_n^2}
+2\lambda\underbrace{(b_nc_n-1)}_{=c_n^2-1-\delta_nc_n}+1-c_n^2\ge 0\nn\\
&\ds\iff\ds
(1-c_n)^2(\lambda^2-2\lambda+1)-2\lambda\delta_nc_n\ge 0\nn\\
&\ds\iff\ds
\frac{(\lambda-1)^2}{2\lambda}
\ge 
\frac{\delta_n}{1-c_n}\frac{c_n}{1+c_n}
=
(1-c_n)^{\gamma}\frac{c_n}{1+c_n}.
\label{P-R counter1}
\end{align}
Since the RHS of \eqref{P-R counter1} goes to $0$ as $n\to+\infty$, 
for any $\lambda>1$ there exists an $n_{\lambda}$ such that 
the inequality in \eqref{P-R counter1} holds for every $n\ge n_{\lambda}$.
Thus, 
\eqref{P-R counter5} holds.

Let $p_n,q_n$ be the Nussbaum-Szko\l a distributions \cite{NSz} on $\{0,1\}^2$, defined as
\begin{align*}
p_n(i,j):=\kappa_{n,i}|\inner{e_i}{f_{n,j}}|^2,\ds\ds\ds
q_n(i,j):=\eta_{n,j}|\inner{e_i}{f_{n,j}}|^2,\ds\ds\ds
i,j\in\{0,1\}.
\end{align*}
Then
\begin{align}
Q_{\alpha,1}(\rho_n\|\sigma_n)=Q_{\alpha}^{\cl}(p_n\|q_n)
\ge
\kappa_{n,0}^{\alpha}\eta_{n,1}^{1-\alpha}|\inner{e_0}{f_{n,1}}|^2,
\label{P-R counter2}
\end{align}
where the equality is well known and easy to verify, and the inequality is obvious.
It is not too difficult to verify that 
\begin{align*}
|\inner{e_0}{f_{n,1}}|^2=\frac{1}{2}\cdot\frac{\delta_n}{c_n}=\frac{1}{2}\cdot\frac{(1-c_n)^{1+\gamma}}{c_n},
\end{align*}
and hence the RHS of \eqref{P-R counter2} is
\begin{align*}
\half\bz\frac{1+c_n}{2}\jz^{\alpha}\bz\frac{1-c_n}{2}\jz^{1-\alpha}\frac{(1-c_n)^{1+\gamma}}{c_n}
=
\underbrace{(1-c_n)^{2+\gamma-\alpha}}_{\to +\infty}\cdot\underbrace{\frac{1}{4}\frac{(1+c_n)^{\alpha}}{c_n}}_{\to \frac{2^{\alpha}}{4}}
\xrightarrow[n\to+\infty]{}+\infty,
\end{align*}
where the convergence to $+\infty$ follows from the choice of $\gamma$ in \eqref{P-R counter3}. This proves the assertion for $z=1$. 
If $z\in(0,1)$, we may apply 
\eqref{ALT1} with $z_1:=z$, $z_2:=1$, to obtain that 
$\lim_{n\to+\infty}Q_{\alpha,z}(\rho_n\|\sigma_n)=+\infty$, proving the assertion for such $z$ values.

The proof for the case $z>1$ is completed by applying \eqref{ALT2}
with $z_1=1$ and $z_2=z$ to the above, which yields
\begin{align*}
Q_{\alpha,z}(\rho_n\|\sigma_n)
&\ge 
Q_{\alpha,1}(\rho_n\|\sigma_n)^z\norm{\rho_n}_{\infty}^{\alpha(1-z)}
\bz\Tr\sigma_n^{1-\alpha}\jz^{1-z}\\
&=
\underbrace{(1-c_n)^{(2+\gamma-\alpha)z}\cdot\bz\frac{1}{4}\frac{(1+c_n)^{\alpha}}{c_n}\jz^z}_{=Q_{\alpha,1}(\rho_n\|\sigma_n)^z}
\cdot\underbrace{\bz\frac{1+c_n}{2}\jz^{\alpha(1-z)}}_{=\norm{\rho_n}_{\infty}^{\alpha(1-z)}}
\cdot
\Bigg[
\underbrace{\bz\frac{1+c_n}{2}\jz^{1-\alpha}+\bz\frac{1-c_n}{2}\jz^{1-\alpha}
}_{=\Tr\sigma_n^{1-\alpha}\le 2\bz\frac{1-c_n}{2}\jz^{1-\alpha}}
\Bigg]^{1-z}
\\
&\ge 
\underbrace{(1-c_n)^{(2+\gamma-\alpha)z+(1-\alpha)(1-z)}}_{\xrightarrow[n\to +\infty]{}+\infty}
\cdot
\underbrace{\bz\frac{1}{4}\frac{(1+c_n)^{\alpha}}{c_n}\jz^z}_{\xrightarrow[n\to +\infty]{} 2^{(\alpha-2)z}}
\cdot
\underbrace{\bz\frac{1+c_n}{2}\jz^{\alpha(1-z)}}_{\xrightarrow[n\to +\infty]{} 1}
\cdot 2^{\alpha(1-z)}
\xrightarrow[n\to+\infty]{}+\infty,
\end{align*}
where the convergence to $+\infty$ follows
from the choice of $\gamma$ in \eqref{P-R counter4}.
\end{proof}

\begin{definition}
Let $\hil$ be a finite-dimensional Hilbert space.
We say that a function $\divv:\,\S(\hil)\times\S(\hil)\to[0,+\infty]$
is a \ki{pseudo-distance} on $\S(\hil)$, if 
$\divv(\rho\|\sigma)\ge 0$ for all $\rho,\sigma\in\S(\hil)$, with equality if and only if $\rho=\sigma$.
\end{definition}

According to \cite[Proposition A.30]{MO-cqconv-cc}, 
for any $\alpha\in(1,+\infty)$ and $z\in(0,+\infty)$,
$D_{\alpha,z}$ is a pseudo-distance on $\S(\hil)$ for any finite-dimensional Hilbert space 
$\hil$.

\begin{definition}
For two pseudo-distances $\divv_1,\divv_2$ on $\S(\hil)$, we say that 
$\divv_1\le\divv_2$ on $\S(\hil)$ if 
$\divv_1(\rho\|\sigma)\le\divv_2(\rho\|\sigma)$ for every $\rho,\sigma\in\S(\hil)$.
\end{definition}

Proposition \ref{prop:a>2 discont} yields the following:

\begin{cor}\label{cor:a>2 discont}
Let $\hil$ be a finite-dimensional Hilbert space with $\dim\hil\ge 2$, 
and $\divv$ be a pseudo-distance on $\S(\hil)$ such that 
$\divv\ge D_{\alpha,z}$ on $\S(\hil)$
for some $\alpha>2$ and $z\in(0,\alpha-1)$.
Then there exist two sequences of states 
$\hat\rho_n,\hat\sigma_n\in\S(\hil)$, $n\in\bN$, such that 
\begin{align}\label{divv discont1}
\lim_{n\to+\infty}D_{\max}(\hat\rho_n\|\hat\sigma_n)=0,
\ds\ds\ds
\lim_{n\to+\infty}\divv(\hat\rho_n\|\hat\sigma_n)=+\infty.
\end{align}
In particular, 
$\divv$
is not continuous on 
$\bz\S(\hil)\times\S(\hil)\jz_{\lambda}$
for any $\lambda>1$.
\end{cor}
\begin{proof}
By assumption, there exists a $\gamma>0$ such that 
\eqref{P-R counter3} or \eqref{P-R counter4} holds.
Let $\rho_n,\sigma_n\in\S(\bC^2)$, $n\in\bN$, be
as in Proposition \ref{prop:a>2 discont}, let $V:\,\bC^2\to\hil$ be any isometry, and 
$\hat\rho_n:=V\rho_n V^*$, 
$\hat\sigma_n:=V\sigma_n V^*$, $n\in\bN$.
Then \eqref{divv discont1} follows from the 
isometric invariance of $D_{\alpha,z}$ and $D_{\max}$, and the assumption that 
$\divv\ge D_{\alpha,z}$ on $\S(\hil)$.
From this, discontinuity of $\divv$ on $\bz\S(\hil)\times\S(\hil)\jz_{\lambda}$ 
for every $\lambda>1$ follows immediately.
\end{proof}

Regarding the problem of continuity on sets of the form $(\S(\hil)\times\S(\hil))_{\lambda}$,
Corollary \ref{cor:a>2 discont}
exhibits the most extreme form of discontinuity, 
with $\lim_{n\to+\infty}D_{\max}(\rho_n\|\sigma_n)=0$ and 
$\lim_{n\to+\infty}\divv(\rho_n\|\sigma_n)=+\infty$. 
According to Lemma \ref{lemma:Dmax bound} below, the latter is not possible 
if $\divv$ is a quantum R\'enyi $\alpha$-divergence that is monotone under CPTP maps, 
even under the weaker assumption that $(D_{\max}(\rho_n\|\sigma_n))_{n\in\bN}$ is bounded. 
Inspired by the example in \cite[Figure 1]{FF2020}, 
we present a different construction in Lemma \ref{lemma:a>1 discont} below
to prove a weaker form of discontinuity, 
where $\liminf_{n\to+\infty} \divv(\rho_n\|\sigma_n)>0$ while 
$\lim_{n\to+\infty}\rho_n=\lim_{n\to+\infty}\sigma_n$.
This construction works also in some cases where $\divv$ is a 
monotone R\'enyi divergence.

We will repeatedly use the simple observation that for any 
$\sigma\in\B(\hil)\pne$ and $\psi\in\hil$, 
and any $p,q\in\bR$,
\begin{align}
\bz \sigma^p\pr{\psi}\sigma^p\jz^q&=
\norm{\sigma^p\psi}^{2(q-1)} \sigma^p\pr{\psi}\sigma^p,
\label{sandwich power1}\\
\Tr\bz \sigma^p\pr{\psi}\sigma^p\jz^q&=
\norm{\sigma^p\psi}^{2q}.\label{sandwich power2}
\end{align}

\begin{lemma}\label{lemma:a>1 discont}
For every $c>0$, there exist qubit states $\rho_{c,\ep},\sigma_{c,\ep}$, $\ep\in[0,1/c)$,
where $\rho_{c,\ep}$ is pure and $\sigma_{c,\ep}$ is of full rank,
such that 
\begin{align*}
\rho_0=\lim_{\ep\searrow 0}\rho_{c,\ep}=\lim_{\ep\searrow 0}\sigma_{c,\ep}=\sigma_0,
\end{align*}
while 
\begin{align}
D_{\max}(\rho_{c,\ep}\|\sigma_{c,\ep})
=
\left\{\begin{array}{ll}
D_{\alpha}^{\max}(\rho_{c,\ep}\|\sigma_{c,\ep}),&\alpha\in(0,+\infty),\\
D_{\alpha,\alpha-1}(\rho_{c,\ep}\|\sigma_{c,\ep}),&\alpha>1,
\end{array}\right\}
\xrightarrow[\ep\searrow 0]{}\log\bz 1+\frac{1}{c}\jz.\label{a>1 discont1}
\end{align}
\end{lemma}
\begin{proof}
For every $c>0$ and $\ep\in[0,1)$, let
\begin{align*}
\psi_{\ep}:=\begin{bmatrix}\sqrt{\ep} \\ \sqrt{1-\ep}\end{bmatrix}\in\bC^2,\ds\ds
\rho_{c,\ep}:=\rho_{\ep}:=\pr{\psi_{\ep}},\ds\ds\ds
\sigma_{c,\ep}:=\begin{bmatrix} c\ep & 0 \\ 0 & 1-c\ep\end{bmatrix}\in\S(\bC^2).
\end{align*}
Then 
\begin{align*}
D_{\max}(\rho_{c,\ep}\|\sigma_{c,\ep})
&=
\log\snorm{\sigma_{c,\ep}^{-1/2}\pr{\psi_{\ep}}\sigma_{c,\ep}^{-1/2}}_{\infty}
=
\log\snorm{\sigma_{c,\ep}^{-1/2}\psi_{\ep}}^2,\\
D_{\alpha,\alpha-1}(\rho_{c,\ep}\|\sigma_{c,\ep})
&=
\frac{1}{\alpha-1}\log\Tr
\bz\sigma_{c,\ep}^{-1/2}\pr{\psi_{\ep}}\sigma_{c,\ep}^{-1/2}\jz^{\alpha-1}\\
&=
\log\snorm{\sigma_{c,\ep}^{-1/2}\psi_{\ep}}^2,
&\alpha\in(1,+\infty),
\end{align*}
where in the last equality we used \eqref{sandwich power2}.
For $\alpha\in(0,1)$, we have
\begin{align}
D_{\alpha}^{\max}(\rho_{c,\ep}\|\sigma_{c,\ep})
=
\frac{1}{\alpha-1}\log\Tr
\sigma_{c,\ep}^{1/2}\bz\sigma_{c,\ep}^{-1/2}\pr{\psi_{\ep}}\sigma_{c,\ep}^{-1/2}\jz^{\alpha}
\sigma_{c,\ep}^{1/2}
=
\log\snorm{\sigma_{c,\ep}^{-1/2}\psi_{\ep}}^2,
\label{a>1 discont proof1}
\end{align}
where the first equality follows from 
\eqref{maxRenyi persp}, and the second one from \eqref{sandwich power1}.
Since $\alpha\mapsto D_{\alpha}^{\max}(\rho_{c,\ep}\|\sigma_{c,\ep})$ is monotone increasing on $(0,+\infty)$ by definition and the monotonicity of the classical R\'enyi 
divergences in $\alpha$, and $D_{\alpha}^{\max}\le D_{\max}$ for any 
$\alpha\in(0,+\infty)$ (see, e.g., Lemma \ref{lemma:Dmax bound}), we get that 
\eqref{a>1 discont proof1} holds also for $\alpha\in[1,+\infty)$. 
Finally, 
\begin{align}
\snorm{\sigma_{c,\ep}^{-1/2}\psi_{\ep}}^2=\frac{1}{c}+\frac{1-\ep}{1-c\ep}
\xrightarrow[\ep\searrow 0]{}1+\frac{1}{c},\label{a>1 discont proof2}
\end{align}
as stated.
\end{proof}

Lemma \ref{lemma:a>1 discont} yields the following 
(cf.~Corollary \ref{cor:a>2 discont}):

\begin{prop}\label{prop:a>1 discont}
Let $\hil$ be a finite-dimensional Hilbert space with $\dim\hil\ge 2$, 
let $\lambda\in(1,+\infty)$ and $\pwr\in(0,+\infty)$.
Then there exist two sequences of states 
$\rho_n,\sigma_n\in\S(\hil)$, $n\in\bN$, 
where $\rho_n$ is pure and $\sigma_n$ is of full rank, 
such that 
\begin{align}
&\lim_{n\to+\infty}\rho_n=\lim_{n\to+\infty}\sigma_n,\label{divv discont2}\\
&\lim_{n\to+\infty}D_{\max}(\rho_n\|\sigma_n^{\pwr})
=\log\lambda,
\label{divv discont3}
\end{align}
while for every $\alpha\in(1,+\infty)$,
\begin{align}\label{divv discont4}
\lim_{n\to+\infty}D_{\alpha,(\alpha-1)/\pwr}(\rho_n\|\sigma_n)=\frac{1}{\pwr}\log\lambda>0.
\end{align}
If, moreover, $\pwr\in(0,1]$, then for every $\alpha\in(0,+\infty)$, 
\begin{align}\label{divv discont5}
D_{\alpha}^{\max}(\rho_n\|\sigma_n)=D_{\max}(\rho_n\|\sigma_n)\xrightarrow[n\to+\infty]{}
\begin{cases}
\log\lambda,&\pwr=1,\\
+\infty,&\pwr\in(0,1).
\end{cases}
\end{align}
\end{prop}
\begin{proof}
Let $c:=(\lambda-1)\inv$, and $0<\ep_n$, $n\in\bN$, be a sequence converging to $0$. 
Let $V:\,\bC^2\to\hil$ be any isometry, and $\rho_n:=V\rho_{c,\ep_n}V^*$, 
\begin{align*}
\tilde\sigma_n:=
\begin{cases}
V\sigma_{c,\ep_n}V^*,&\dim\hil=2,\\
(1-\ep_n)V\sigma_{c,\ep_n}V^*+\ep_n\frac{I-VV^*}{\Tr(I-VV^*)},&\dim\hil>2,
\end{cases}
\end{align*}
with the states $\rho_{c,\ep}$, $\sigma_{c,\ep}$ given in 
Lemma \ref{lemma:a>1 discont}, and let 
$\sigma_n:=\tilde\sigma_n^{1/\pwr}/\Tr\tilde\sigma_n^{1/\pwr}$.
Then 
\begin{align*}
\lim_{n\to+\infty}\rho_n=
\lim_{n\to+\infty}\tilde\sigma_n=
\lim_{n\to+\infty}\sigma_n=
V\begin{bmatrix}0 & 0 \\ 0 & 1\end{bmatrix}V^*,
\end{align*}
proving \eqref{divv discont2}.

Let us assume in the following that $\dim\hil>2$ (the case $\dim\hil=2$ is similar and easier). Then
we have
\begin{align}
D_{\max}\bz\rho_n\|\sigma_n^{\pwr}\jz
&=
D_{\max}\bz\rho_n\|\tilde\sigma_n\jz+\pwr\log\Tr(\tilde\sigma_n)^{1/\pwr}\nn\\
&=
D_{\max}\bz V\rho_{c,\ep_n}V^*\|(1-\ep_n)V\sigma_{c,\ep_n}V^*\jz+
\pwr\log\Tr(\tilde\sigma_n)^{1/\pwr}\nn\\
&=
\underbrace{D_{\max}\bz \rho_{c,\ep_n}\|\sigma_{c,\ep_n}\jz}_{
\xrightarrow[n\to+\infty]{}\log\lambda}-\log(1-\ep_n)+
\pwr\log\Tr(\tilde\sigma_n)^{1/\pwr},\label{divv discont proof1}
\end{align}
where the first and the third equalities follow from the scaling law \eqref{scaling},
the second equality from \eqref{ort3},
and the limit from \eqref{a>1 discont1}. From this, \eqref{divv discont3} follows immediately. 

Next, we consider
\begin{align*}
D_{\alpha,(\alpha-1)/\pwr}(\rho_n\|\sigma_n)
&=
\frac{1}{\alpha-1}\log\Tr\bz
\bz\frac{\tilde\sigma_n^{1/\pwr}}{\Tr\tilde\sigma_n^{1/\pwr}}\jz^{\frac{(1-\alpha)\pwr}{2(\alpha-1)}}
\rho_n^{\frac{\alpha\pwr}{\alpha-1}}
\bz\frac{\tilde\sigma_n^{1/\pwr}}{\Tr\tilde\sigma_n^{1/\pwr}}\jz^{\frac{(1-\alpha)\pwr}{2(\alpha-1)}}
\jz^{\frac{\alpha-1}{\pwr}}\\
&=
\frac{1}{\alpha-1}\log\Tr\bz\sigma_{c,\ep_n}^{-1/2}\pr{\psi_{\ep_n}}\sigma_{c,\ep_n}^{-1/2}\jz^{\frac{\alpha-1}{\pwr}}-\frac{1}{\pwr}\log(1-\ep_n)
+\log\Tr\tilde\sigma_n^{1/\pwr}\\
&=
\frac{1}{\pwr}\log\snorm{\sigma_{c,\ep_n}^{-1/2}\psi_{\ep_n}}^2
-\frac{1}{\pwr}\log(1-\ep_n)+\log\Tr\tilde\sigma_n^{1/\pwr}\\
&\ds\xrightarrow[n\to+\infty]{}\frac{1}{\pwr}\log\lambda,
\end{align*}
where the first equality is by definition, 
the second equality follows from \eqref{ort1} and the isometric invariance of 
$D_{\alpha,(\alpha-1)/\pwr}$, the third equality follows from \eqref{sandwich power2},
and the limit from 
\eqref{a>1 discont proof2}. This proves \eqref{divv discont4}.

Finally, let us prove \eqref{divv discont5} under the assumption that 
$\pwr\in(0,1]$. Assume first that $\alpha\in(0,1)$. Then
\begin{align*}
D_{\alpha}^{\max}(\rho_n\|\sigma_n)
&=
D_{\alpha}^{\max}(\rho_n\|\tilde\sigma_n^{1/\pwr})-\frac{1}{\pwr}\log(1-\ep_n)+\log\Tr\tilde\sigma_n^{1/\pwr}\\
&=
\frac{1}{\alpha-1}\log\Tr\tilde\sigma_n^{1/\pwr}\bz\tilde\sigma_n^{-1/2\pwr}\rho_n\tilde\sigma_n^{-1/2\pwr}\jz^{\alpha}-\frac{1}{\pwr}\log(1-\ep_n)
+\log\Tr\tilde\sigma_n^{1/\pwr}\\
&=
\frac{1}{\alpha-1}\log\Tr\tilde\sigma_{c,\ep_n}^{1/\pwr}
\bz\tilde\sigma_{c,\ep_n}^{-1/2\pwr}\pr{\psi_{\ep_n}}\tilde\sigma_{c,\ep_n}^{-1/2\pwr}\jz^{\alpha}-\frac{1}{\pwr}\log(1-\ep_n)
+\log\Tr\tilde\sigma_n^{1/\pwr}\\
&=
\frac{1}{\alpha-1}\log\snorm{\sigma_{c,\ep_n}^{-1/2\pwr}\psi_{\ep_n}}^{2\alpha-2}
-\frac{1}{\pwr}\log(1-\ep_n)
+\log\Tr\tilde\sigma_n^{1/\pwr}\\
&=
\log\bz c^{-1/\pwr}\ep_n^{1-1/\pwr}+(1-c\ep_n)^{-1/\pwr}(1-\ep_n)\jz
-\frac{1}{\pwr}\log(1-\ep_n)
+\log\Tr\tilde\sigma_n^{1/\pwr},
\end{align*}
where the first equality follows by the scaling law \eqref{scaling},
the second equality from \eqref{maxRenyi persp},
the third equality from \eqref{ort2} and the isometric invariance of $D_{\alpha}^{\max}$,
and the fourth equality follows by a straightforward computation.
Similarly,
\begin{align*}
D_{\max}(\rho_n\|\sigma_n)
&=
D_{\max}(\rho_n\|\tilde\sigma_n^{1/\pwr})+\log\Tr\tilde\sigma_n^{1/\pwr}\\
&=
\log\snorm{\sigma_{c,\ep_n}^{-1/2\pwr}\pr{\psi_{\ep_n}}\sigma_{c,\ep_n}^{-1/2\pwr}}^2
-\frac{1}{\pwr}\log(1-\ep_n)
+\log\Tr\tilde\sigma_n^{1/\pwr},
\end{align*}
proving the equality in \eqref{divv discont5} for $\alpha\in(0,1)$. 
Since $D_{\alpha}^{\max}$ is monotone increasing in $\alpha$, and 
$D_{\alpha}^{\max}\le D_{\max}$, $\alpha\in(0,+\infty)$,
(this is well known, but we also give a proof in 
Lemma \ref{lemma:Dmax bound} below), the equality extends to every 
$\alpha\in(0,+\infty)$. From these, the limit in \eqref{divv discont5} follows immediately. 
\end{proof}

\begin{cor}\label{cor:discont by domination}
Let $\pwr\in(0,+\infty)$, let $\hil$ be a finite-dimensional Hilbert space, and 
$\divv$ be a pseudo-distance on $\S(\hil)$.
If $\divv\ge D_{\alpha,(\alpha-1)/\pwr}$ on $\S(\hil)$ for some $\alpha\in(1,+\infty)$, 
or if $\pwr\in(0,1]$ and 
$\divv\ge D_{\alpha}^{\max}$ for some $\alpha\in(0,+\infty)$, then 
$\divv$ is not continuous on $(\S(\hil)\times\S(\hil))_{\lambda,\pwr}$ for any 
$\lambda>1$. 
\end{cor}
\begin{proof}
Immediate from Proposition \ref{prop:a>1 discont}.
\end{proof}

\begin{cor}\label{cor:discont examples}
Let $\hil$ be a finite-dimensional Hilbert space with $\dim\hil \ge 2$ and let $\pwr\in(0,+\infty)$. 
The following quantum R\'enyi divergences are not continuous on 
$(\S(\hil)\times\S(\hil))_{\lambda,\pwr}$ for any $\lambda>1$, depending on the value of 
$\pwr$:
\begin{enumerate}
\item
$\pwr\in(0,1)$:
\begin{align*}
D_{\alpha}\nw,\ds\ds\alpha\in[(1-\pwr)\inv,+\infty) \ds\ds\ds\text{(sandwiched)};
\end{align*}
\item
$\pwr\in(0,1]$:
\begin{align*}
&D_{\max};\\
&D_{\alpha}^{\max}, \ds\ds\alpha\in(0,+\infty);
\end{align*}
\item
$\pwr\in(0,+\infty)$:
\begin{align*}
&D_{\alpha,z}, \ds\ds\alpha\in(1,+\infty),\ds z\in(0,(\alpha-1)/\pwr];\\
&D_{\alpha,1},\ds\ds\alpha\in[\pwr+1,+\infty)\ds\ds\ds\text{(Petz-type).}
\end{align*}
\end{enumerate}
\end{cor}
\begin{proof}
Proposition \ref{prop:a>1 discont} yields immediately 
the discontinuity of $D_{\max}$, $D_{\alpha}^{\max}$, $\alpha\in(0,+\infty)$, and 
$D_{\alpha,(\alpha-1)/\pwr}$, $\alpha\in(1,+\infty)$ for the given values of $\pwr$. 
By \eqref{ALT3}, $z\in(0,(\alpha-1)/\pwr]$ $\imp$ $D_{\alpha,z}\ge D_{\alpha,(\alpha-1)/\pwr}$, whence the 
discontinuity of $D_{\alpha,z}$, $\alpha\in(1,+\infty)$, $z\in(0,(\alpha-1)/\pwr]$ follows
by Corollary \ref{cor:discont by domination}, and 
the assertions about the sandwiched and the Petz-type R\'enyi divergences follow as special cases. 
\end{proof}

\begin{rem}
In Appendix \ref{sec:discont ex}
we give a different example proving the discontinuity of 
$D_{\alpha,z}$ on $(\S(\hil)\times\S(\hil))_{\lambda}$ for 
$\alpha>1$ and $z\in(0,\alpha-1]$, which might be interesting 
due to the construction technique.
\end{rem}

\subsection{Continuity of $(\alpha,z)$-divergences}
\label{sec:cont}

The following is our main result on the continuity of R\'enyi $(\alpha,z)$-divergences, which we state and prove in the stronger form of continuity on the level of operators. 

\begin{theorem}\label{thm:main}
Let $f:\,[0,+\infty)\to[0,+\infty)$ be a 
$\pwr$-function.
Then,
for any finite-dimensional Hilbert space $\hil$, the function
\begin{align*}
\C_{f}(\hil)&:=\{(p,q)\in(0,+\infty)^2:\,q/\pwr<\min\{p,1\}\}\times
(0,+\infty)\times(\B(\hil)\times\B(\hil))_{f}\\
&\ds\ni(p,q,z,\rho,\sigma)\mapsto(\sigma^{-q/2}\rho^p\sigma^{-q/2})^z
\end{align*}
is continuous. 
\end{theorem}
\begin{proof}
By Lemma \ref{lemma:uniform conv},
it suffices to show that for
any sequence $(p_n,q_n,\rho_n,\sigma_n)\in\C_{f}(\hil)$
converging to some $(p,q,\rho,\sigma)\in\C_{f}(\hil)$, we have 
$\lim_{n\to+\infty}\sigma_n^{-q_n/2}\rho_n^{p_n}\sigma_n^{-q_n/2}= \sigma^{-q/2}\rho^p\sigma^{-q/2}$. 

Set $C_n:=\rho_n^{1/2}\sigma_n^{-q_n/2}$.
Since the condition $\rho_n\le f(\sigma_n)$ implies 
$\rho_n^0\le (f(\sigma_n))^0\le\sigma_n^0$, we have
\begin{align}\label{contt proof1}
\rho_n^{1/2}=C_n\sigma_n^{q_n/2},
\end{align}
and for any $\delta>0$ such that $\kappa-\delta>q$, 
\begin{align}
 C_n^*C_n&=
\sigma_n^{-q_n/2}\rho_n\sigma_n^{-q_n/2}\nn\\
&\le
\sigma_n^{-q_n/2}f(\sigma_n)\sigma_n^{-q_n/2}\nn\\
&=
(\id_{[0,+\infty)}^{\kappa-\delta-q_n}(\sigma_n)
(\id_{[0,+\infty)}^{-\kappa+\delta}\cdot f)(\sigma_n)\nn\\
&\xrightarrow[n\to+\infty]{}
(\id_{[0,+\infty)}^{\kappa-\delta-q}(\sigma)
(\id_{[0,+\infty)}^{-\kappa+\delta}\cdot f)(\sigma)\nn\\
&=
(\id_{[0,+\infty)}^{-q}\cdot f)(\sigma),\label{contt proof2}
\end{align}
where the convergence follows from Lemma \ref{lemma:uniform conv}, 
since
$\id_{[0,+\infty)}^{-\pwr+\delta}\cdot f$ is continuous on $[0,+\infty)$
(see Remark \eqref{rem:zero limit}).

Thus, $\vnorm{C_n}\le\gamma$ for some constant $\gamma>0$. 
The finite dimensionality of $\hil$ implies that 
$\{X\in \B(\cH):\vnorm{X}\le\gamma\}$ is compact, whence there exists a limit point
$C$ of $(C_n)_{n\in\bN}$, so that $\lim_{k\to+\infty}C_{n_k}= C\in \B(\cH)$ for some subsequence $(n_k)_{k\in\bN}$. 
By letting
$k\to\infty$ in \eqref{contt proof1} and \eqref{contt proof2} for $n=n_k$, 
and using again Lemma \ref{lemma:uniform conv}, we get
\begin{align}
\rho^{1/2}=C\sigma^{q/2},\ds\ds\text{and}\ds\ds C^*C\le(\id_{[0,+\infty)}^{-q}\cdot f)(\sigma).
\end{align}
In particular, $(C^*C)^0\le \sigma^0$, or equivalently, 
$\ran C^* \subseteq\ran \sigma$, which in turn is equivalent to 
$\ker\sigma\subseteq\ker C$. 
Thus, 
we get
\begin{align*}
C=\rho^{1/2}\sigma^{-q/2}.
\end{align*}
In particular, the sequence $(C_n)_{n\in\bN}$ has a unique limit point, 
whence the sequence is convergent with 
$\lim_{n\to+\infty}C_n=C$. 

Assume first that $p>1$ (and hence $p_n>1$ for every large enough $n$). Then, by the above and Lemma \ref{lemma:uniform conv},
\begin{align*}
\sigma_n^{-q_n/2}\rho_n^{p_n}\sigma_n^{-q_n/2}=C_n^*\rho_n^{p_n-1}C_n
\xrightarrow[n\to+\infty]{} C^*\rho^{p-1}C=\sigma^{-q/2}\rho^p\sigma^{-q/2}.
\end{align*}
In the general case,
one may choose an $r\in(0,1)$ such that 
$q/\pwr<r<p$, or equivalently, 
$(q/r)/\pwr<1<p/r$,
and hence also 
$(q_n/r)/\pwr<1<p_n/r$ for every large enough $n$. 
Then,
since $\id_{[0,+\infty)}^r$ is operator monotone, we have
\begin{align*}
\rho_n^r\le (f(\sigma_n))^r=\tilde f_r(\sigma_n^r),
\end{align*}
where $\tilde f_r(x):=(f(x^{1/r}))^r$, $x\in[0,+\infty)$, is again a $\pwr$-function. 
Thus, the above special case yields
\begin{align*}
\sigma_n^{-q_n/2}\rho_n^{p_n}\sigma_n^{-q_n/2}&=(\sigma_n^r)^{-q_n/2r}(\rho_n^r)^{p_n/r}(\sigma_n^r)^{-q_n/2r}\xrightarrow[n\to+\infty]{}(\sigma^r)^{-q/2r}(\rho^r)^{p/r}(\sigma^r)^{-q/2r}=\sigma^{-q/2}\rho^p\sigma^{-q/2},
\end{align*}
as required.
\end{proof}

\begin{cor}\label{cor:az-continuity}
For any finite-dimensional Hilbert space $\hil$, and any $\pwr$-function $f$, 
the function
\begin{align*}
\tilde\C_{f}(\hil)&:=
\{(\alpha,z)\in(1,+\infty)\times(0,+\infty):\,(1-\pwr)\alpha<1,\,(\alpha-1)/\pwr<z\}
\times
(\B(\hil)\times\B(\hil))_{f}\\
&\ni(\alpha,z,\rho,\sigma)
\mapsto D_{\alpha,z}(\rho\|\sigma)
\end{align*}
is continuous. 
\end{cor}
\begin{proof}
The map $(\alpha,z)\mapsto (\alpha/z,(\alpha-1)/z)=:(p,q)$ is continuous from 
$\{(\alpha,z)\in(1,+\infty)\times(0,+\infty):\,(1-\pwr)\alpha<1,\,(\alpha-1)/\pwr<z\}$ to
$\{(p,q)\in(0,+\infty)^2:\,q/\pwr<\min\{p,1\}\}$, and 
\begin{align*}
Q_{\alpha,z}(\rho\|\sigma)=\Tr(\sigma^{-q/2}\rho^p\sigma^{-q/2})^z
\end{align*}
for any $\rho,\sigma\in\B(\hil)\pne$ such that $\rho^0\le\sigma^0$.
Thus, the 
asserted continuity follows immediately from Theorem \ref{thm:main}.
\end{proof}

\begin{rem}
Recall from Remark \ref{rem:alpha restriction}
that the condition $(1-\pwr)\alpha<1$ is essential for the 
above established continuity 
property of the R\'enyi $(\alpha,z)$-divergences.
\end{rem}

Combining Lemma \ref{lemma:B S cont} and Corollaries \ref{cor:discont examples} and \ref{cor:az-continuity} yields
the following characterization of the continuity of the R\'enyi $(\alpha,z)$-divergences
on sets of the form $(\B(\hil)\times\B(\hil))_{\lambda,\pwr}$ when $\alpha>1$. 
(Recall that the case $\alpha\in(0,1)$ is obvious, as stated in Lemma \ref{lemma:0-1 cont}.)

\begin{theorem}\label{thm:az cont parameters}
Let $\hil$ be a finite-dimensional Hilbert space with $\dim\hil\ge 2$, let 
$\pwr\in(0,+\infty)$, $\alpha\in(1,+\infty)$ be such that 
$(1-\pwr)\alpha<1$, and let 
$z\in(0,+\infty)$. Then the following are equivalent:
\begin{enumerate}
\item\label{az cont dom0-2}
$D_{\alpha,z}$ is continuous on $(\B(\hil)\times\B(\hil))_{\lambda,\pwr'}$ for every $\lambda>0$ and $\pwr'\in[\pwr,+\infty)$.
\item\label{az cont dom1-2}
$D_{\alpha,z}$ is continuous on $(\B(\hil)\times\B(\hil))_{\lambda,\pwr}$ for every $\lambda>0$.
\item\label{az cont dom3-2}
$D_{\alpha,z}$ is continuous on $(\B(\hil)\times\B(\hil))_{\lambda,\pwr}$ for some $\lambda>0$.
\item\label{az cont dom4-1}
$D_{\alpha,z}$ is continuous on $([0,cI]\times[0,cI])_{\lambda,\pwr}$ for every $\lambda>0$ and $c>0$.
\item\label{az cont dom5-2}
$D_{\alpha,z}$ is continuous on $([0,cI]\times[0,cI])_{\lambda,\pwr}$ for some $\lambda>0$ and $c>0$.
\item\label{az cont dom6-2}
$D_{\alpha,z}$ is continuous on $(\S(\hil)\times\S(\hil))_{\lambda,\pwr}$ for every $\lambda>1$.
\item\label{az cont dom7-2}
$D_{\alpha,z}$ is continuous on $(\S(\hil)\times\S(\hil))_{\lambda,\pwr}$ for some $\lambda>1$.
\item\label{az cont dom8-2}
$z>(\alpha-1)/\pwr$.
\end{enumerate}
\end{theorem}
\begin{proof}
The equivalence of \ref{az cont dom0-2}--\ref{az cont dom6-2}
was given in Lemma \ref{lemma:B S cont}, and the implication 
\ref{az cont dom6-2}$\imp$\ref{az cont dom7-2}
is trivial.
The implication 
\ref{az cont dom7-2}$\imp$\ref{az cont dom8-2}
follows from Corollary \ref{cor:discont examples}, and the implication 
\ref{az cont dom8-2}$\imp$\ref{az cont dom1-2}
from Corollary \ref{cor:az-continuity}.
\end{proof}
\medskip

With the help of Corollary \ref{cor:az-continuity} and Theorem \ref{thm:az cont parameters}, we obtain the following analogues of Proposition \ref{prop:meas cont} for the regularized measured and the Petz-type R\'enyi divergences.

\begin{prop}\label{prop:regmeas cont}
Let $\hil$ be a finite-dimensional Hilbert space, and let $f$ be a $\pwr$-function.
\begin{enumerate}
\item\label{regmeasured cont1}
The functions 
\begin{align}
&(0,1)\times\B(\hil)\pne\times\B(\hil)\pne\ni(\alpha,\rho,\sigma)\mapsto
\oll{D}_{\alpha}^{\meas}(\rho\|\sigma),\label{regmeas jointcont}\\
&\{\alpha\in(1,+\infty):\,(1-\pwr)\alpha<1\}\times\bz\B(\hil)\times\B(\hil)\jz_{f}\ni(\alpha,\rho,\sigma)\mapsto
\oll{D}_{\alpha}^{\meas}(\rho\|\sigma)\label{regmeas jointcont2}
\end{align}
are continuous.
\item\label{regmeasured cont2}
For every $\alpha\in(0,1)$,
\begin{align}\label{regmeas cont1}
\oll{D}_{\alpha}^{\meas} \ds\text{ is continuous on  }\ds\B(\hil)\pne\times\B(\hil)\pne,
\end{align}
and for every $\alpha\in[1,+\infty)$ such that $(1-\pwr)\alpha<1$,
\begin{align}\label{regmeas cont}
\oll{D}_{\alpha}^{\meas}=D_{\alpha}\nw \ds\text{ is continuous on  }\ds
\bz\B(\hil)\times\B(\hil)\jz_{f}.
\end{align}
\item\label{regmeasured cont3}
For any $\rho,\sigma\in \B(\hil)\pne$,
$\alpha\mapsto\oll{\psi}_{\alpha}^{\meas}(\rho\|\sigma)$ is convex on $[1/2,+\infty)$, and 
$\alpha\mapsto \oll{D}_{\alpha}^{\meas}(\rho\|\sigma)$ is monotone increasing and continuous
 on $(0,+\infty)$.
\end{enumerate}
\end{prop}
\begin{proof}
\ref{regmeasured cont1}\s 
is immediate from \eqref{reg measured}, Lemma \ref{lemma:0-1 cont} and 
Corollary \ref{cor:az-continuity}.

\ref{regmeasured cont2}\s
is immediate from \ref{regmeasured cont1} for $\alpha\in(0,1)$, and 
for $\alpha\in(1,+\infty)$ with $(1-\pwr)\alpha<1$.
Note that by Theorems 5 and 7 in \cite{Renyi_new}, for any fixed $\rho,\sigma$, $\alpha\mapsto D_{\alpha}\nw(\rho\|\sigma)$
goes to $D_1\nw(\rho\|\sigma)=\DU(\rho\|\sigma)$ monotone decreasingly as $\alpha\searrow 1$, so by the above,
$D_1\nw$ is the infimum of continuous functions on 
$\bz\B(\hil)\times\B(\hil)\jz_{f}$,
and hence it is upper semi-continuous. 
Lower semi-continuity of $D_1\nw$ on the whole of $\B(\hil)\pne\times\B(\hil)\pne$
is well known, and can be easily seen from the simple fact that $D_1\nw(\rho\|\sigma)=\sup_{\ep>0}D_1\nw(\rho\|\sigma+\ep I)$. 

\ref{regmeasured cont3}\s
By \eqref{reg measured}, \cite[Corollary 4]{HT14}, and \ref{cl cont3} in Lemma \ref{lemma:classical cont}, 
$[1/2,+\infty)\ni\alpha\mapsto \oll{\psi}_{\alpha}^{\meas}(\rho\|\sigma)$
is the pointwise limit of convex functions, and hence is itself convex. 
The monotonicity 
of $\alpha\mapsto \oll{D}_{\alpha}^{\meas}(\rho\|\sigma)$ 
follows immediately from 
\ref{measured cont3} of Proposition \ref{prop:meas cont} and \eqref{regularized measured Renyi}, since the supremum of 
monotone functions in monotone.
Continuity of $\alpha\mapsto \oll{D}_{\alpha}^{\meas}(\rho\|\sigma)$
on $(0,1)\cup(1,+\infty)$ is obvious by \eqref{reg measured}, and at $\alpha=1$ it follows from 
\cite[Theorem 5]{Renyi_new}.
\end{proof}

\begin{rem}
Note that the monotonicity of $\alpha\mapsto \oll{D}_{\alpha}^{\meas}$
follows immediately by definition from the easily verifiable monotonicity in the classical case. Once the 
equalities in 
\eqref{reg measured} are established, this gives the monotonicity of 
$\alpha\mapsto D_{\alpha}\nw$ on $[1/2,+\infty)$. The latter was proved earlier in 
\cite{Renyi_new} by different methods.
\end{rem}

\begin{rem}\label{rem:Umegaki cont}
The continuity of $D_1\nw=\DU$ on 
$\bz\B(\hil)\times\B(\hil)\jz_{\lambda}$ follows also from 
\cite[Theorem 3.7]{Araki_relentr}, 
where its extension to the 
general von Neumann algebra setting was proved. 
See also \cite{Shirokov2022} for a different proof in the infinite-dimensional Hilbert space setting. 

Note that Proposition \ref{prop:regmeas cont} above gives that 
$\DU$ is continuous on $(\B(\hil)\times\B(\hil))_{\lambda,\pwr}$ for any 
$\lambda>1$ and $\pwr\in(0,+\infty)$, which in the case 
$\pwr\in(0,1)$ is strictly stronger than continuity on 
$\bz\B(\hil)\times\B(\hil)\jz_{\lambda}$ (see also Lemma \ref{lemma:B S cont} in this respect).
\end{rem}

\begin{rem}
In Appendix \ref{sec:variational} we give a different proof of the 
continuity of $\oll{D}_{\alpha}^{\meas}$ on $\bz\B(\hil)\times\B(\hil)\jz_{\lambda}$,
$\lambda>0$, 
for $\alpha\in(1,2]$.
\end{rem}

\begin{prop}\label{prop:Petz cont}
Let $\hil$ be a finite-dimensional Hilbert space, and let $f$ be a $\pwr$-function.
\begin{enumerate}
\item\label{Petz cont1}
The functions 
\begin{align}
&(0,1)\times\B(\hil)\pne\times\B(\hil)\pne\ni(\alpha,\rho,\sigma)\mapsto
D_{\alpha,1}(\rho\|\sigma),\label{Petz jointcont}\\
&(1,\pwr+1)\times\bz\B(\hil)\times\B(\hil)\jz_{f}\ni(\alpha,\rho,\sigma)\mapsto
D_{\alpha,1}(\rho\|\sigma)\label{Petz jointcont2}
\end{align}
are continuous.
\item\label{Petz cont2}
For every $\alpha\in(0,1)$,
\begin{align}\label{Petz contt1}
D_{\alpha,1} \ds\text{ is continuous on  }\ds\B(\hil)\pne\times\B(\hil)\pne,
\end{align}
and for every $\alpha\in[1,\pwr+1)$,
\begin{align}\label{Petz contt2}
D_{\alpha,1} \ds\text{ is continuous on  }\ds
\bz\B(\hil)\times\B(\hil)\jz_{f}.
\end{align}
\item\label{Petz cont3}
For any $\rho,\sigma\in \B(\hil)\pne$,
$\alpha\mapsto\psi_{\alpha,1}(\rho\|\sigma)$ is convex on $(0,+\infty)$, and 
$\alpha\mapsto D_{\alpha,1}(\rho\|\sigma)$ is monotone increasing and continuous
 on $(0,+\infty)$.
\end{enumerate}
\end{prop}
\begin{proof}
\ref{Petz cont1}\s 
is immediate from Lemma \ref{lemma:0-1 cont} and 
Corollary \ref{cor:az-continuity}.

\ref{Petz cont2}\s
is immediate from \ref{Petz cont1} for $\alpha\in(0,1)$ and $\alpha\in(1,\pwr+1)$, 
and for $\alpha=1$ from the continuity of the Umegaki relative entropy, i.e., the case 
$\alpha=1$ in \eqref{regmeas cont}.

\ref{Petz cont3}\s
The convexity of $\alpha\mapsto\psi_{\alpha,1}(\rho\|\sigma)$ is well known, and straightforward to verify by computing its second derivative, 
and the monotonicity of 
$\alpha\mapsto D_{\alpha,1}(\rho\|\sigma)$ follows immediately, 
as it has already been noted, e.g., in \cite[Lemma II.2]{MH}.
\end{proof}

In particular, Lemma \ref{lemma:0-1 cont}, 
Corollary \ref{cor:discont examples},
Theorem \ref{thm:az cont parameters}, and 
Propositions \ref{prop:regmeas cont} and \ref{prop:Petz cont} yield the following:

\begin{cor}\label{cor:sand Petz cont}
Let $\alpha,\lambda,\pwr\in(0,+\infty)$.
\begin{enumerate}
\item
The sandwiched R\'enyi divergences
$D_\alpha^*=D_{\alpha,\alpha}$ 
are continuous on 
$(B(\cH)\times B(\cH))_{\lambda,\pwr}$ if and only if 
$(1-\pwr)\alpha<1$.

\item
The Petz-type R\'enyi divergences
$D_{\alpha,1}$ are continuous
on $(B(\cH)\times B(\cH))_{\lambda,\pwr}$
if and only if $\alpha\in(0,\pwr+1)$.
\end{enumerate}
\end{cor}

\subsection{Continuity of $\alpha\mapsto D_{\alpha,z(\alpha)}(\rho\|\sigma)$}
\label{sec:path cont}

The following lemma complements Theorem \ref{thm:az cont parameters} on
the continuity of quantum R\'enyi $(\alpha,z)$-divergences in the parameters $\alpha,z$,
for fixed arguments $\rho$ and $\sigma$.

\begin{prop}\label{prop:a-z cont in a}
Let $J\subseteq(0,+\infty)$ be an interval and $J\ni\alpha\mapsto z(\alpha)\in(0,+\infty)$ be a continuous 
function on it. Then $J\ni\alpha\mapsto D_{\alpha,z(\alpha)}(\rho\|\sigma)$ is continuous on $J$ for any $\rho,\sigma\in\B(\hil)\pne$.
In particular, if $1\in J$ then
\begin{align}\label{alpha=1 limit}
\lim_{\alpha\to 1}D_{\alpha,z(\alpha)}(\rho\|\sigma)
=
D_{1,z(1)}(\rho\|\sigma)=
\DU(\rho\|\sigma).
\end{align}
Moreover, \eqref{alpha=1 limit} holds also under the weaker condition that 
$z(\alpha)\in(0,+\infty]$ and 
$\liminf_{\alpha\to 1}z(\alpha)>0$. 
\end{prop}
\begin{proof}
If $J$ contains only one point then the assertion is trivial, and hence for the rest we assume that $J$ is non-degenerate.
Let $\rho,\sigma\in\B(\hil)\pne$ be fixed.

Continuity on $J\cap(0,1)$ is obvious from the definition \eqref{Q alpha z def}, and similarly for the continuity on $J\cap(1,+\infty)$ when $\rho^0\le\sigma^0$. 
When $\rho^0\nleq\sigma^0$ then $D_{\alpha,z(\alpha)}\equiv +\infty$ is continuous on 
$(1,+\infty)$. Hence, we only need to consider continuity at $1$ in the case when $1\in J$. 

First we prove that 
\begin{align}\label{alpha=1 limit3}
\lim_{\alpha\to 1}D_{\alpha,z}(\rho\|\sigma)
=
\DU(\rho\|\sigma),\ds\ds\ds z\in(0,+\infty].
\end{align}
Assume first that $z\in(0,+\infty)$.
For $\alpha<1$ we have
\begin{align}\label{alpha=1 limit1}
D_{\alpha,z}(\rho\|\sigma)=
\frac{\log\Tr\bz\sigma^{\frac{1-\alpha}{2z}}\rho^{\frac{\alpha}{z}}\sigma^{\frac{1-\alpha}{2z}}\jz^z
-\log\Tr\bz\sigma^0\rho^{\frac{1}{z}}\sigma^0\jz^z}{\alpha-1}
+\frac{\log\Tr\bz\sigma^0\rho^{\frac{1}{z}}\sigma^0\jz^z-\log\Tr\rho}{\alpha-1},
\end{align}
and the same is true for $\alpha>1$ when $\rho^0\le\sigma^0$.
The limit of the first term above at $1$ is just the derivative of 
$\alpha\mapsto\log\Tr\bz\sigma^{\frac{1-\alpha}{2z}}\rho^{\frac{\alpha}{z}}\sigma^{\frac{1-\alpha}{2z}}\jz^z$
at $1$, which is easily seen to be 
\begin{align}\label{alpha=1 limit2}
&\frac{d}{d\alpha}\log\Tr\bz\sigma^{\frac{1-\alpha}{2z}}\rho^{\frac{\alpha}{z}}\sigma^{\frac{1-\alpha}{2z}}\jz^z\Big\vert_{\alpha=1}\nn\\
&\ds=
\frac{1}{\Tr\bz\sigma^0\rho^{\frac{1}{z}}\sigma^0\jz^z}
\left[\Tr\bz\bz\sigma^0\rho^{\frac{1}{z}}\sigma^0\jz^{z-1}\rho^{\frac{1}{z}}\log\rho\jz
-\Tr\bz\bz\sigma^0\rho^{\frac{1}{z}}\sigma^0\jz^{z}\log\sigma\jz
\right].
\end{align}
If $\rho^0\le\sigma^0$ then this is equal to $\DU(\rho\|\sigma)$, and the second term in 
\eqref{alpha=1 limit1} is zero for every $\alpha\in(0,1)\cup(1,+\infty)$, proving the assertion. 
(This case also follows from \cite[Proposition 3]{LinTomamichel15}.)
If $\rho^0\nleq \sigma^0$ then the expression in \eqref{alpha=1 limit2} is finite, while the second term in 
\eqref{alpha=1 limit1} goes to $+\infty=\DU(\rho\|\sigma)$ as $\alpha\nearrow 1$, according to Lemma \ref{lemma:power inequality}. Since $\lim_{\alpha\searrow 1}D_{\alpha,z}(\rho\|\sigma)=\lim_{\alpha\searrow 1}+\infty=+\infty$ is obvious by definition, the proof is complete also in this case.
The case $z=+\infty$ follows immediately from the above and 
\eqref{Umegaki bound on az}.

Consider finally the case where we only assume that 
\begin{align*}
0<x:=\liminf_{\alpha\to 1}z(\alpha).
\end{align*}
Then
\begin{align}
\DU(\rho\|\sigma)
&=
\lim_{\alpha\nearrow 1}D_{\alpha,x/2}(\rho\|\sigma\label{alpha=1 limit4})\\
&\le
\liminf_{\alpha\nearrow 1}D_{\alpha,z(\alpha)}(\rho\|\sigma)\label{alpha=1 limit5}\\
&\le
\limsup_{\alpha\nearrow 1}D_{\alpha,z(\alpha)}(\rho\|\sigma)\label{alpha=1 limit6}\\
&\le
\lim_{\alpha\nearrow 1}D_{\alpha,+\infty}(\rho\|\sigma)\label{alpha=1 limit7}\\
&=
\DU(\rho\|\sigma),
\end{align}
where the first and the last equalities follow from \eqref{alpha=1 limit3},
and the first and the last inequalities follow from 
\eqref{ALT3} and \eqref{Umegaki bound on az}. Thus, all the expressions in \eqref{alpha=1 limit4}--\eqref{alpha=1 limit7}
are equal to each other. Equality for the analogous quantities with 
$\lim_{\alpha\searrow 1}$ follows the same way, with all inequalities reversed, and 
the liminf and limsup interchanged.
\end{proof}

\begin{rem}
The equality of the limit and the Umegaki relative entropy in
\eqref{alpha=1 limit} was also proved in \cite{LinTomamichel15} under the stronger assumption that 
$\rho^0\le\sigma^0$ and $\alpha\mapsto z(\alpha)$ is continuously differentiable in a neighbourhood of $1$.
\end{rem}

As the following simple example demonstrates, the condition 
$\liminf_{\alpha\to 1}z(\alpha)>0$ in Proposition 
\ref{prop:a-z cont in a} is optimal:

\begin{prop}\label{prop:zero Renyi pure}
Let $\sigma\in\B(\hil)\pne$ be a positive semi-definite operator and 
$\psi\in\ran\sigma$ be a unit 
vector that is not an eigenvector of $\sigma$. 
Let $P_s$ denote the spectral projection of $\sigma$ corresponding to the singleton
$\{s\}$, and let $s_{\min}:=\min\{s>0:\,P_s\psi\ne 0\}$, $s_{\max}:=\max\{s>0:\,P_s\psi\ne 0\}$.
Then, for every $\alpha_1\in(0,1)$ and $\alpha_2\in(1,+\infty)$,
\begin{align}
&\log s_{\max}\inv=D_{\alpha_1,0}(\pr{\psi}\|\sigma)
\label{zero Renyi pure 1}\\
&\ds<\sum_{s}\norm{P_s\psi}^2\log s\inv=\DU(\pr{\psi}\|\sigma)
\label{zero Renyi pure 2}\\
&\ds<\log\sum_{s}\norm{P_s\psi}^2 s\inv=D_{\max}(\pr{\psi}\|\sigma)
\label{zero Renyi pure 3}\\
&\ds<\log s_{\min}\inv=
D_{\alpha_2,0}(\pr{\psi}\|\sigma).
\label{zero Renyi pure 4}
\end{align}
In particular, 
\begin{align}
\lim_{\alpha\nearrow 1}D_{\alpha,0}(\pr{\psi}\|\sigma)
<\DU(\pr{\psi}\|\sigma)<D_{\max}(\pr{\psi}\|\sigma)
<\lim_{\alpha\searrow 1}D_{\alpha,0}(\pr{\psi}\|\sigma).
\end{align}
\end{prop}
\begin{proof}
The equalities in \eqref{zero Renyi pure 1}--\eqref{zero Renyi pure 4} follow by straightforward computations, which we omit. 
The assumption that $\psi$ is not an eigenvector of $\sigma$ guarantees the existence of 
$0<s_1< s_2$ such that $\norm{P_{s_1}\psi}>0$, $\norm{P_{s_2}\psi}>0$, 
which yields the strict inequalities in 
\eqref{zero Renyi pure 1}--\eqref{zero Renyi pure 4}
(using also the strict concavity of $\log$ for \eqref{zero Renyi pure 3}).
\end{proof}

As it turns out, 
\begin{align}\label{zero Renyi limit1}
\lim_{\alpha\nearrow 1}D_{\alpha,0}(\rho\|\sigma)
<\DU(\rho\|\sigma)
<\lim_{\alpha\searrow 1}D_{\alpha,0}(\rho\|\sigma).
\end{align}
is a general phenomenon, in the sense that it holds for all 
$(\rho,\sigma)$ pairs in an open dense set in $\B(\hil)\pne\times\B(\hil)\pne$
for any finite-dimensional Hilbert space $\hil$ with $\dim\hil\ge 2$. 
In particular,
in the qubit case we have the following:
\begin{prop}\label{prop:zero Renyi limit}
Let $\rho,\sigma\in\B(\bC^2)$ be non-commuting qubit states such that 
$\rho^0\le\sigma^0$. Then 
\begin{align*}
\lim_{\alpha\nearrow 1}D_{\alpha,0}(\rho\|\sigma)
< \DU(\rho\|\sigma)
< \lim_{\alpha\searrow 1}D_{\alpha,0}(\rho\|\sigma).
\end{align*}
\end{prop}

We give the proof of Proposition \ref{prop:zero Renyi limit}
together with a more detailed analysis of the inequalities in 
\eqref{zero Renyi limit1} in Appendix \ref{sec:zero Renyi limit}.

\subsection{Boundedness by $D_{\max}$}
\label{sec:Dmax bound}

In Section \ref{sec:chdiv cont},
we will need that various quantum R\'enyi $\alpha$-divergences are dominated by 
$D_{\max}$.
We start with the following observation.

\begin{lemma}\label{lemma:Dmax bound}
Let $\hil$ be a finite-dimensional Hilbert space and $D_{\alpha}^q$ be a quantum 
R\'enyi $\alpha$-divergence 
with the monotonicity property 
$D_{\alpha}^q(\N(p)\|\N(q))\le D_{\alpha}^{\cl}(p\|q)$ for any 
$p,q\in\ell^{\infty}([\dim\hil])\pne$
and (completely) positive 
trace-preserving map $\N:\,\ell^{\infty}([\dim\hil])\to\B(\hil)$. Then 
\begin{align}\label{Dmax bound1}
D_{\alpha}^q(\rho\|\sigma)\le D_{\max}(\rho\|\sigma),
\ds\ds\ds \rho,\sigma\in\B(\hil)\pne.
\end{align}
If the above monotonicity property holds for every finite-dimensional Hilbert space $\hil$ 
(in particular, if $D_{\alpha}^q$ is monotone under CPTP maps), then 
\begin{align}\label{Dmax bound2}
D_{\alpha}^q\le D_{\max}.
\end{align}
\end{lemma}
\begin{proof}
Let $\rho,\sigma\in\B(\hil)\pne$.
If $D_{\max}(\rho\|\sigma)=+\infty$ then the inequality in 
\eqref{Dmax bound1} holds trivially, and hence for the rest we assume $D_{\max}(\rho\|\sigma)<+\infty$, or equivalently, 
$\rho^0\le\sigma^0$. In this case there exists a reverse test $(p,q,\Gamma)$
with $p,q\in\ell^{\infty}([\dim\hil])\pne$, $\Gamma(p)=\rho$, $\Gamma(q)=\sigma$, 
and $D_{\max}(p\|q)=D_{\max}(\rho\|\sigma)$
\cite[Section 4.2]{Matsumoto_newfdiv} (see also \cite[Section 4.2.3]{TomamichelBook}).
Thus, 
\begin{align*}
D_{\alpha}^q(\rho\|\sigma)=D_{\alpha}^q(\Gamma(p)\|\Gamma(q))
\le D_{\alpha}^{\cl}(p\|q)\le D_{\max}(p\|q)
=
D_{\max}(\rho\|\sigma),
\end{align*}
where the first inequality follows by assumption, and the second inequality
is straightforward to verify. This proves \eqref{Dmax bound1}, and 
\eqref{Dmax bound2} follows immediately from this under the given assumption.
\end{proof}

\begin{rem}
It would be interesting to know whether there exist quantum R\'enyi $\alpha$-divergences
that are monotone under classical--to--quantum channels as in Lemma \ref{lemma:Dmax bound},
but not monotone under general CPTP maps.
\end{rem}

\begin{rem}
Note that $\rho_{c,\ep}$ and $\sigma_{c,\ep}$ in
Lemma \ref{lemma:a>1 discont} 
do not commute, and hence, by \cite{Hiai-ALT},
if $\alpha>1$ and $z\in(0,\alpha-1)$ then 
\begin{align*}
D_{\alpha,z}(\rho_{c,\ep}\|\sigma_{c,\ep})>
D_{\alpha,\alpha-1}(\rho_{c,\ep}\|\sigma_{c,\ep})=
D_{\max}(\rho_{c,\ep}\|\sigma_{c,\ep}).
\end{align*}
Hence, by Lemma \ref{lemma:Dmax bound}, $D_{\alpha,z}$ is not monotone under CPTP maps
(and hence $Q_{\alpha,z}$ is not jointly convex in its arguments), for 
$\alpha>1$ and $z\in(0,\alpha-1)$. This was originally proved in more generality in 
\cite[Proposition 5.4]{Hiai_concavity2013} by different techniques. 
\end{rem}

Recall that for $\alpha>1$, $D_{\alpha,z}$ is monotone under CPTP maps if and only if 
$\max\{\alpha/2,\alpha-1\}\le z\le \alpha$ \cite{Hiai_concavity2013,Zhang2018}. 
Proposition \ref{prop:Dmax bound} below gives a
different, direct proof of the bound \eqref{Dmax bound2} in Lemma \ref{lemma:Dmax bound} 
for these $(\alpha,z)$ pairs, and also for a larger set of $(\alpha,z)$ pairs, 
for which $D_{\alpha,z}$ need not be monotone under CPTP maps.

Note that by Proposition \ref{prop:a>2 discont}, 
$D_{\alpha,z}\nleq D_{\max}$ on $\S(\hil)$ whenever $\dim\hil\ge 2$ and 
$\alpha>2$, $z\in(0,\alpha-1)$. The following extends this to a larger set of 
$(\alpha,z)$ pairs:

\begin{lemma}\label{lemma:no Dmax bound}
Let $\sigma\in\B(\hil)\p$ be a positive semi-definite operator and $\psi\in\hil$ be a unit vector 
that is not an eigenvector of $\sigma$. Then 
\begin{align*}
D_{\alpha,z}(\pr{\psi}\|\sigma)> D_{\max}(\pr{\psi}\|\sigma)\ds\ds\text{for any}\ds\ds
\alpha\in(1,+\infty),\s z\in(0,\alpha-1).
\end{align*}
\end{lemma}
\begin{proof}
Let $P_s$ denote the spectral projection of $\sigma$ corresponding to the singleton 
$\{s\}$. 
By definition,
\begin{align*}
D_{\alpha,z}(\pr{\psi}\|\sigma)
=
\log\bz\Tr\bz\pr{\psi}\sigma^{\frac{1-\alpha}{z}}\pr{\psi}\jz^z\jz^{\frac{1}{\alpha-1}}
=
\log\bz\sum_{s>0}s^{\frac{1-\alpha}{z}}\inner{\psi}{P_s\psi}\jz^{\frac{z}{\alpha-1}}.
\end{align*}
By the assumption that $\psi$ is not an eigenvector of $\sigma$, we get that 
$\inner{\psi}{P_s\psi}>0$ for at least two different values of $s$, 
and the assumption that $z\in(0,\alpha-1)$ yields that 
$\id_{[0,+\infty)}^{\frac{z}{\alpha-1}}$ is strictly concave. 
Hence, we get that 
\begin{align*}
D_{\alpha,z}(\pr{\psi}\|\sigma)
>
\log\sum_{s>0}s\inv\inner{\psi}{P_s\psi}
=
\log\snorm{\sigma^{-1/2}\psi}^2
=
D_{\max}(\pr{\psi}\|\sigma).
\end{align*}
\end{proof}

\begin{prop}\label{prop:Dmax bound}
Let $(\alpha,z)\in(0,+\infty)\times[0,+\infty]$ be such that $(\alpha,z)\ne (1,0)$.
The following are equivalent:
\begin{enumerate}
\item\label{Dmax boundd1}
$D_{\alpha,z}\le D_{\max}$;
\item\label{Dmax boundd2}
$D_{\alpha,z}\le D_{\max}$ on $\S(\hil)$ for some 
finite-dimensional Hilbert space with $\dim\hil\ge 2$;
\item\label{Dmax boundd3}
$\alpha\in(0,1)$ and $z\in[0,+\infty]$, 
or $\alpha=1$ and $z\in(0,+\infty]$, or 
$\alpha\in(1,+\infty)$ and $z\in[\alpha-1,+\infty]$.
\end{enumerate}
\end{prop}
\begin{proof}
\ref{Dmax boundd1}$\imp$\ref{Dmax boundd2} is obvious, and 
\ref{Dmax boundd2}$\imp$\ref{Dmax boundd3} is immediate from Lemma \ref{lemma:no Dmax bound}
(note that $D_{1,0}$ is not defined).

For
\ref{Dmax boundd3}$\imp$\ref{Dmax boundd1},
we need to prove that for any given $\rho,\sigma\in\B(\hil)\pne$, and any 
pair $(\alpha,z)$ as in \ref{Dmax boundd3}, the inequality 
$D_{\alpha,z}(\rho\|\sigma)\le D_{\max}(\rho\|\sigma)$ holds.
If $D_{\max}(\rho\|\sigma)=+\infty$ then there is nothing to prove, and hence for the rest we assume the contrary. 
By the scaling law \eqref{scaling}, we may assume that 
$\Tr\rho=\Tr\sigma=1$.
Assume that $\rho\le\lambda\sigma$ for some $\lambda>0$, or equivalently, 
$\lambda\inv\rho+\ep I\le\sigma+\ep I$, $\ep>0$. 
Consider first the case 
$\alpha\in(1,+\infty)$ and $z\in(\alpha-1,+\infty)$.
Then 
$\id_{(0,+\infty)}^{\frac{1-\alpha}{z}}$ is operator monotone decreasing
\cite{Bhatia}, whence
$(\sigma+\ep I)^{\frac{1-\alpha}{z}}\le(\lambda\inv\rho+\ep I)^{\frac{1-\alpha}{z}}$, and therefore, by \eqref{trace monotonicity},
\begin{align*}
\Tr\bz\rho^{\frac{\alpha}{2z}}(\sigma+\ep I)^{\frac{1-\alpha}{z}}\rho^{\frac{\alpha}{2z}}\jz^z
\le
\Tr\bz\rho^{\frac{\alpha}{2z}}(\lambda\inv\rho+\ep I)^{\frac{1-\alpha}{z}}\rho^{\frac{\alpha}{2z}}\jz^z,\ds\ds\ep>0.
\end{align*}
Since the LHS above goes to $Q_{\alpha,z}(\rho\|\sigma)$, and the RHS to $\lambda^{\alpha-1}$ as 
$\ep\searrow 0$, we get $D_{\alpha,z}(\rho\|\sigma)\le\log\lambda$. 
Taking the 
infimum over all $\lambda>0$ as above gives 
$D_{\alpha,z}(\rho\|\sigma)\le D_{\max}(\rho\|\sigma)$.
The case $z=+\infty$ follows by taking the limit $z\to+\infty$. 
The inequality for the rest of the $(\alpha,z)$ pairs in 
\ref{Dmax boundd3} follows immediately from the above and \eqref{Umegaki bound on az}.
\end{proof}

\begin{rem}
The inequality $D_{1,1}(=\DU)\le D_{\max}$ was already given in 
\cite[Lemma 10]{Datta}, 
which was improved to $D_{2,1}\le D_{\max}$ in \cite[Lemma 7]{BuscemiDatta2010}.
\end{rem}

We close this section with a simple observation related to Lemma \ref{lemma:Dmax bound}.
For any $\alpha\in(0,+\infty)$ and $\rho,\sigma\in\B(\hil)\pne$, let 
\begin{align*}
D_{\alpha}^{\test}(\rho\|\sigma)&:=\max\{D_{\alpha}((\Tr T\rho,\Tr(I-T)\rho)\|(\Tr T\sigma,\Tr(I-T)\sigma):\,
T\in\B(\hil)\p,\,T\le I\}
\end{align*}
be the \ki{test-measured R\'enyi $\alpha$-divergence}
of $\rho$ and $\sigma$; see \cite{MH-testdiv}. (Note that $D_{\alpha}^{\test}$ is not a quantum R\'enyi divergence in the sense that it need not satisfy \eqref{qRenyi def}.)
According to \cite[Proposition III.32]{MH-testdiv}, 
\begin{align}\label{testmeas infty limit}
\lim_{\alpha\to+\infty}D_{\alpha}^{\test}(\rho\|\sigma)=D_{\max}(\rho\|\sigma)
\end{align}
for any $\rho,\sigma\in\B(\hil)\pne$.
Lemma \ref{lemma:Dmax bound} combined with \eqref{testmeas infty limit} yields  the following:

\begin{cor}
Let $c\in(0,+\infty)$, and for every $\alpha\in(c,+\infty)$, let $\divv_{\alpha}$ be a 
quantum divergence. If
$D_{\alpha}^{\test}(\rho\|\sigma)\le \divv_{\alpha}(\rho\|\sigma)\le D_{\max}(\rho\|\sigma)$
for some $\rho,\sigma\in\B(\hil)\pne$ and every $\alpha\in(c,+\infty)$, then 
\begin{align}\label{alpha infty limit}
\lim_{\alpha\to+\infty}\divv_{\alpha}(\rho\|\sigma)=D_{\max}(\rho\|\sigma).
\end{align}
In particular, this holds when for every $\alpha\in(c,+\infty)$, $\divv_{\alpha}$ is a quantum R\'enyi $\alpha$-divergence that is monotone under CPTP maps. 
\end{cor}
\begin{proof}
Immediate from Lemma \ref{lemma:Dmax bound}, \eqref{testmeas infty limit}, and the facts that 
for any quantum R\'enyi $\alpha$-divergence $D_{\alpha}^q$ that is monotone under CPTP maps, 
$D_{\alpha}^{\test}(\rho\|\sigma)\le D_{\alpha}^{q}(\rho\|\sigma)$
for any $\rho,\sigma\in\B(\hil)\pne$.
\end{proof}

\section{Continuity of quantum R\'enyi $\alpha$-divergences for CP maps}
\label{sec:chdiv cont}

In what follows, we will consider quantum divergences that satisfy the following support condition:
\begin{align}\label{support condition}
\divv(\rho\|\sigma)=+\infty\ds\iff\ds\rho^0\nleq\sigma^0.
\end{align}
This is true for every $\alpha>1$ and $D_{\alpha}^{\meas}$, $\oll{D}_{\alpha}^{\meas}$, and $D_{\alpha,z}$, $z>0$, and also for $D_1\nw=\DU$ and $D_1^{\meas}$.

For $\N_1,\N_2\in \CP^+(\hil,\kil)$, and for an arbitrary 
quantum divergence $\divv$, let 
\begin{align}\label{channel div def}
\divv(\N_1\|\N_2):=\sup\left\{\divv\bz(\id\otimes\N_1)\rho\|(\id\otimes\N_2)\rho\jz:\,
\rho\in\S(\bC^d\otimes\hil),\,d\in\bN\right\}
\end{align}
be the \ki{completely positive super-operator (CPSO) $\divv$-divergence} of $\N_1$ and $\N_2$.
Note that a priori there is no bound on the dimension $d$ of the auxiliary Hilbert space 
in \eqref{channel div def}. However, 
it is easy to see that if $\divv$ is monotone under partial trace, or it is jointly quasi-convex, then we have 
\begin{align}\label{channel div2}
\divv(\N_1\|\N_2)
&=
\sup\left\{\divv\bz(\id\otimes\N_1)\pr{\psi}\|(\id\otimes\N_2)\pr{\psi}\jz:\,
\psi\in\hil\otimes\hil,\,\norm{\psi}=1\right\}.
\end{align}
Indeed, this follows by taking either a purification or a convex decomposition of a general input $\rho$ and then using the Schmidt decomposition and isometric invariance;
see, e.g., Lemmas 6 and 7 in \cite{CMW} for more details.
If \eqref{channel div2} holds for some 
$\N_1,\N_2\in\CP^+(\hil,\kil)$, then obviously also
\begin{align}\label{channel div3}
\divv(\N_1\|\N_2)
&=
\sup\left\{\divv\bz(\id\otimes\N_1)\rho\|(\id\otimes\N_2)\rho\jz:\,
\rho\in\S(\hil\otimes\hil)\right\}.
\end{align}

Lemma \ref{lemma:Dmax bound} and Proposition \ref{prop:Dmax bound} yield immediately the following:
\begin{cor}
Let $D_{\alpha}^q$ be a quantum R\'enyi $\alpha$-divergence that is monotone under CPTP maps, 
or $D_{\alpha}^q=D_{\alpha,z}$ with 
some $(\alpha,z)$ as given in \ref{Dmax boundd3} in Proposition \ref{prop:Dmax bound}.
Then 
\begin{align}\label{ch maxrelentr bound}
D_{\alpha}^q(\N_1\|\N_2)\le D_{\max}(\N_1\|\N_2)
\end{align}
for any $\N_1,\N_2\in\CP^+(\hil,\kil)$. 
\end{cor}

\begin{rem}
Note that the \cpso max-relative entropy 
is the max-relative entropy of the Choi matrices of the \cpsos
(see, e.g., \cite[Lemma 12]{Amortized2020}), whence it
is additive under tensor products. Thus,
if $D_{\alpha}^q$ satisfies \eqref{ch maxrelentr bound} for any pair of \cpsos $\N_1,\N_2$, then it also satisfies
\begin{align}\label{ch maxrelentr bound2}
D_{\alpha}^{q,\reg}(\N_1\|\N_2):=\sup_{n\in\bN}\frac{1}{n}
D_{\alpha}^q(\N_1^{\otimes n}\|\N_2^{\otimes n})\le D_{\max}(\N_1\|\N_2)
\end{align}
for any pair of \cpsos $\N_1,\N_2$, where $D_{\alpha}^{q,\reg}(\N_1\|\N_2)$
is the \ki{regularized \cpso $D_{\alpha}^q$-divergence.}
For the 
sandwiched R\'enyi divergences $D_{\alpha}^q=D_{\alpha}\nw$ with $\alpha>1$, the above follows already from the combination of 
\cite[Proposition 10]{Amortized2020} and \cite[Theorem 5.4]{FF2020}.
\end{rem}

$D_{\max}(\N_1\|\N_2)$ was defined in \cite[Definition 19]{GFWRSCW} as 
$\log\inf\{\lambda:\,\lambda\N_2-\N_1\text{ is CP}\}$,
and its equality to the above used definition was pointed out in 
\cite[Remark 13]{Amortized2020}.
Most of the following lemma can be easily obtained from this and 
\cite[Lemma 12]{Amortized2020}.
We give a detailed proof for completeness.

\begin{lemma}\label{lemma:ch div finite}
Let $\N_1,\N_2\in\CP^+(\hil,\kil)$. Then the following are equivalent:
\begin{enumerate}
\item\label{finite1'}
$\divv(\N_1\|\N_2)<+\infty$ for every quantum divergence $\divv$ satisfying
$\divv\le D_{\max}$.
\item\label{finite1-1}
$D_{\max}(\N_1\|\N_2)<+\infty$.
\item\label{finite2}
$\divv(\N_1\|\N_2)<+\infty$ for some quantum divergence $\divv$ satisfying
\eqref{support condition}.
\item\label{finite3}
$\bz(\id\otimes\N_1)\rho\jz^0\le\bz(\id\otimes\N_2)\rho\jz^0$ for all 
$\rho\in\S(\hil\otimes\hil)$.
\item\label{finite4}
There exists a $\lambda>0$ such that $\N_1\cple\lambda\N_2$. 
\item\label{finite5}
There exists a $\lambda>0$ such that
\begin{align}
\S_{\N_1,\N_2}
&:=
\left\{\bz(\id\otimes\N_1)\rho,(\id\otimes\N_2)\rho\jz:\,\rho\in\S(\hil\otimes\hil)\right\}\nn\\
&\subseteq
\left\{(\sigma_1,\sigma_2)\in\B(\hil\otimes\hil)\pne:\,\sigma_1\le\lambda\sigma_2\right\}.
\label{channel pair image}
\end{align}
\end{enumerate}
Moreover, the infimum of all $\lambda$ for which 
\ref{finite4} holds is the same as
the infimum of all $\lambda$ for which 
\ref{finite5} holds, and is equal to 
$\exp(D_{\max}(\N_1\|\N_2))$.
\end{lemma}
\begin{proof}
The implications \ref{finite1'}$\iff$\ref{finite1-1}$\imp$\ref{finite2} are obvious, and 
\ref{finite2}$\imp$\ref{finite3} is also clear due to the support condition 
\eqref{support condition}.
Assume \ref{finite3}. Let $(e_i)_{i=1}^d$ be an orthonormal basis in $\hil$, and $\psi:=\frac{1}{\sqrt{d}}\sum_{i=1}^d e_i\otimes e_i$ be a maximally entangled pure state on 
$\hil\otimes\hil$. 
The choice $\rho=\pr{\psi}$ in \ref{finite3} yields that 
there exists a $\lambda>0$ such that
\begin{align*}
\frac{1}{d}C(\N_1)=(\id\otimes\N_1)\pr{\psi}\le \lambda(\id\otimes\N_2)\pr{\psi}=\lambda
\frac{1}{d}C(\N_2),
\end{align*}
where for a linear map $\map:\,\B(\hil)\to\B(\kil)$, 
$C(\map):=d(\id\otimes\map)\pr{\psi}$ is its \ki{Choi matrix} in the given basis
\cite{Choi1975}.
Thus, $C(\lambda\N_2-\N_1)\ge 0$, i.e., 
$\lambda\N_2-\N_1$ is completely positive, according to Choi's criterion
\cite{Choi1975}.
This proves \ref{finite3}$\imp$\ref{finite4}.
Using again the Choi matrix, it is obvious that 
\ref{finite4} holds with some $\lambda>0$ if and only if 
\ref{finite5} holds with the same $\lambda$, if and only if 
$D_{\max}(\N_1\|\N_2)\le\log\lambda$, proving 
\ref{finite4}$\iff$\ref{finite5}$\imp$\ref{finite1-1}, and also the last assertion 
about the optimal $\lambda$.
\end{proof}

\begin{prop}\label{prop:channeldiv attainability}
Let $D_{\alpha}^q$ be a quantum R\'enyi $\alpha$-divergence 
for some $\alpha\in(0,+\infty)$ such that 
\begin{enumerate}
\item\label{channeldiv attainability1}
$D_{\alpha}^q$ is monotone under partial traces, or it is jointly quasi-convex in its arguments;
\item\label{channeldiv attainability2}
$D_{\alpha}^q$ is continuous on $(\B(\hil)\times\B(\hil))_{\lambda}$ for any 
$\lambda>0$ and any finite-dimensional Hilbert space $\hil$. 
\end{enumerate}
Then for every $\N_1,\N_2\in\CP^+(\hil,\kil)$ such that $D_{\max}(\N_1\|\N_2)<+\infty$, there exists a unit vector $\psi\in\hil\otimes\hil$ attaining the \cpso $D_{\alpha}^q$-divergence, i.e., 
\begin{align}\label{channeldiv attainability3}
D_{\alpha}^q(\N_1\|\N_2)
=D_{\alpha}^q\bz(\id\otimes\N_1)\pr{\psi}\|(\id\otimes\N_2)\pr{\psi}\jz.
\end{align}
Moreover, for any $\lambda>0$, 
and any finite-dimensional Hilbert spaces $\hil,\kil$, 
$(\N_1,\N_2)\mapsto D_{\alpha}^q(\N_1\|\N_2)$ is continuous on 
\begin{align*}
\CP(\hil,\kil)^2_{\lambda}:=
\{(\N_1,\N_2)\in\CP^+(\hil,\kil)^2:\,\N_1\cple\lambda\N_2\}.
\end{align*}
\end{prop}
\begin{proof}
Let $\N_1,\N_2\in\CP^+(\hil,\kil)$ be such that $D_{\max}(\N_1\|\N_2)<+\infty$.
By Lemma
\ref{lemma:ch div finite}, we have \eqref{channel pair image} with 
$\log\lambda=D_{\max}(\N_1\|\N_2)$. Thus, by assumption \ref{channeldiv attainability2},
$D_{\alpha}^q$ is continuous on the compact set 
$\{((\id\otimes\N_1)\pr{\psi},(\id\otimes\N_2)\pr{\psi}):\,
\psi\in\hil\otimes\hil,\,\norm{\psi}=1\}$, from which \eqref{channeldiv attainability3} follows for some unit vector $\psi\in\hil\otimes\hil$
due to assumption \ref{channeldiv attainability1} and \eqref{channel div2}.

By assumption \ref{channeldiv attainability2} and Lemma \ref{lemma:ch div finite},
\begin{align*}
\CP(\hil,\kil)^2_{\lambda}\times\S(\hil\otimes \hil)\ni(\N_1,\N_2,\rho)\mapsto
D_{\alpha}^q\bz(\id\otimes\N_1)\rho\|(\id\otimes\N_2)\rho\jz
\end{align*}
is continuous, and this yields immediately the asserted continuity of 
$(\N_1,\N_2)\mapsto D_{\alpha}^q(\N_1\|\N_2)$  
by \eqref{channel div3} and \ref{usc3} of Lemma \ref{lemma:usc}.
\end{proof}

\begin{rem}
By the trace-preserving property of channels, if $\N_1,\N_2\in\cptp(\hil,\kil)$
are such that 
$\N_1\le\lambda\N_2$ then $\lambda\ge 1$, and $\lambda=1$ implies $\N_1=\N_2$.
Hence, if we are only interested in channels then condition \ref{channeldiv attainability2} in 
Proposition \ref{prop:channeldiv attainability}
can be replaced with the a priori weaker requirement of 
$D_{\alpha}^q$ being continuous on $(\S(\hil)\times\S(\hil))_{\lambda}$ for any 
$\lambda\ge 1$ and any finite-dimensional Hilbert space $\hil$.
\end{rem}

\begin{example}\label{ex:chdiv cont in channels}
According to Proposition \ref{prop:meas cont}, Proposition \ref{prop:regmeas cont}, 
Lemma \ref{lemma:0-1 cont} and Corollary \ref{cor:az-continuity}, the
continuity condition \ref{channeldiv attainability2} in Proposition \ref{prop:channeldiv attainability}
hold for $D_{\alpha}^{\meas}$, $\alpha\in(0,+\infty)$,
$\oll{D}_{\alpha}^{\meas}$, $\alpha\in(0,+\infty)$, and
for $D_{\alpha,z}$
when $\alpha\in(0,1)$ and $z\ge\max\{\alpha,1-\alpha\}$, 
$\alpha\in(1,2)$ and $\alpha/2\le z\le\alpha$, or 
$\alpha\ge 2$ and $\alpha-1<z\le\alpha$.
All these R\'enyi divergences are monotone under CPTP  maps,
by definition in the case of the (regularized) measured R\'enyi divergences, and by 
\cite{Zhang2018} in the case of the $(\alpha,z)$-divergences. Hence, 
the conclusions of Proposition \ref{prop:channeldiv attainability} hold 
for all these quantum R\'enyi divergences. In particular, 
the \cpso Umegaki relative entropy is continuous on $\CP(\hil,\kil)^2_{\lambda}$.
\end{example}
\medskip

The continuity results in Example \ref{ex:chdiv cont in channels}
can be improved to 
joint continuity results in the R\'enyi parameters and the \cpsos,
as follows.

\begin{prop}\label{prop:channel joint cont}
Let $\hil,\kil$ be a finite-dimensional Hilbert spaces, and $\lambda>0$. The 
following functions are continuous:
\begin{align}
&(0,1)\times\CP^+(\hil,\kil)^2\ni(\alpha,\N_1,\N_2)\mapsto D_{\alpha}^{\meas}(\N_1\|\N_2),
\label{chdiv jointcont1}\\
&(1,+\infty)\times\CP(\hil,\kil)_{\lambda}^2\ni(\alpha,\N_1,\N_2)\mapsto D_{\alpha}^{\meas}(\N_1\|\N_2),
\label{chdiv jointcont4}\\
&(0,1)\times\CP^+(\hil,\kil)^2\ni(\alpha,\N_1,\N_2)\mapsto\oll{D}_{\alpha}^{\meas}(\N_1\|\N_2),
\label{chdiv jointcont2}\\
&(1,+\infty)\times\CP(\hil,\kil)_{\lambda}^2\ni(\alpha,\N_1,\N_2)\mapsto \oll{D}_{\alpha}^{\meas}(\N_1\|\N_2),
\label{chdiv jointcont5}\\
&\left\{\alpha\in(0,1),\,z\ge\max\{\alpha,1-\alpha\}\right\}\times\CP^+(\hil,\kil)^2
\ni(\alpha,z,\N_1,\N_2)\mapsto D_{\alpha,z}(\N_1\|\N_2),\label{chdiv jointcont3}\\
&\D\times\CP(\hil,\kil)_{\lambda}^2\ni(\alpha,z,\N_1,\N_2)\mapsto D_{\alpha,z}(\N_1\|\N_2),
\label{chdiv jointcont6}
\end{align}
where $\D:=\{(\alpha,z)\in(1,2)\times(1/2,+\infty):\,\alpha/2\le z\le\alpha\}
\cup\{(\alpha,z)\in[2,+\infty)\times(1,+\infty):\,\alpha-1<z\le\alpha\}$.
\end{prop}
\begin{proof}
First, note that for each function, the given R\'enyi divergences are monotone under CPTP maps; this is obvious for the (regularized) measured R\'enyi divergences, and for the
$(\alpha,z)$-divergences in the given parameter ranges it follows from 
\cite{Zhang2018}.
Hence, the corresponding \cpso R\'enyi divergences can be written as in 
\eqref{channel div3}.
By Corollary \ref{cor:az-continuity}, 
the function 
\begin{align*}
\D\times\CP(\hil,\kil)_{\lambda}^2\times \S(\hil\otimes\hil)\ni(\alpha,z,\N_1,\N_2,\rho)\mapsto 
D_{\alpha,z}((\id\otimes\N_1)(\rho)\|(\id\otimes\N_2)(\rho))
\end{align*}
is continuous, and the continuity of \eqref{chdiv jointcont6}
follows from this by \ref{usc3} of Lemma \ref{lemma:usc}. 

The proofs of the continuity of \eqref{chdiv jointcont1}--\eqref{chdiv jointcont3} go exactly the same way, 
using Proposition \ref{prop:meas cont}, Proposition \ref{prop:regmeas cont}, or 
Lemma \ref{lemma:0-1 cont} in place of Corollary \ref{cor:az-continuity} in the above argument.
\end{proof}

The joint continuity results in Proposition \ref{prop:channel joint cont} immediately imply continuity of $D_{\alpha}^{\meas}(\N_1\|\N_2)$ and 
$\oll{D}_{\alpha}^{\meas}(\N_1\|\N_2)$ in $\alpha$
on $(0,1)\cup(1,+\infty)$, and of 
$D_{\alpha,z(\alpha)}$ along a continuous path $\alpha\mapsto z(\alpha)$ 
in the domain of $(\alpha,z)$ pairs considered in the proposition. 
What is obviously missing is continuity at $\alpha=1$, which is what we will consider below. 

First, note that by \eqref{ALT3} and \eqref{Umegaki bound on az},
\begin{align}\label{channel ALT}
D_{\alpha_1, z_1}(\N_1\|\N_2)
\le D_{\alpha_1,\hat z_1}(\N_1\|\N_2)
\le
\DU(\N_1\|\N_2)
\le D_{\alpha_2,\tilde z_2}(\N_1\|\N_2)
\le D_{\alpha_2,z_2}(\N_1\|\N_2)
\end{align}
for any $\N_1,\N_2\in\CP^+(\hil,\kil)$ and $\alpha_1<1<\alpha_2$, 
$z_1\le\hat z_1$, $z_2\le\tilde z_2$. 
This yields immediately the following simple observation:

\begin{lemma}\label{lemma:channel limit by domination}
Let $\N_1,\N_2\in\CP^+(\hil,\kil)$ and $J\subseteq(0,+\infty)$ be an interval that contains $1$.
Assume that $J\setminus\{1\}\ni\alpha\mapsto z(\alpha)\in[0,+\infty]$ is such that 
\begin{align*}
\lim_{\alpha\to 1}D_{\alpha,z(\alpha)}(\N_1\|\N_2)=\DU(\N_1\|\N_2).
\end{align*}
Then for any $J\setminus\{1\}\ni\alpha\mapsto \tilde z(\alpha)\in[z(\alpha),+\infty]$,
\begin{align*}
\lim_{\alpha\to 1}D_{\alpha,\tilde z(\alpha)}(\N_1\|\N_2)=\DU(\N_1\|\N_2).
\end{align*}
\end{lemma}

Next, we make some general observations
on the continuity of 
\cpso R\'enyi $\alpha$-divergences in the parameter $\alpha$.

\begin{lemma}\label{lemma:chdiv monotone in a}
Let $(D_{\alpha}^q)_{\alpha\in (\alpha_0,\alpha_1)}$ be quantum R\'enyi $\alpha$-divergences such that $\alpha\mapsto D_{\alpha}^q$ is monotone increasing. 
Then $D_{\alpha_1^-}^q:=\lim_{\alpha\nearrow \alpha_1}D_{\alpha}^q$ exists, and 
for any $\N_1,\N_2\in\CP^+(\hil,\kil)$,
\begin{align*}
D_{\alpha}^q(\N_1\|\N_2)\nearrow D_{\alpha_1^-}^q(\N_1\|\N_2)\ds\text{as}\ds
\alpha\nearrow \alpha_1.
\end{align*}
\end{lemma}
\begin{proof}
Monotonicity of $\alpha\mapsto D_{\alpha}^q(\rho\|\sigma)$ for any pair of non-zero PSD operators
$\rho,\sigma$ implies the monotonicity of $\alpha\mapsto D_{\alpha}^q(\N_1\|\N_2)$. 
Hence,
\begin{align*}
\lim_{\alpha\nearrow \alpha_1}D_{\alpha}^q(\N_1\|\N_2)
&=
\sup_{\alpha\in(\alpha_0,\alpha_1)}D_{\alpha}^q(\N_1\|\N_2)\\
&=
\sup_{\alpha\in(\alpha_0,\alpha_1)}
\sup_{d\in\bN}\sup_{\rho\in\S(\bC^d\otimes\hil)}
D_{\alpha}^q\bz(\id\otimes\N_1)\rho\|(\id\otimes\N_2)\rho\jz\\
&=
\sup_{d\in\bN}\sup_{\rho\in\S(\bC^d\otimes\hil)}
\underbrace{\sup_{\alpha\in(\alpha_0,\alpha_1)}
D_{\alpha}^q\bz(\id\otimes\N_1)\rho\|(\id\otimes\N_2)\rho\jz}_{
D_{\alpha_1^-}^q\bz(\id\otimes\N_1)\rho\|(\id\otimes\N_2)\rho\jz}\\
&=
D_{\alpha_1^-}^q(\N_1\|\N_2).
\end{align*}
\end{proof}

\begin{example}\label{ex:leftcont}
By Lemma \ref{lemma:chdiv monotone in a}, 
Proposition \ref{prop:meas cont}, Proposition \ref{prop:regmeas cont}, and 
Proposition \ref{prop:Petz cont}
for any $\N_1,\N_2\in\CP^+(\hil,\kil)$,
$\alpha\mapsto D_{\alpha}^q(\N_1\|\N_2)$ is monotone increasing and 
left continuous on $(0,+\infty)$,
where $D_{\alpha}^q$ stands for $D_{\alpha}^{\meas}$, $\oll{D}_{\alpha}^{\meas}$, or 
$D_{\alpha,1}$. In particular,
\begin{align}
\lim_{\alpha\nearrow 1}D_{\alpha}^{\meas}(\N_1\|\N_2)&=D_1^{\meas}(\N_1\|\N_2),
\label{left limit1}\\
\lim_{\alpha\nearrow 1}\oll{D}_{\alpha}^{\meas}(\N_1\|\N_2)&=\DU(\N_1\|\N_2),
\label{left limit2}\\
\lim_{\alpha\nearrow 1}D_{\alpha,1}(\N_1\|\N_2)&=\DU(\N_1\|\N_2).
\label{left limit3}
\end{align}
\end{example}

\begin{cor}\label{cor:az leftcont at1}
Let $c\in(0,1)$ and $(c,1)\ni\alpha\mapsto z(\alpha)\in[\alpha,+\infty]$.
Then, for any $\N_1,\N_2\in\CP^+(\hil,\kil)$,
\begin{align}
\lim_{\alpha\nearrow 1}D_{\alpha,z(\alpha)}(\N_1\|\N_2)=\DU(\N_1\|\N_2).
\label{az left limit}
\end{align}
If, moreover, $z\mapsto z(\alpha)$ is continuous, and $z(\alpha)\ge 
\max\{\alpha,1-\alpha\}$, then 
$\alpha\mapsto D_{\alpha,z(\alpha)}(\N_1\|\N_2)$ is also continuous on $(c,1)$. 
\end{cor}
\begin{proof}
The statement about the limit in \eqref{az left limit}
is obvious from 
\eqref{left limit2} and \eqref{channel ALT}.
The statement about the continuity is obvious from the continuity in 
\eqref{chdiv jointcont3}.
\end{proof}

The following is our main observation:

\begin{lemma}\label{lemma:chdiv rightcont}
Let $(D_{\alpha}^q)_{\alpha\in(\alpha_0,\alpha_1)}$ be quantum R\'enyi $\alpha$-divergences.
Assume that 
\begin{enumerate}
\item\label{chdiv upper limit1}
$\alpha\mapsto D_{\alpha}^q$ is monotone increasing;
\item\label{chdiv upper limit2}
for every $\alpha\in(\alpha_0,\alpha_1)$, $D_{\alpha}^q$ satisfies the support condition 
\eqref{support condition};
\item\label{chdiv upper limit3}
for every $\alpha\in(\alpha_0,\alpha_1)$, $D_{\alpha}^q$ is monotone under partial traces, 
or it is jointly quasi-convex;
\item\label{chdiv upper limit4}
for every $\alpha\in(\alpha_0,\alpha_1)$, every $\lambda>0$, and every finite-dimensional 
Hilbert space $\hil$, $D_{\alpha}^q$ is continuous on 
$(\B(\hil)\times\B(\hil))_{\lambda}$. 
\end{enumerate}
Then $D_{\alpha_0^+}^q:=\lim_{\alpha\searrow \alpha_0}D_{\alpha}^q$ exists, and 
if it also satisfies \eqref{support condition} then 
for any $\N_1,\N_2\in\CP^+(\hil,\kil)$,
\begin{align*}
D_{\alpha}^q(\N_1\|\N_2)\searrow D_{\alpha_0^+}^q(\N_1\|\N_2)\ds\text{as}\ds
\alpha\searrow \alpha_0.
\end{align*}
\end{lemma}
\begin{proof}
The existence of $D_{\alpha_0^+}^q$ is obvious from assumption \ref{chdiv upper limit1},
and assumption \ref{chdiv upper limit3} implies that 
$D_{\alpha_0^+}$ is also monotone under partial traces 
or it is quasi-convex. Hence, the \cpso divergence 
corresponding to $D_{\alpha_0^+}^q$ can be written as in \eqref{channel div3}.
Monotonicity of $\alpha\mapsto D_{\alpha}^q(\rho\|\sigma)$ for any pair of non-zero PSD operators
$\rho,\sigma$ implies the monotonicity of $\alpha\mapsto D_{\alpha}^q(\N_1\|\N_2)$, and also 
that 
\begin{align*}
\lim_{\alpha\searrow \alpha_0}D_{\alpha}^q(\N_1\|\N_2)
&=
\inf_{\alpha\in(\alpha_0,\alpha_1)}D_{\alpha}^q(\N_1\|\N_2)\\
&=
\inf_{\alpha\in(\alpha_0,\alpha_1)}\sup_{(\sigma_1,\sigma_2)\in\S_{\N_1,\N_2}}
D_{\alpha}^q(\sigma_1\|\sigma_2)\\
&\ge
\sup_{(\sigma_1,\sigma_2)\in\S_{\N_1,\N_2}}\inf_{\alpha\in(\alpha_0,\alpha_1)}
D_{\alpha}^q(\sigma_1\|\sigma_2)\\
&=
D_{\alpha_0^+}^q(\N_1\|\N_2),
\end{align*}
where $\S_{\N_1,\N_2}=\left\{(\id\otimes\N_1)\rho,(\id\otimes\N_2)\rho):\,\rho\in\S(\hil\otimes\hil)\right\}$ as in \eqref{channel pair image}.
The converse inequality is trivial when $D_{\alpha_0^+}^q(\N_1\|\N_2)=+\infty$, and hence for the 
rest we assume the contrary. 
By the assumption that $D_{\alpha_0^+}^q$ 
satisfies \eqref{support condition}, we get from Lemma \ref{lemma:ch div finite} the existence of some 
$\lambda>0$ such that 
\begin{align*}
\S_{\N_1,\N_2}
\subseteq
\left\{(\sigma_1,\sigma_2)\in\B(\hil\otimes\hil)\pne:\,\sigma_1\le\lambda\sigma_2\right\}.
\end{align*}
By assumption \ref{chdiv upper limit4} $D_{\alpha}^q$ is continuous on the compact set $\S_{\N_1,\N_2}$ for every $\alpha\in(\alpha_0,\alpha_1)$.
Hence, we may apply the minimax theorem from Lemma \ref{lemma:minimax2} to obtain 
the third equality below, and the rest of them are obvious by definition:
\begin{align*}
\lim_{\alpha\searrow \alpha_0}D_{\alpha}^q(\N_1\|\N_2)&=
\inf_{\alpha\in(\alpha_0,\alpha_1)}D_{\alpha}^q(\N_1\|\N_2)\\
&=
\inf_{\alpha\in(\alpha_0,\alpha_1)}\sup_{(\sigma_1,\sigma_2)\in\S_{\N_1,\N_2}}
D_{\alpha}^q(\sigma_1\|\sigma_2)\\
&=
\sup_{(\sigma_1,\sigma_2)\in\S_{\N_1,\N_2}}\inf_{\alpha\in(\alpha_0,\alpha_1)}
D_{\alpha}^q(\sigma_1\|\sigma_2)\\
&=
\sup_{(\sigma_1,\sigma_2)\in\S_{\N_1,\N_2}}
D_{\alpha_0^+}^q(\sigma_1\|\sigma_2)\\
&=
D_{\alpha_0^+}^q(\N_1\|\N_2).
\end{align*}
\end{proof}

\begin{example}\label{ex:rightcont}
By Lemma \ref{lemma:chdiv rightcont}, 
Proposition \ref{prop:meas cont}, Proposition \ref{prop:regmeas cont}, and 
Proposition \ref{prop:Petz cont},
for any $\N_1,\N_2\in\CP^+(\hil,\kil)$,
$\alpha\mapsto D_{\alpha}^{\meas}(\N_1\|\N_2)$ 
and
$\alpha\mapsto \oll{D}_{\alpha}^{\meas}(\N_1\|\N_2)$ 
are right continuous on $[1,+\infty)$, and
$\alpha\mapsto D_{\alpha,1}(\N_1\|\N_2)$ 
is right continuous on $[1,2)$.
In particular,
\begin{align}
\lim_{\alpha\searrow 1}D_{\alpha}^{\meas}(\N_1\|\N_2)&=D_1^{\meas}(\N_1\|\N_2),
\label{right limit1}\\
\lim_{\alpha\searrow 1}\oll{D}_{\alpha}^{\meas}(\N_1\|\N_2)&=\DU(\N_1\|\N_2),
\label{right limit2}\\
\lim_{\alpha\searrow 1}D_{\alpha,1}(\N_1\|\N_2)&=\DU(\N_1\|\N_2).
\label{right limit3}
\end{align}
\end{example}

\begin{cor}\label{cor:az rightcont at1}
Let $c\in(1,+\infty)$ and $(1,c)\ni\alpha\mapsto z(\alpha)\in[1,+\infty]$.
Then
\begin{align}
\lim_{\alpha\searrow 1}D_{\alpha,z(\alpha)}(\N_1\|\N_2)=\DU(\N_1\|\N_2).
\label{az right limit}
\end{align}
If, moreover, $z\mapsto z(\alpha)$ is continuous, and
for all $\alpha$, $z(\alpha)\in\D$ given in 
Proposition \ref{prop:channel joint cont}, then 
$\alpha\mapsto D_{\alpha,z(\alpha)}(\N_1\|\N_2)$ is also continuous on $(1,c)$. 
\end{cor}
\begin{proof}
The statement about the limit in \eqref{az right limit}
is obvious from 
\eqref{right limit3} and \eqref{channel ALT}.
The statement about the continuity is obvious from the continuity in 
\eqref{chdiv jointcont6}.
\end{proof}

Summarizing the above, 
we obtain the following:
\begin{theorem}
For any $\N_1,\N_2\in\CP^+(\hil,\kil)$, 
$\alpha\mapsto D_{\alpha}^{\meas}(\N_1\|\N_2)$ and
$\alpha\mapsto \oll{D}_{\alpha}^{\meas}(\N_1\|\N_2)$ are
continuous on $(0,+\infty)$, and 
$\alpha\mapsto D_{\alpha,1}(\N_1\|\N_2)$ is continuous on $(0,2)$. 
In particular,
\begin{align}
\lim_{\alpha\to 1}D_{\alpha}^{\meas}(\N_1\|\N_2)
&=
D_1^{\meas}(\N_1\|\N_2),\label{ch limit at 1-1}\\
\lim_{\alpha\to 1}\oll{D}_{\alpha}^{\meas}(\N_1\|\N_2)
&=\DU(\N_1\|\N_2),\label{ch limit at 1-2}\\
\lim_{\alpha\to 1}D_{\alpha,1}(\N_1\|\N_2)
&=\DU(\N_1\|\N_2).\label{ch limit at 1-3}
\end{align}
Moreover, if $J\subseteq(0,+\infty)$ is an interval containing $1$, and 
$J\setminus\{1\}\ni\alpha\mapsto z(\alpha)$ is such that 
$\min\{\alpha,1\}\le z(\alpha)\le +\infty$ for all $\alpha\in J\setminus\{1\}$ then 
\begin{align}\label{channel div limit z(alpha)}
\lim_{\alpha\to 1}D_{\alpha,z(\alpha)}(\N_1\|\N_2)=
\DU(\N_1\|\N_2).
\end{align}
\end{theorem}

\begin{rem}
The continuity of $\alpha\mapsto D_{\alpha,1}(\N_1\|\N_2)$ can be 
extended to $(0,+\infty)$ in the case when $\N_1$ is trace-preserving, based on 
the convexity of the $\psi$ function. We present the details in Appendix \ref{sec:cont from conv}.
\end{rem}

Finally, we remark that the \cpso R\'enyi divergences are lower semi-continuous in their arguments whenever the underlying R\'enyi divergences are lower semi-continuous on pairs of non-zero PSD operators:

\begin{prop}\label{prop:channel div lsc}
Let $D_{\alpha}^q$ be a quantum R\'enyi divergence for some $\alpha\in(0,+\infty]$ that is 
lower semi-continuous on $\B(\hil)\pne\times\B(\hil)\pne$ for any finite-dimensional Hilbert space $\hil$. Then $\CP^+(\hil,\kil)^2\ni(\N_1,\N_2)\mapsto D_{\alpha}^q(\N_1\|\N_2)$ is 
lower semi-continuous for any finite-dimensional Hilbert spaces $\hil,\kil$.
\end{prop}
\begin{proof}
Obvious from the fact that the supremum of lower semi-continuous functions is lower semi-continuous. 
\end{proof}

Proposition \ref{prop:channel div lsc} and Proposition \ref{prop:lsc quantum Renyi} yield immediately the following:

\begin{prop}
The \cpso divergences corresponding to the following quantum R\'enyi divergences are lower semi-continuous on 
$\CP^+(\hil,\kil)^2$ for any finite-dimensional Hilbert spaces $\hil,\kil$:
\begin{enumerate}
\item
$D_{\alpha,z}$, $\alpha\in(0,+\infty)$, $z\in(0,+\infty]$; 
\item
$D_{\max}$;
\item
$D_{\alpha}^{\meas}$, $\alpha\in(0,+\infty)$;
\item
$D_{\alpha}^{\max}$, $\alpha\in(0,2]$.
\end{enumerate}
\end{prop}

\section{Conclusion}

We have presented various continuity properties of quantum R\'enyi divergences, which, in particular, allowed us to prove the continuity of the sandwiched 
and the Petz-type channel R\'enyi $\alpha$-divergences at $\alpha=1$
(more generally, \cpso R\'enyi $\alpha$-divergences).
There are a number of questions naturally emerging from our investigations, answering which might be interesting
both from the point of view of quantum information theory and from the purely mathematical point of view of matrix analysis.

Probably the most natural question is whether 
continuity at $\alpha=1$ holds for the \cpso R\'enyi 
$(\alpha,z)$-divergences 
along paths of the form $J\ni\alpha\mapsto z(\alpha)$, 
where $1\in J$, and $z(\alpha)<\min\{\alpha,1\}$ for every
$\alpha\in(1-\delta,1+\delta)\cap J$ with some $\delta>0$, which would give an improvement 
of \eqref{channel div limit z(alpha)}.
Since continuity on $(\S(\hil)\times\S(\hil))_{\lambda}$ is available for 
$z>\alpha-1>0$, the missing ingredient for the applicability of 
Lemma \ref{lemma:chdiv monotone in a} and the
minimax argument in Lemma \ref{lemma:chdiv rightcont} is the monotonicity of 
$\alpha\mapsto D_{\alpha,z(\alpha)}$ in 
$\alpha$. 
To the best of our knowledge, this is only known for $z\equiv 1$ and $z(\alpha)=\alpha$, i.e., the Petz-type and the sandwiched R\'enyi divergences. 
Other natural candidates might be $z\equiv z_0$ for some fixed $z_0\ne 1$, 
or $z(\alpha)=\alpha/2$ for $\alpha\in(1,2]$. 

We also proved joint continuity of various R\'enyi divergences in the inputs and the
parameter $\alpha$ on domains of the form $(1,c)\times(\B(\hil)\times\B(\hil))_{\lambda}$
for some $c>0$ and every $\lambda>0$. If this could be improved to continuity on 
$[1,c)\times(\B(\hil)\times\B(\hil))_{\lambda}$ then with the combination of Lemma \ref{lemma:usc},
it would give an alternative and very straightforward proof of right continuity of the corresponding 
\cpso R\'enyi $\alpha$-divergences at $\alpha=1$, which, in particular, would not even require monotonicity 
of the given R\'enyi divergences in $\alpha$.
To the best of our knowledge, this 
type of joint continuity is an open question 
even for the classical R\'enyi divergences.

The proof method in Appendix \ref{sec:variational} 
for the continuity of the sandwiched R\'enyi divergences on 
$(\B(\hil)\times\B(\hil))_{\lambda}$ uses a variational representation and the 
operator monotonicity of $\id_{[0,+\infty)}^{p}$ for $p\in(0,1]$,
which limited the applicability of this technique to $\alpha\in(1,2]$. It would be 
interesting to know if this limitation can be removed, and more generally,
if the proof method based on the variational formula can be modified 
to recover the results in Corollary \ref{cor:az-continuity}.

\section*{Acknowledgments}

The work of MM was partially funded by the
National Research, Development and 
Innovation Office of Hungary via the research grant K124152, and
by the Ministry of Innovation and
Technology and the National Research, Development and Innovation
Office within the Quantum Information National Laboratory of Hungary
(Grant No. 2022-2.1.1-NL-2022-00004).
MM is grateful to 
Nilanjana Datta, Bjarne Bergh, and Robert Salzmann for noticing an error in the proof of 
\cite[Lemma 10]{CMW}, which was the original motivation for this paper,
and to Mark Wilde for discussions on possible ways to remedy that error, which eventually led to 
two different solutions, the one presented here, and another one in \cite[Appendix B]{Dingetal2020}, as well as for helpful feedback on an earlier draft of this paper.
The authors are particularly indebted to  
Omar Fawzi for the observation that the proof method for the continuity of the channel sandwiched
R\'enyi divergences at $\alpha=1$ works also for the measured R\'enyi divergences,
and for calling our attention to the example in \cite{FF2020}
showing the discontinuity of 
the maximal R\'enyi divergences.
Furthermore, the authors are grateful to two anonymous referees for carefully reading 
the first submitted version of the manuscript and for a number of comments that helped to 
improve the presentation. They are particularly indebted to one of the referees for 
recommending the study of continuity on more general sets than the one in \eqref{S lambda}, 
originally considered in the first submitted version.

\appendix

\section{Lower semi-continuity of maximal $f$-divergences}
\label{sec:maxfdiv lsc}

Let $f:\,(0,\infty)\to\bR$ be a convex function, and 
$P_f:\,[0,+\infty)\times[0,+\infty)\to\bR\cup\{+\infty\}$ be its \ki{perspective function}, defined by 
\begin{align*}
P_f(x,y):=\lim_{\ep\searrow 0}(y+\ep)f\bz\frac{x+\ep}{y+\ep}\jz
=
\begin{cases}
yf\bz\frac{x}{y}\jz, & \text{if $x,y>0$}, \\
y\lim_{t\searrow 0}f(t), & \text{if $x=0$}, \\
x\lim_{t\to+\infty}f(t)/t, & \text{if $y=0$},
\end{cases}
\end{align*}
with the convention $0\cdot\infty:=0$. For a finite set $\X$ and 
$p,q\in\ell^{\infty}(\X)\pne$, the \ki{classical $f$-divergence} of $p$ and $q$ is defined as
\cite{AliSilvey,Csiszarfdiv1,Csiszarfdiv2}
\begin{align*}
S_f^{\cl}(p\|q):=\sum_{x\in\X}P_f(p(x),q(x)).
\end{align*}
It is easy to see from its definition that $P_f$ is subadditive
(see, e.g., \cite[Lemma A.1]{HMPB}), positive homogeneous, and 
therefore jointly convex.
It is also not too difficult to see directly from its definition that it is lower 
semi-continuous on $[0,+\infty)\times[0,+\infty)$. A more elegant proof of this fact can
be obtained by first noting that 
\begin{align*}
P_f(x,y)=P_{f-f(1)}(x,y)+\underbrace{P_{f(1)}(x,y)}_{=yf(1)},
\end{align*}
whence it is enough to prove lower semi-continuity under the assumption that $f(1)=0$. In that case, for any $0<\ep<\ep'$, 
\begin{align*}
P_f(x+\ep',y+\ep')=P_f(x+\ep+(\ep'-\ep)\|y+\ep+(\ep'-\ep))
\le
P_f(x+\ep\|y+\ep)+\underbrace{(\ep'-\ep)f(1)}_{=0},
\end{align*}
whence $(0,+\infty)\ni\ep\mapsto P_f(x+\ep\|y+\ep)$ is monotone increasing, and 
\begin{align*}
P_f(x,y)=\lim_{\ep\searrow 0}P_f(x+\ep,y+\ep)=\sup_{\ep>0}P_f(x+\ep,y+\ep).
\end{align*}
Since $(x,y)\mapsto P_f(x+\ep,y+\ep)$ is clearly continuous on 
$[0,+\infty)\times[0,+\infty)$, the fact that the supremum of continuous functions is lower semi-continuous gives that 
$P_f$ is lower 
semi-continuous on $[0,+\infty)\times[0,+\infty)$.
By the above, we get immediately the following well-known properties of the classical $f$-divergences:

\begin{lemma}\label{lemma:cl fdiv properties}
For any finite set $\X$, and any convex function $f:\,(0,+\infty)\to\bR$, 
$\ell^{\infty}(\X)\pne\times\ell^{\infty}(\X)\pne\ni(p,q)\mapsto S_f^{\cl}(p\|q)$ is 
subadditive, positive homogenous, jointly convex, and lower semi-continuous. 
\end{lemma}

For $T\in\bR^{\Y\times\X}$ and $p\in\ell^{\infty}(\X)$, we define $Tp\in\ell^{\infty}(\Y)$ in the obvious way as $(Tp)(y):=\sum_{x\in\X}T_{y,x}p(x)$, $y\in\Y$. Joint convexity of the perspective function yields the following (see \cite{Csiszarfdiv1} or \cite[Proposition A.3]{HMPB}):

\begin{lemma}\label{lemma:cl fdiv mon}
Let $f:\,(0,+\infty)\to\bR$ be a convex function, $p,q\in\ell^{\infty}(\X)\pne$, and 
$T\in\bR^{\Y\times\X}$ be stochastic, i.e., $T_{y,x}\ge 0$ for all $x\in\X$, $y\in\Y$, and 
$\sum_{y\in\Y}T_{y,x}=1$ for every $x\in\X$. Then 
\begin{align*}
S_f(Tp\|Tq)\le S_f(p\|q).
\end{align*}
\end{lemma}

For any $\rho,\sigma\in B(\cH)\pne$, 
Matsumoto's \ki{maximal $f$-divergence} of $\rho$ and $\sigma$
is defined as \cite{Matsumoto_newfdiv}
\begin{align}\label{F-2}
S_f^{\max}(\rho\|\sigma):=\inf\bigl\{S_f^\cl(p\|q):
\mbox{$(\Gamma,p,q)$ reverse test for $(\rho,\sigma)$}\bigr\}.
\end{align}
(See Section \ref{sec:Renyidiv} for the definition of a reverse test.)
It is well known that with $f_{\alpha}:=s(\alpha)\id_{[0,+\infty)}^{\alpha}$, 
where $s(\alpha):=-1$, $\alpha\in(0,1)$ and $s(\alpha):=1$, $\alpha\in(1,+\infty)$, 
we have 
\begin{align}\label{max Renyi from maxfdiv}
D_{\alpha}^{\max}(\rho\|\sigma)=\frac{1}{\alpha-1}\log
\bz s(\alpha)S_{f_{\alpha}}^{\max}(\rho\|\sigma)\jz
-\frac{1}{\alpha-1}\log\Tr\rho,
\end{align}
and $\eta(t):=t\log t$, $t\in[0,+\infty)$, yields
\begin{align}\label{maxrelentr}
D_1^{\max}(\rho\|\sigma)=\frac{1}{\Tr\rho}S_{\eta}^{\max}(\rho\|\sigma).
\end{align}
\medskip

From the definition in \eqref{F-2}, it is not clear whether there is any upper bound on the size of the sets on which the reverse tests are supported. However, we have the following:

\begin{lemma}\label{L-8}
Let $n_0:=(\dim\cH)^2$. Then for every $\rho,\sigma\in B(\cH)\pne$,
\begin{align}\label{F-3}
S_f^{\max}(\rho\|\sigma)=\min\bigl\{S_f^\cl(p\|q):
\mbox{$(\Gamma,p,q)$ a reverse test for $(\rho,\sigma)$ on $\ell_{n_0+1}^\infty$}\bigr\}.
\end{align}
\end{lemma}

\begin{proof}
We prove that if $(\Gamma,p,q)$ is a reverse test on $\ell_n^\infty$ of $(\rho,\sigma)$ and
$n>n_0+1$, then there is a reverse test $(\tilde\Gamma,\tilde p,\tilde q)$ on $\ell_{n-1}^\infty$
of $(\rho,\sigma)$ such that $S_f^\cl(\tilde p\|\tilde q)\le S_f^\cl(p\|q)$. Let $(\Gamma,p,q)$
be such a reverse test. Put $\omega_i:=\Gamma(\delta_i)$ for $1\le i\le n$, where
$(\delta_i)_{i=1}^n$ is the standard basis of $\ell_n^\infty$. Then
\[
\rho=\Gamma(p)=\sum_{i=1}^np_i\omega_i,\qquad
\sigma=\Gamma(q)=\sum_{i=1}^nq_i\omega_i.
\]
Since $B(\cH)_{\sa}:=\{A\in\B(\hil):\,A^*=A\}$ is of real dimension $n_0$, note that the convex hull
$\conv\{\omega_i\}_{i=1}^n$ is included in a subspace of $B(\cH)_{\sa}$ whose dimension is
at most $n_0$. Hence by Carath\'eodory's theorem, we see that some $\omega_k$ is in the
convex hull of $\omega_i$ ($i\in\{1,\dots,n\}\setminus\{k\}$). We may assume without loss of
generality that $\omega_n\in\conv\{\omega_i\}_{i=1}^{n-1}$ so that
$\omega_n=\sum_{i=1}^{n-1}\lambda_i\omega_i$ with $\lambda_i\ge0$ and
$\sum_{i=1}^{n-1}\lambda_i=1$. Now define $\tilde p,\tilde q\in(\ell_{n-1}^\infty)\pne$ by
\[
\tilde p:=(p_1+\lambda_1p_n,\ldots,p_{n-1}+\lambda_{n-1}p_n),\qquad
\tilde q:=(q_1+\lambda_1q_n,\ldots,q_{n-1}+\lambda_{n-1}q_n),
\]
and $\tilde\Gamma:\ell_{n-1}^\infty\to B(\cH)$ by
\[
\tilde\Gamma(x):=\Gamma(x_1,\dots,x_{n-1},0),
\qquad x=(x_1,\dots,x_{n-1})\in\ell_{n-1}^\infty.
\]
Then $\tilde\Gamma$ is a positive trace-preserving map, and we have
\begin{align*}
\tilde\Gamma(\tilde p)&=\sum_{i=1}^{n-1}(p_i+\lambda_ip_n)\Gamma(\delta_i)
=\sum_{i=1}^{n-1}(p_i+\lambda_ip_n)\omega_i \\
&=\sum_{i=1}^{n-1}p_i\omega_i+p_n\sum_{i=1}^{n-1}\lambda_i\omega_i
=\sum_{i=1}^{n-1}p_i\omega_i+p_n\omega_n=\rho,
\end{align*}
and similarly $\tilde\Gamma(\tilde q)=\sigma$. Hence $(\tilde\Gamma,\tilde p,\tilde q)$ is a
reverse test on $\ell_{n-1}^\infty$ of $(\rho,\sigma)$. Moreover, since $\tilde p=T p$ and
$\tilde q=T q$ with the $(n-1)\times n$ stochastic matrix
\[
T:=\begin{bmatrix}1&&&&\lambda_1\\
&\ddots&&\mbox{\LARGE$0$}&\vdots \\
&&\ddots&&\vdots \\
&\mbox{\LARGE$0$}&&1&\lambda_{n-1}\end{bmatrix},
\]
we have $S_f^\cl(\tilde p\|\tilde q)\le S_f^\cl(p\|q)$ by Lemma \ref{lemma:cl fdiv mon}.

From the fact proved above it is immediate to see that
\[
S_f^{\max}(\rho\|\sigma)=\inf\bigl\{S_f^\cl(p\|q):
\mbox{$(\Gamma,p,q)$ is a reverse test on $\ell_{n_0+1}^\infty$ of $(\rho,\sigma)$}\bigr\}.
\]
Hence there is a sequence of reverse tests $(\Gamma_k,\rho_k,\sigma_k)$ on
$\ell_{n_0+1}^\infty$ of $(\rho,\sigma)$ such that $S_f^\cl(p_k\|q_k)\to S_f^{\max}(\rho\|\sigma)$.
By finite dimensionality we can choose a subsequence $(\Gamma_{k_l},p_{k_l},q_{k_l})$ such
that $\Gamma_{k_l}\to\Gamma_0$, $p_{k_l}\to p_0$ and $q_{k_l}\to q_0$ so that
$\Gamma_0:\ell_{n_0+1}^\infty\to B(\cH)$ is a positive trace-preserving map and
$p_0,q_0\in(\ell_{n_0+1}^\infty)\pne$. Since $\Gamma_{k_l}(p_{k_l})=\rho$ and
$\Gamma_{k_l}(q_{k_l})=\sigma$, we have $\Gamma_0(p_0)=\rho$ and
$\Gamma_0(q_0)=\sigma$ so that $(\Gamma_0,p_0,q_0)$ is a reverse test for $(\rho,\sigma)$.
Since Lemma \ref{lemma:cl fdiv properties} gives
\[
S_f^\cl(p_0\|q_0)\le\liminf_{l\to\infty}S_f^\cl(p_{k_l}\|q_{k_l})
=S_f^{\max}(\rho\|\sigma),
\]
the minimum in \eqref{F-3} is attained by $(\Gamma_0,p_0,q_0)$.
\end{proof}

\begin{theorem}\label{thm:maxfdiv lsc}
For any finite-dimensional Hilbert space $\hil$, and any 
convex function $f:\,(0,\infty)\to \bR$, 
the function $\B(\hil)\pne\times\B(\hil)\pne\ni(\rho\|\sigma)\mapsto
S_f^{\max}(\rho\|\sigma)$ is jointly lower
semi-continuous.
\end{theorem}

\begin{proof}
Let $\rho_k,\sigma_k\in B(\cH)\pne$, $k\in\bN$ be 
two sequences such that
that $\rho_k\to\rho$ and
$\sigma_k\to\sigma$ for some $\rho,\sigma\in\B(\hil)\pne$. For each $k$, by Lemma \ref{L-8} we can choose a reverse test
$(\Gamma_k,p_k,q_k)$ on $\ell_{n_0+1}^\infty$ of $(\rho_k,\sigma_k)$ such that
$S_f^\cl(p_k\|q_k)=S_f^{\max}(\rho_k\|\sigma_k)$. Choose a subsequence $(k_l)_{l\in\bN}$
such that
\[
\lim_{l\to\infty}S_f^\cl(p_{k_l}\|q_{k_l})=\liminf_{k\to\infty}S_f^\cl(p_k\|q_k).
\]
By taking a further subsequence if necessary, we may assume that $\Gamma_{k_l}\to\Gamma_0$,
$p_{k_l}\to p_0$ and $q_{k_l}\to q_0$ so that $\Gamma_0$ is a positive trace-preserving map
and
\[
\Gamma_0(p_0)=\lim_{l\to\infty}\Gamma_{k_l}(p_{k_l})=\lim_{l\to\infty}\rho_{k_l}=\rho,
\]
and similarly $\Gamma_0(q_0)=\sigma$. Hence $(\Gamma_0,p_0,q_0)$ is a reverse test of
$(\rho,\sigma)$, and 
\[
S_f^{\max}(\rho_0\|\sigma_0)\le
S_f^\cl(p_0\|q_0)\le\liminf_{l\to\infty}S_f^\cl(p_{k_l}\|q_{k_l})
=\liminf_{k\to\infty}S_f^\cl(p_k\|q_k)=\liminf_{k\to\infty}S_f^{\max}(\rho_k\|\sigma_k),
\]
where the first inequality 
is by definition, and the second inequality 
follows from Lemma \ref{lemma:cl fdiv properties}.
\end{proof}

\section{Continuity of the $(\alpha,z)$-divergences from a variational formula}
\label{sec:variational}

Here we give an alternative proof of the continuity of the 
sandwiched R\'enyi $\alpha$-divergences on sets of the form 
$(\S(\hil)\times\S(\hil))_{\lambda}$
for every $\alpha\in(1,2]$, which in turn is sufficient for the minimax argument 
presented in Lemma \ref{lemma:chdiv rightcont}.

The following is a simple modification of the 
variational formulas given in \cite{Zhang2018}, and previously in 
\cite{FL} for the $z=\alpha>1$ case. 
We include a proof for completeness.

\begin{lemma}\label{lemma:variational}
Let $\rho,\sigma\in\B(\hil)\pne$ and $\alpha\in(1,2]$, $z>0$, be such that 
$\rho^{\frac{\alpha}{z}}\le\lambda\sigma^{\frac{\alpha}{z}}$ for some $\lambda>0$.
Then
\begin{align}
Q_{\alpha,z}(\rho\|\sigma)
&=
\max\left\{
\alpha\Tr \bz \rho^{\frac{\alpha}{2z}} H\rho^{\frac{\alpha}{2z}} \jz^{\frac{z}{\alpha}}
+(1-\alpha)\Tr\bz \sigma^{\frac{\alpha-1}{2z}}H\sigma^{\frac{\alpha-1}{2z}}\jz^{\frac{z}{\alpha-1}}
\,:\,H\in\B(\hil)\pne\right\}\label{a-z variational3}\\
&=
\max\left\{
\alpha\Tr \bz \rho^{\frac{\alpha}{2z}} H\rho^{\frac{\alpha}{2z}} \jz^{\frac{z}{\alpha}}
+(1-\alpha)\Tr\bz \sigma^{\frac{\alpha-1}{2z}}H\sigma^{\frac{\alpha-1}{2z}}\jz^{\frac{z}{\alpha-1}}
\,:\,H\in[0,\max\{1,\lambda\}I]\right\}\label{a-z variational2}\\
&=
\max\left\{
\alpha\Tr \bz \rho^{\frac{\alpha}{2z}} H\rho^{\frac{\alpha}{2z}} \jz^{\frac{z}{\alpha}}
+(1-\alpha)\Tr\bz \sigma^{\frac{\alpha-1}{2z}}H\sigma^{\frac{\alpha-1}{2z}}\jz^{\frac{z}{\alpha-1}}
\,:\,H\in[0,\lambda^{\alpha-1}I]\right\}\label{a-z variational}.
\end{align}
\end{lemma}
\begin{proof}
For any $H\in\B(\hil)\p$ and $\alpha>1$, we have 
\begin{align*}
\Tr \bz \rho^{\frac{\alpha}{2z}}H\rho^{\frac{\alpha}{2z}}\jz^{\frac{z}{\alpha}}
&=
\Tr\abs{H^{1/2}\rho^{\frac{\alpha}{2z}}}^{\frac{2z}{\alpha}}
=
\Tr\abs{H^{1/2}\sigma^{\frac{\alpha-1}{2z}}\sigma^{\frac{1-\alpha}{2z}}\rho^{\frac{\alpha}{2z}}}^{\frac{2z}{\alpha}}
\\
&=
\norm{H^{1/2}\sigma^{\frac{\alpha-1}{2z}}\sigma^{\frac{1-\alpha}{2z}}\rho^{\frac{\alpha}{2z}}}_{\frac{2z}{\alpha}}^{\frac{2z}{\alpha}}
\le
\norm{H^{1/2}\sigma^{\frac{\alpha-1}{2z}}}_{\frac{2z}{\alpha-1}}^{\frac{2z}{\alpha}}
\norm{\sigma^{\frac{1-\alpha}{2z}}\rho^{\frac{\alpha}{2z}}}_{2z}^{\frac{2z}{\alpha}}
\\
&=
\left[\Tr\bz\sigma^{\frac{\alpha-1}{2z}}H \sigma^{\frac{\alpha-1}{2z}}\jz^{\frac{z}{\alpha-1}}\right]^{\frac{\alpha-1}{\alpha}}
\Bigg[\underbrace{\Tr\bz 
\rho^{\frac{\alpha}{2z}}\sigma^{\frac{1-\alpha}{z}}\rho^{\frac{\alpha}{2z}}\jz^{z}}_{=Q_{\alpha,z}(\rho\|\sigma)}\Bigg]^{\frac{1}{\alpha}}
\\
&\le
\frac{\alpha-1}{\alpha}\Tr\bz\sigma^{\frac{\alpha-1}{2z}}H \sigma^{\frac{\alpha-1}{2z}}\jz^{\frac{z}{\alpha-1}}
+
\frac{1}{\alpha}Q_{\alpha,z}(\rho\|\sigma),
\end{align*}
where the first inequality is due to the operator H\"older inequality, and 
the second inequality is trivial from the convexity of the exponential function.
A simple rearrangement yields that LHS$\ge$RHS in \eqref{a-z variational3}, and it is 
obvious that the maximum in \eqref{a-z variational3} is lower bounded by the maximum in 
\eqref{a-z variational2}.
Note that if $\lambda\in(0,1]$ then 
$\lambda^{\alpha-1}\in[\lambda,1]$, while if $\lambda>1$ then 
$\lambda^{\alpha-1}\in(1,\lambda]$, proving that 
maximum in \eqref{a-z variational2} is lower bounded by the maximum in 
\eqref{a-z variational}.

On the other hand, with the choice 
$H_{\alpha,z}:=\sigma^{\frac{1-\alpha}{2z}}\bz\sigma^{\frac{1-\alpha}{2z}}\rho^{\frac{\alpha}{z}} 
\sigma^{\frac{1-\alpha}{2z}}\jz^{\alpha-1}\sigma^{\frac{1-\alpha}{2z}}$, a straightforward 
computation yields
\begin{align*}
\alpha\Tr \bz \rho^{\frac{\alpha}{2z}} H_{\alpha,z}\rho^{\frac{\alpha}{2z}} \jz^{\frac{z}{\alpha}}
+(1-\alpha)\Tr\bz \sigma^{\frac{\alpha-1}{2z}}H_{\alpha,z}\sigma^{\frac{\alpha-1}{2z}}\jz^{\frac{z}{\alpha-1}}=Q_{\alpha,z}(\rho\|\sigma),
\end{align*}
proving the equality in \eqref{a-z variational3}.
By assumption, $\rho^{\frac{\alpha}{z}}\le\lambda\sigma^{\frac{\alpha}{z}}$.
Since $t\mapsto t^{\alpha-1}$ is operator monotone for $\alpha\in(1,2]$
(see, e.g., \cite{Bhatia}), we get 
\begin{align*}
H_{\alpha,z}
=
\sigma^{\frac{1-\alpha}{2z}}
\underbrace{\bz\sigma^{\frac{1-\alpha}{2z}}\rho^{\frac{\alpha}{z}} 
\sigma^{\frac{1-\alpha}{2z}}\jz^{\alpha-1}}_{\le\lambda^{\alpha-1}\sigma^{\frac{\alpha-1}{z}}}
\sigma^{\frac{1-\alpha}{2z}}
\le
\lambda^{\alpha-1}\sigma^0
\le
\lambda^{\alpha-1}I.
\end{align*}
This shows that the maximum in \eqref{a-z variational} is equal to $Q_{\alpha,z}(\rho\|\sigma)$, thus completing the proof.
\end{proof}

\begin{prop}\label{prop:az cont}
Let $\hil$ be a finite-dimensional Hilbert space.
\begin{enumerate}
\item\label{az cont1}
For any $\alpha\in(1,2]$, $z>0$, and  
$\lambda>0$, $Q_{\alpha,z}$ and $D_{\alpha,z}$ are continuous on 
\begin{align*}
\bz\B(\hil)\times\B(\hil)\jz_{\alpha,z,\lambda}:=\{(\rho,\sigma)\in\B(\hil)\pne\times\B(\hil)\pne:\,\rho^{\frac{\alpha}{z}}\le\lambda\sigma^{\frac{\alpha}{z}}\}.
\end{align*}

\item\label{az cont2}
For any $\lambda>0$,
\begin{align}\label{az joint cont2}
\{(\alpha,z):\,\alpha\in(1,2],\,z\ge \alpha\}\times(\B(\hil)\times\B(\hil))_{\lambda}
\ni(\alpha,z,\rho,\sigma)\mapsto D_{\alpha,z}(\rho\|\sigma)
\end{align}
is continuous.
\end{enumerate}
\end{prop}
\begin{proof}
\ref{az cont1}\s
Since the power functions $\id_{[0,+\infty)}^p$ 
are continuous for any $p>0$, continuity of the functional calculus implies that the map
\begin{align*}
(\rho,\sigma,H)\mapsto
\alpha\Tr \bz \rho^{\frac{\alpha}{2z}} H\rho^{\frac{\alpha}{2z}} \jz^{\frac{z}{\alpha}}
+(1-\alpha)\Tr\bz \sigma^{\frac{\alpha-1}{2z}}H\sigma^{\frac{\alpha-1}{2z}}\jz^{\frac{z}{\alpha-1}}
\end{align*}
is continuous on $\bz\B(\hil)\times\B(\hil)\jz_{\alpha,z,\lambda}\times [0,\lambda^{\alpha-1}I]$.
Thus, by Lemma \ref{lemma:variational} and \ref{usc3} of Lemma \ref{lemma:usc},
$(\rho,\sigma)\mapsto Q_{\alpha,z}(\rho\|\sigma)$ is continuous on 
$\bz\B(\hil)\times\B(\hil)\jz_{\alpha,z,\lambda}$, from which the 
continuity of $D_{\alpha,z}$ on the same set follows immediately.

\ref{az cont2}\s Note that if $\alpha\in(1,2]$, $z\ge \alpha$, 
and $\rho\le\lambda\sigma$, then 
$\rho^{\frac{\alpha}{z}}\le\lambda^{\frac{\alpha}{z}}\sigma^{\frac{\alpha}{z}}$
due to the operator monotonicity of 
$\id_{[0,+\infty)}^{\frac{\alpha}{z}}$ (see \cite{Bhatia}). 
If $\lambda\ge 1$ then $\lambda^{\frac{\alpha}{z}}\le \lambda$, 
while if $\lambda\in(0,1)$ then $\lambda^{\frac{\alpha}{z}}< 1$.
Hence,
\begin{align*}
\rho^{\frac{\alpha}{z}}
\le\max\{1,\lambda\}\sigma^{\frac{\alpha}{z}},
\end{align*}
and therefore $Q_{\alpha,z}(\rho\|\sigma)$ can be expressed as in 
\eqref{a-z variational2}.
Since
\begin{align*}
&\{(\alpha,z):\,\alpha\in(1,2],\,z\ge \alpha\}\times(\B(\hil)\times\B(\hil))_{\lambda}\times
\left[0,\max\{1,\lambda\} I\right]\\
&\ds\ni(\alpha,z,\rho,\sigma,H)\mapsto
\alpha\Tr \bz \rho^{\frac{\alpha}{2z}} H\rho^{\frac{\alpha}{2z}} \jz^{\frac{z}{\alpha}}
+(1-\alpha)\Tr\bz \sigma^{\frac{\alpha-1}{2z}}H\sigma^{\frac{\alpha-1}{2z}}\jz^{\frac{z}{\alpha-1}}
\end{align*}
is continuous (see Lemma \ref{lemma:uniform conv}), the continuity of 
\begin{align}
\{(\alpha,z):\,\alpha\in(1,2],\,z\ge \alpha\}\times(\B(\hil)\times\B(\hil))_{\lambda}
\ni(\alpha,z,\rho,\sigma)\mapsto Q_{\alpha,z}(\rho\|\sigma)
\end{align}
follows by 
\eqref{a-z variational2} and \ref{usc3} of Lemma \ref{lemma:usc},
from which the continuity of \eqref{az joint cont2} follows immediately.
\end{proof}

As a special case of Proposition \ref{prop:az cont}, we get the following:
\begin{cor}\label{cor:sandwiched jointcont}
For any finite-dimensional Hilbert space $\hil$, and any $\lambda>0$, 
\begin{align*}
(1,2]\times(\B(\hil)\times\B(\hil))_{\lambda}\ni(\alpha,\rho,\sigma)
\mapsto D_{\alpha}\nw(\rho\|\sigma)
\end{align*}
is continuous. In particular, $D_{\alpha}\nw$ is continuous on 
$(\B(\hil)\times\B(\hil))_{\lambda}$ for any $\lambda>0$ and $\alpha\in(1,2]$. 
\end{cor}

Note that Corollary \ref{cor:sandwiched jointcont} gives  weaker versions of 
\eqref{regmeas jointcont2} and \eqref{regmeas cont},
with $\alpha$ restricted to the interval $(1,2]$ in the former.
However, this is still sufficient to prove the continuity of the sandwiched 
\cpso R\'enyi  divergence at $\alpha=1$, 
as in Lemma \ref{lemma:chdiv rightcont}.

\section{$D_{\alpha,0}(\rho\|\sigma)$ at $\alpha=1$}
\label{sec:zero Renyi limit}

Recall that by \eqref{z=0 def}, 
$Q_{\alpha,0}(\rho\|\sigma):=\lim_{z\searrow0}Q_{\alpha,z}(\rho\|\sigma)$
exists for any $\rho,\sigma\in\B(\hil)\pne$. Moreover, 
by \cite[Theorem 2.5]{AH2019}, the limit
\[
Z_\alpha(\rho\|\sigma):=\lim_{z\searrow0}(\rho^{\alpha\over2z}\sigma^{1-\alpha\over z}
\rho^{\alpha\over2z})^z
\]
exists, whence
\[
Q_{\alpha,0}(\rho\|\sigma)
=\Tr Z_\alpha(\rho\|\sigma).
\]
By \eqref{Umegaki bound on az},
\begin{align}\label{Renyi alpha 0 limit1}
\lim_{\alpha\nearrow 1}D_{\alpha,0}(\rho\|\sigma)\le
\DU(\rho\|\sigma)\le 
\lim_{\alpha\searrow 1}D_{\alpha,0}(\rho\|\sigma).
\end{align} 
Our goal here is to analyze the case when equality holds in either of these inequailities.

Let 
\[
\rho=\sum_{i=1}^da_i|v_i\>\<v_i|,\qquad
\sigma=\sum_{i=1}^db_i|w_i\>\<w_i|
\]
be eigendecompositions of $\rho,\sigma\in\B(\hil)\pne$, respectively, where
\[
a_1\ge a_2\ge\ldots\ge a_d,\qquad b_1\ge b_2\ge\ldots\ge b_d
\]
are the eigenvalues of $\rho,\sigma$, respectively, in decreasing order with multiplicities
and  $(v_i)_{i=1}^d$, $(w_i)_{i=1}^d$ are the corresponding eigenvectors forming
orthonormal bases of $\cH$. Moreover, let $0=i_0<i_1<\ldots<i_{l-1}<i_l=d$ and
$0=j_0<j_1<\ldots<j_{m-1}<j_m=d$ be taken so that
\begin{align*}
&a_1=\ldots=a_{i_1}>a_{i_1+1}=\ldots=a_{i_2}>\ldots>a_{i_{l-1}+1}=\ldots=a_{i_l}, \\
&b_1=\ldots=b_{j_1}>b_{j_1+1}=\ldots=b_{j_2}>\ldots>b_{j_{m-1}+1}=\ldots=b_{j_m}.
\end{align*}
For any subset $I,J$ of $\{1,\ldots,d\}$ with $|I|=|J|=k$, let $\det[\<v_i,w_j\>]_{I,J}$ denote
the determinant of the $k\times k$ matrix $[\<v_i,w_j\>]_{i\in I,j\in J}$.


Let
\[
\lambda_1(\alpha)\ge\lambda_2(\alpha)\ge\ldots\ge\lambda_d(\alpha)
\]
be the eigenvalues of $Z_\alpha(\rho\|\sigma)$ in decreasing order with multiplicities.
Since
\[
\rho^\alpha=\sum_{i=1}^da_i^\alpha|v_i\>\<v_i|,\qquad
\sigma^{1-\alpha}=\sum_{i=1}^db_i^{1-\alpha}|w_i\>\<w_i|,
\]
the next lemma follows from \cite[Theorem 3.1]{AH2019}.

\begin{lemma}\label{L-1}
\begin{itemize}
\item[(1)] The following conditions are equivalent:
\begin{itemize}
\item[(a)] $(\lambda_i(\alpha))_{i=1}^d=(a_i^\alpha b_i^{1-\alpha})_{i=1}^d$ for some
(equivalently, for all) $\alpha\in(0,1)$.
\item[(b)] For any $k\in\{i_1,\ldots,i_{l-1},j_1,\ldots,j_{m-1}\}$ there exist
$I_k,J_k\subset\{1,\ldots,d\}$ with $|I_k|=|J_k|=k$ such that
\begin{align*}
&\{1,\ldots,i_{r-1}\}\subseteq I_k\subseteq\{1,\ldots,i_r\}\quad
\mbox{for some  $r\in\{1,\ldots,l\}$}, \\
&\{1,\ldots,j_{s-1}\}\subseteq J_k\subseteq\{1,\ldots,j_s\}\quad
\mbox{for some $s\in\{1,\ldots,m\}$}, \\
&\det[\<v_i,w_j\>]_{I_k,J_k}\ne0.
\end{align*}
\end{itemize}

\item[(2)] When $\sigma$ is invertible, the following conditions are equivalent:
\begin{itemize}
\item[(a)$'$] $(\lambda_i(\alpha))_{i=1}^d=(a_i^\alpha b_{d+1-i}^{1-\alpha})_{i=1}^d$ for
some (equivalently, for all) $\alpha>1$.
\item[(b)$'$] For any $k\in\{i_1,\ldots,i_{l-1},d-j_{m-1},\ldots,d-j_1\}$ there exist
$I_k,J_k\subset\{1,\ldots,d\}$ with $|I_k|=|J_k|=k$ such that
\begin{align*}
&\{1,\ldots,i_{r-1}\}\subseteq I_k\subseteq\{1,\ldots,i_r\}\quad
\mbox{for some $r\in\{1,\ldots,l\}$}, \\
&\{j_s+1,\ldots,d\}\subseteq J_k\subseteq\{j_{s-1}+1,\ldots,d\}\quad
\mbox{for some $s\in\{1,\ldots,m\}$}, \\
&\det[\<v_i,w_j\>]_{I_k,J_k}\ne0.
\end{align*}
\end{itemize}
\end{itemize}
\end{lemma}

In particular, when $a_1>\ldots>a_d$ and $b_1>\ldots>b_d$, condition (b) means
that $\det[\<v_i,w_j\>]_{1\le i,j\le k}\ne0$ for $1\le k\le d$, and (b)$'$ means that
$\det[\<v_i,w_{d+1-j}\>]_{1\le i,j\le k}\ne0$ for $1\le k\le d$. Note that the set of
$(\rho,\sigma)\in\B(\cH)\pne\times\B(\cH)\pne$ satisfying (b) and (b)$'$ (hence (a) and
(a)$'$ as well) is an open dense subset of $\B(\cH)\pne\times\B(\cH)\pne$.

The next proposition characterizes the equality cases of the inequalities in \eqref{Renyi alpha 0 limit1}

\begin{prop}\label{P-2}
Let $\rho,\sigma\in\B(\cH)\pne$ and assume that $\sigma$ is invertible. Let
$a_1\ge\ldots\ge a_d$ and $b_1\ge\ldots\ge b_d$ be the eigenvalues of $\rho$ and
$\sigma$, respectively, in decreasing order.
\begin{itemize}
\item[(1)] Assume that condition (b) of Lemma \ref{L-1} holds. Then
\[
\lim_{\alpha\nearrow1}D_{\alpha,0}(\rho\|\sigma)=D^\Um(\rho\|\sigma)
\]
if and only if there exists an orthonormal basis $(e_i)_{i=1}^d$ of $\cH$ such that
\[
\rho=\sum_{i=1}^da_i|e_i\>\<e_i|,\qquad\sigma=\sum_{i=1}^db_i|e_i\>\<e_i|.
\]
\item[(2)] Assume that condition (b)$'$ of Lemma \ref{L-1} holds. Then
\[
\lim_{\alpha\searrow1}D_{\alpha,0}(\rho\|\sigma)=D^\Um(\rho\|\sigma)
\]
if and only if there exists an orthonormal basis $(e_i)_{i=1}^d$ of $\cH$ such that
\[
\rho=\sum_{i=1}^da_i|e_i\>\<e_i|,\qquad\sigma=\sum_{i=1}^db_{d+1-i}|e_i\>\<e_i|.
\]
\end{itemize}
\end{prop}

\begin{proof}
(1)\enspace
By assumption (b) we have (a) of Lemma \ref{L-1}. Hence for any $\alpha\in(0,1)$,
$Q_{\alpha,0}(\rho\|\sigma)=\sum_{i=1}^da_i^\alpha b_i^{1-\alpha}$ and
\[
D_{\alpha,0}(\rho\|\sigma)
={1\over\alpha-1}\log{\sum_{i=1}^da_i^\alpha b_i^{1-\alpha}\over\sum_{i=1}^da_i}.
\]
Therefore,
\[
\lim_{\alpha\nearrow1}D_{\alpha,0}(\rho\|\sigma)
={\sum_{i=1}^d(a_i\log a_i-a_i\log b_i)\over\sum_{i=1}^da_i}
={\Tr\rho\log\rho-\sum_{i=1}^da_i\log b_i\over\Tr\rho}
\]
so that the equality $\lim_{\alpha\nearrow1}D_{\alpha,0}(\rho\|\sigma)=D^\Um(\rho\|\sigma)$
holds if and only if
\begin{align}\label{F-*2}
\Tr\rho\log\sigma=\sum_{i=1}^da_i\log b_i.
\end{align}
It is clear that this holds if the condition stated in (1) is satisfied. To prove the converse,
note that
\begin{align}\label{F-*3}
\Tr\rho\log\sigma=\sum_{i=1}^d\<v_i,(\rho\log\sigma)v_i\>
=\sum_{i=1}^da_i\<v_i,(\log\sigma)v_i\>.
\end{align}
Let $(c_i)_{i=1}^d$ be the decreasing rearrangement of $(\<v_i,(\log\sigma)v_i\>)_{i=1}^d$.
Then it is clear that
\begin{align}\label{F-*4}
\sum_{i=1}^da_i\<v_i,(\log\sigma)v_i\>\le\sum_{i=1}^da_ic_i.
\end{align}
Since $(\log b_i)_{i=1}^d$ is the eigenvalue vector of $\log\sigma$ in decreasing order,
we have the majorization $(c_i)_{i=1}^d\prec(\log b_i)_{i=1}^d$ (see, e.g., 
\cite[Exercise II.1.12]{Bhatia}). Hence we have
\begin{equation}\label{F-*5}
\begin{aligned}
\sum_{i=1}^da_ic_i
&=\sum_{k=1}^l(a_{i_k}-a_{i_{k+1}})\sum_{i=1}^{i_k}c_i\quad
(\mbox{where $a_{i_{l+1}}:=0$}) \\
&\le\sum_{k=1}^l(a_{i_k}-a_{i_{k+1}})\sum_{i=1}^{i_k}\log b_i
=\sum_{i=1}^da_i\log b_i.
\end{aligned}
\end{equation}
Combining \eqref{F-*3}--\eqref{F-*5} shows that equality \eqref{F-*2} holds if and only if
\begin{align*}
\sum_{i=1}^{i_k}\<v_i,(\log\sigma)v_i\>=\sum_{i=1}^{i_k}\log b_i,\qquad1\le k\le l.
\end{align*}
Letting $P_k:=\sum_{i=i_{k-1}+1}^{i_k}|v_i\>\<v_i|$ for $1\le k\le l$, by Lemma \ref{L-3} below
(applied to $A=\log\sigma$) we find that for every $k=1,\ldots,l$, $(P_1+\ldots+P_k)\cH$ is
a reducing subspace for $\log\sigma$, hence so is $P_k\cH$ for every $k$. Replacing
$(v_i)_{i=i_{k-1}+1}^{i_k}$ with the eigenvectors of $\log\sigma$ forming an orthonormal
basis of $P_k\cH$ for each $k$, we obtain an orthonormal basis $(e_i)_{i=1}^d$ such that
$\rho=\sum_{i=1}^da_i|e_i\>\<e_i|$ and $\log\sigma=\sum_{i=1}^d(\log b_i)|e_i\>\<e_i|$,
whence $\sigma=\sum_{i=1}^db_i|e_i\>\<e_i|$.

(2)\enspace
By assumption (b)$'$ we have (a)$'$ of Lemma \ref{L-1}. Hence for any $\alpha>1$,
$Q_{\alpha,0}(\rho\|\sigma)=\sum_{i=1}^da_i^\alpha b_{d+1-i}^{1-\alpha}$. The proof is
similar to the above (1), by replacing $b_i$ with $b_{d+1-i}$ and inequalities \eqref{F-*4}
and \eqref{F-*5} with
\[
\sum_{i=1}^da_i\<v_i,(\log\sigma)v_i\>\ge\sum_{i=1}^da_ic_{d+1-i}
\]
and
\begin{align*}
\sum_{i=1}^da_ic_{d+1-i}
&=\sum_{k=1}^l(a_{i_k}-a_{i_{k+1}})\sum_{i=1}^{i_k}c_{d+1-i} \\
&\ge\sum_{k=1}^l(a_{i_k}-a_{i_{k+1}})\sum_{i=1}^{i_k}\log b_{d+1-i}
=\sum_{i=1}^da_i\log b_{d+1-i}.
\end{align*}
We omit the details.
\end{proof}

\begin{lemma}\label{L-3}
Let $A\in B(\cH)$ be self-adjoint with the eigenvalues $\lambda_1\ge\ldots\ge\lambda_d$
in decreasing order. If a rank $k$ projection $P$ satisfies $\Tr AP=\sum_{i=1}^k\lambda_i$
(or $\Tr AP=\sum_{i=1}^k\lambda_{d+1-i}$), then $P\cH$ is a reducing subspace for $A$, 
i.e., $AP=PAP$.
\end{lemma}

\begin{proof}
Diagonalize $PAP$ with an orthonormal basis $(v_i)_{i=1}^k$ and expand $(v_i)_{i=1}^k$
into an orthonormal basis $(v_i)_{i=1}^d$ of $\cH$. Representing $A$ with respect to this
basis, we may assume that $A$ is a matrix $[a_{ij}]_{i,j=1}^d$ so that $[a_{ij}]_{i,j=1}^k$ is
diagonal with $\sum_{i=1}^ka_{ii}=\sum_{i=1}^k\lambda_i$. Let $(f_i)_{i=1}^d$ be the
standard basis of $\bC^d$. To prove the lemma, we need to show that $a_{ij}=0$ if $i\le k<j$.
Assume the contrary that $a_{ij}\ne0$ for some $i\le k<j$. Choose a $\theta\in\bR$ such that
$e^{i\theta}a_{ij}=|a_{ij}|$, and for each $t\in(0,1)$ let $P_t$ be the rank $k$ projection onto
the subspace spanned by $f_i$ for $i\in\{1,\ldots,k\}\setminus\{i\}$ plus
$u_t:=\sqrt tf_i+\sqrt{1-t}\,e^{i\theta}f_j$. A simple computation gives
\begin{align*}
\<u_t,Au_t\>&=ta_{ii}+(1-t)a_{jj}+2\sqrt{t(1-t)}\,|a_{ij}| \\
&=a_{ii}+(1-t)\Biggl(-a_{ii}+a_{jj}+2\sqrt{t\over1-t}\,|a_{ij}|\Biggr),
\end{align*}
which is larger than $a_{ii}$ for $t$ sufficiently near $1$. This means that
$\Tr AP_t>\sum_{i=1}^k\lambda_i$, contradicting the well-known formula (see, e.g.,
\cite[Exercise II.1.13]{Bhatia})
\[
\max\{\Tr AQ:\mbox{$Q$ a projection of rank $k$}\}=\sum_{i=1}^k\lambda_i.
\]
The proof of the case $\Tr AP=\sum_{i=1}^k\lambda_{d+1-i}$ is similar.
\end{proof}

It is obvious that the assumption of $\sigma$ being invertible in Proposition \ref{P-2} can
be replaced with $\rho^0\le\sigma^0$ (we may argue with restriction on $\sigma^0\cH$).
Hence we have the following:

\begin{cor}\label{cor:zero Renyi limit}
For any non-commuting $\rho,\sigma$ with $\rho^0\le\sigma^0$, if conditions (b) and (b)$'$
of Lemma \ref{L-1} hold, then
\[
\lim_{\alpha\nearrow1}D_{\alpha,0}(\rho\|\sigma)
<D^\Um(\rho\|\sigma)<\lim_{\alpha\searrow1}D_{\alpha,0}(\rho\|\sigma).
\]
\end{cor}
\medskip

From the above, we easily obtain the following:
\medskip

\noindent \textbf{Proof of Proposition \ref{prop:zero Renyi limit}}\ds
If $\rho=\sum_{i=1}^2a_i\pr{v_i}$ and $\sigma=\sum_{i=1}^2 b_i\pr{w_i}$ are non-commuting 
qubit states then $a_1>a_2$ and $b_1>b_2$ can be assumed, and non-commutativity also implies that $\inner{v_i}{w_i}\ne 0$, $i=1,2$. Hence, conditions (b) and (b)$'$
of Lemma \ref{L-1} hold, and Proposition \ref{prop:zero Renyi limit} follows from 
Corollary \ref{cor:zero Renyi limit}. \hfill$\Box$

\section{Continuity in $\alpha$ from convexity}
\label{sec:cont from conv}

\begin{lemma}\label{lemma:cont from conv}
Let $J\subseteq (1,+\infty)$ be an open interval, and 
$(D_{\alpha}^q)_{\alpha\in J}$ be quantum R\'enyi $\alpha$-divergences such that 
for any $\hil$ and any $\rho,\sigma\in\B(\hil)\pne$, 
$\alpha\mapsto\psi_{\alpha}^q(\rho\|\sigma)$ is convex on $J$.
Then for any $\N_1,\N_2\in\CP^+(\hil,\kil)$,
\begin{align}\label{channel psi}
\alpha\mapsto \psi_{\alpha}^q(\N_1\|\N_2)
:=
\sup_{d\in\bN}\sup_{\rho\in\S(\bC^d\otimes\hil)}
\psi_{\alpha}^q\bz(\id\otimes\N_1)\rho\|(\id\otimes\N_2)\rho\jz
\end{align}
is convex on $J$,
and if $\N_1$ is trace-preserving then
$\alpha\mapsto D_{\alpha}^q(\N_1\|\N_2)$ is
continuous on the interior of the interval
$J\cap\{\alpha:\,D_{\alpha}^q(\N_1\|\N_2)<+\infty\}$.
\end{lemma}
\begin{proof}
Since the supremum of convex functions is convex, the assumption guarantees 
the convexity of \eqref{channel psi}, which is then 
continuous on the interior of the interval on which it is finite. 
If $\N_1$ is trace-preserving then 
\begin{align*}
D_{\alpha}^q(\N_1\|\N_2)=\frac{1}{\alpha-1}\psi_{\alpha}^q(\N_1\|\N_2),
\ds\ds\ds\alpha\in (1,+\infty),
\end{align*}
and hence the continuity of $\alpha\mapsto \psi_{\alpha}^q(\N_1\|\N_2)$ implies
the continuity of $\alpha\mapsto D_{\alpha}^q(\N_1\|\N_2)$.
\end{proof}

\begin{cor}\label{ex:ch cont from ch conv}
For any $\N_1\in\cptp(\hil,\kil)$ and $\N_2\in\CP^+(\hil,\kil)$,
$\alpha\mapsto D_{\alpha,1}(\N_1\|\N_2)$ is
continuous on $(1,+\infty)$.
\end{cor}
\begin{proof}
If $D_{\max}(\N_1\|\N_2)=+\infty$ then by Lemma \ref{lemma:ch div finite},
$D_{\alpha,1}(\N_1\|\N_2)=+\infty$
for all $\alpha\in(1,+\infty)$, and its
continuity in $\alpha$ holds trivially. 
If $D_{\max}(\N_1\|\N_2)<+\infty$ then 
$D_{\alpha,1}(\N_1\|\N_2)$
is finite for all $\alpha\in(1,+\infty)$. Hence, by 
Lemma \ref{lemma:cont from conv} and
Proposition \ref{prop:Petz cont}, 
$\alpha\mapsto \psi_{\alpha,1}(\N_1\|\N_2)$ is a finite-valued convex, and hence continuous, function
on $(1,+\infty)$, whence $\alpha\mapsto D_{\alpha,1}(\N_1\|\N_2)$ is continuous on $(1,+\infty)$.
\end{proof}

\begin{rem}
Note that the continuity of \eqref{chdiv jointcont6} in Proposition \ref{prop:channel joint cont} yields the continuity of 
$\alpha\mapsto D_{\alpha,1}(\N_1\|\N_2)$ without the assumption that $\N_1$ is 
trace-preserving, but only on the interval $(1,2)$.
\end{rem}

\begin{rem}
Note that the same argument as in Corollary \ref{ex:ch cont from ch conv}, 
using the convexity parts of Propositions
\ref{prop:meas cont} and \ref{prop:regmeas cont}, yields the continuity of 
$\alpha\mapsto D_{\alpha}^{\meas}(\N_1\|\N_2)$
and 
$\alpha\mapsto \oll{D}_{\alpha}^{\meas}(\N_1\|\N_2)$
on $(1,+\infty)$, but only in the case when $\N_1$ is trace-preserving, which is weaker than what can be obtained from \eqref{chdiv jointcont4} and \eqref{chdiv jointcont5} in Proposition 
\ref{prop:channel joint cont}.
\end{rem}

\section{Discontinuity example}
\label{sec:discont ex}

The following proposition gives an alternative 
to Lemma \ref{lemma:a>1 discont} to prove
the discontinuity of $D_{\alpha,z}$ on $(\S(\hil)\times\S(\hil))_{\lambda}$ 
for $\alpha>1$ and $z\in(0,\alpha-1]$.
Note that in the construction below, all states are invertible, whereas in 
Lemma \ref{lemma:a>1 discont} the $\rho$ states were pure.

\begin{prop}
For any $\gamma\in(0,1)$, there exist 
two sequences of qubit states $(\rho_n)_{n\in\bN}$ and 
$(\sigma_n)_{n\in\bN}$ such that 
\begin{align}\label{counter1}
\lim_{n\to+\infty}\rho_n=\lim_{n\to+\infty}\sigma_n,\ds\ds\ds\ds
\lim_{n\to+\infty}D_{\max}(\rho_n\|\sigma_n)
=\log\frac{1+\gamma+\sqrt{(1-\gamma)^2+4\gamma^2}}{2},
\end{align}
and for any $\alpha>1$ and $z\in(0,\alpha-1]$, 
\begin{align}\label{counter2}
\liminf_{n\to+\infty}D_{\alpha,z}(\rho_n\|\sigma_n)\ge\log(1+\gamma^2)>0.
\end{align}
\end{prop}
\begin{proof}
For $\ep,\gamma\in(0,1)$, let
\begin{align*}
&\sigma_\eps:=\begin{bmatrix}1&\eps\\\eps&\eps\end{bmatrix}^2
=\begin{bmatrix}1+\eps^2&\eps+\eps^2\\\eps+\eps^2&2\eps^2\end{bmatrix},\qquad
C_{\gamma }:=\begin{bmatrix}1& \gamma \\ \gamma & \gamma \end{bmatrix}, \\
&\rho_{\gamma,\eps}:=\sigma_\eps^{1/2}C_{\gamma }\sigma_\eps^{1/2}
=\begin{bmatrix}1&\eps\\\eps&\eps\end{bmatrix}
\begin{bmatrix}1& \gamma \\ \gamma & \gamma \end{bmatrix}
\begin{bmatrix}1&\eps\\\eps&\eps\end{bmatrix}
=\begin{bmatrix}1+2\gamma\eps+\gamma\eps^2 &
\ep+\gamma\ep+2\gamma\ep^2 \\
\ep+\gamma\ep+2\gamma\ep^2 
&\ep^2(1+3\gamma)
\end{bmatrix}.
\end{align*}
Clearly,
\begin{align*}
\rho_0:=\lim_{\ep\to 0}\rho_{\gamma,\ep}=
\sigma_0:=\lim_{\ep\to 0}\sigma_{\ep}=
\begin{bmatrix}1&0\\0&0\end{bmatrix}.
\end{align*}
Moreover,
\begin{align}\label{Dmax limit}
D_{\max}(\rho_{\gamma,\eps}\|\sigma_{\ep})
=
\log\snorm{\sigma_\eps^{-1/2}\rho_{\gamma,\eps}\sigma_\eps^{-1/2}}_{\infty}
=
\log\norm{C_{\gamma}}_{\infty}
=
\log\bz 1+\gamma+\sqrt{1-2\gamma+5\gamma^2}\jz-\log 2.
\end{align}

By \eqref{ALT1}, $(0,+\infty)\ni z\mapsto Q_{\alpha,z}$ is decreasing, 
so if $z\in(0,\alpha-1]$ then 
\begin{align*}
Q_{\alpha,z}(\rho_{\gamma,\eps}\|\sigma_\eps)
&\ge Q_{\alpha,\alpha-1}(\rho_{\gamma,\eps}\|\sigma_\eps) \\
&=\Tr\Bigl(\sigma_\eps^{-1/2}\rho_{\gamma,\eps}^{\alpha\over\alpha-1}
\sigma_\eps^{-1/2}\Bigr)^{\alpha-1} \\
&=\Tr\Bigl(\sigma_\eps^{-1/2}\rho_{\gamma,\eps}\rho_{\gamma,\eps}^{2-\alpha\over\alpha-1}
\rho_{\gamma,\eps}\sigma_\eps^{-1/2}\Bigr)^{\alpha-1} \\
&=\Tr\Bigl(C_{\gamma}\sigma_\eps^{1/2}\rho_{\gamma,\eps}^{2-\alpha\over\alpha-1}
\sigma_\eps^{1/2}C_{\gamma}\Bigr)^{\alpha-1}.
\end{align*}
When $1<\alpha<2$, 
\[
\Tr\Bigl(C_{\gamma}\sigma_\eps^{1/2}\rho_{\gamma,\eps}^{2-\alpha\over\alpha-1}
\sigma_\eps^{1/2}C_{\gamma}\Bigr)^{\alpha-1}
\xrightarrow[\eps\searrow0]{}\Tr(C_{\gamma}\sigma_0\rho_0\sigma_0C_{\gamma})^{\alpha-1}
=
\Tr(C_{\gamma}\sigma_0C_{\gamma})^{\alpha-1}.
\]
When $\alpha=2$,
\[
\Tr\Bigl(C_{\gamma}\sigma_\eps^{1/2}\rho_{\gamma,\eps}^{2-\alpha\over\alpha-1}
\sigma_\eps^{1/2}C_{\gamma}\Bigr)^{\alpha-1}=\Tr(C_{\gamma}\sigma_\eps C_{\gamma})
\xrightarrow[\eps\searrow0]{}\Tr(C_{\gamma}\sigma_0C_{\gamma}).
\]

Now let $\alpha>2$. Note that 
\[
(1+(2\gamma+1)\eps)I-\rho_{\gamma,\eps}
=\begin{bmatrix}\eps-\gamma\eps^2 & -(1+\gamma)\ep-2\gamma\ep^2\\
 -(1+\gamma)\ep-2\gamma\ep^2 &
 1+(2\gamma+1)\eps-(1+3\gamma)\ep^2\end{bmatrix}\ge0
\]
for any $\eps>0$ sufficiently small. 
Since $\alpha>2$, $\id_{(0,+\infty)}^{2-\alpha\over\alpha-1}$ is 
monotone decreasing, whence
\[
\rho_{\gamma,\eps}^{2-\alpha\over\alpha-1}\ge(1+(2\gamma+1)\eps)^{2-\alpha\over\alpha-1}I.
\]
Using also Lemma \ref{lemma:trace functions}, we get 
\[
\Tr\Bigl(C_{\gamma}\sigma_\eps^{1/2}\rho_{\gamma,\eps}^{2-\alpha\over\alpha-1}
\sigma_\eps^{1/2}C_{\gamma}\Bigr)^{\alpha-1}
\ge
(1+(2\gamma+1)\eps)^{2-\alpha}\Tr(C_{\gamma}\sigma_\eps C_{\gamma})^{\alpha-1}
\xrightarrow[\ep\searrow 0]{}
\Tr(C_{\gamma}\sigma_0C_{\gamma})^{\alpha-1}.
\]

Combining the above, we get that 
for any $\alpha>1$ and $z\in(0,\alpha-1]$,
\begin{align}\label{liminf pos}
\liminf_{\eps\searrow0}Q_{\alpha,z}(\rho_{\gamma,\eps}\|\sigma_\eps)
\ge\Tr(C_{\gamma}\sigma_0C_{\gamma})^{\alpha-1}
=(1+\gamma^2)^{\alpha-1}.
\end{align}

Finally, let $\ep_n\in(0,1)$, $n\in\bN$, be any sequence converging to $0$, and define
\begin{align*}
\rho_n:=\frac{\rho_{\gamma,\ep_n}}{\Tr \rho_{\gamma,\ep_n}},\ds\ds\ds
\sigma_n:=\frac{\sigma_{\ep_n}}{\Tr \sigma_{\ep_n}}.
\end{align*}
Using that $\lim_{\ep\searrow 0}\Tr\rho_{\gamma,\ep}=\lim_{\ep\searrow 0}\Tr\sigma_{\ep}=1$,
the scaling law  \eqref{scaling} yields
\begin{align*}
\lim_{n\to+\infty}D_{\max}(\rho_n\|\sigma_n)=
\lim_{n\to+\infty}D_{\max}(\rho_{\gamma,\ep_n}\|\sigma_{\ep_n}),
\ds\ds\ds
\liminf_{n\to+\infty}D_{\alpha,z}(\rho_n\|\sigma_n)=
\liminf_{n\to+\infty}D_{\alpha,z}(\rho_{\gamma,\ep_n}\|\sigma_{\ep_n}),
\end{align*}
whence
\eqref{counter1} and \eqref{counter2} follow from 
\eqref{Dmax limit} and \eqref{liminf pos}.
\end{proof}

\bibliography{bibliography220802}

\end{document}